\documentclass[11pt]{article}

\usepackage[a4paper, margin=2.54cm]{geometry}
\usepackage[colorlinks, backref]{hyperref}
\usepackage{tikz}
\usepackage[outdir=figure/]{epstopdf}
\usepackage{amsthm, amsmath, amssymb}
\usepackage{bm}
\usepackage{appendix}
\usepackage[numbers]{natbib}
\usepackage{booktabs}
\theoremstyle{definition}
\newtheorem{exmp}{Example}[section]

\newtheorem{proposition}{Proposition}[section]

\usepackage{authblk}
\usepackage{subcaption}
\usepackage[font={small, it}, labelfont=bf]{caption}
\RequirePackage[capitalize,nameinlink]{cleveref}[0.19]

\crefname{section}{Section}{sections}
\crefname{subsection}{Subsection}{subsections}
\Crefname{section}{Section}{Sections}
\Crefname{subsection}{Subsection}{Subsections}
\crefname{exmp}{Example}{Examples}

\Crefname{figure}{Figure}{Figures}

\crefformat{equation}{\textup{#2(#1)#3}}
\crefrangeformat{equation}{\textup{#3(#1)#4--#5(#2)#6}}
\crefmultiformat{equation}{\textup{#2(#1)#3}}{ and \textup{#2(#1)#3}}
{, \textup{#2(#1)#3}}{, and \textup{#2(#1)#3}}
\crefrangemultiformat{equation}{\textup{#3(#1)#4--#5(#2)#6}}%
{ and \textup{#3(#1)#4--#5(#2)#6}}{, \textup{#3(#1)#4--#5(#2)#6}}{, and \textup{#3(#1)#4--#5(#2)#6}}

\Crefformat{equation}{#2Equation~\textup{(#1)}#3}
\Crefrangeformat{equation}{Equations~\textup{#3(#1)#4--#5(#2)#6}}
\Crefmultiformat{equation}{Equations~\textup{#2(#1)#3}}{ and \textup{#2(#1)#3}}
{, \textup{#2(#1)#3}}{, and \textup{#2(#1)#3}}
\Crefrangemultiformat{equation}{Equations~\textup{#3(#1)#4--#5(#2)#6}}%
{ and \textup{#3(#1)#4--#5(#2)#6}}{, \textup{#3(#1)#4--#5(#2)#6}}{, and \textup{#3(#1)#4--#5(#2)#6}}

\title{Limits of accuracy for parameter estimation and localisation in Single-Molecule Microscopy via sequential Monte Carlo methods}
\author[$1$, $2$]{Alix Marie d'Avigneau}
\author[$2$]{Sumeetpal S. Singh}
\author[$1$]{Raimund J. Ober}
\affil[$1$]{\small Center  for  Cancer  Immunology,  Faculty  of  Medicine,  University  of  Southampton,  Southampton,  UK}
\affil[$2$]{\small Signal  Processing  and  Communications  Group,  Department  of  Engineering,  University  of  Cambridge, Cambridge, UK}

\date{}

\usepackage{enumitem}
\newlist{steps}{enumerate}{1}
\setlist[steps, 1]{label = \textit{Step} \arabic*}

\usepackage{algorithm}
\usepackage{algpseudocode}

\usepackage[only,llbracket,rrbracket,llparenthesis,rrparenthesis]{stmaryrd}
\usepackage{bm}
\usepackage{bbm}
\newcommand{\norm}[1]{\left\lVert#1\right\rVert}
\DeclareMathOperator*{\argmax}{argmax} 
\numberwithin{equation}{section}

\usepackage{subcaption}
\usepackage[graphicx]{realboxes}

\begin{document}\sloppy

\maketitle

\begin{abstract}
	 Assessing the quality of parameter estimates for models describing the motion of single molecules in cellular environments is an important problem in fluorescence microscopy. In this work, we consider the fundamental data model, where molecules emit photons at random time instances and these photons arrive at random locations on the detector according to complex point spread functions (PSFs). The randomness and non-Gaussian PSF of the detection process, and the random trajectory of the molecule, makes inference challenging. Moreover, the presence of other closely spaced molecules causes further uncertainty in the origin of the measurements, which impacts the statistical precision of the estimates. We quantify the limits of accuracy of model parameter estimates and separation distance between closely spaced molecules (known as the resolution problem) by computing the Cram\'{e}r-Rao lower bound (CRLB), or equivalently the inverse of the Fisher information matrix (FIM), for the variance of estimates. Results on the CRLB obtained from the fundamental model are crucial, in that they provide a lower bound for more practical scenarios. While analytic expressions for the FIM can be derived for static and deterministically moving molecules, the analytical tools to evaluate the FIM for molecules whose trajectories follow stochastic differential equations (SDEs) are still for the most part missing. 
	We address this by presenting a general sequential Monte Carlo (SMC) based methodology for both parameter inference and computing the desired accuracy limits for non-static molecules and a non-Gaussian fundamental detection model. For the first time, we are able to estimate the FIM for stochastically moving molecules observed through the Airy and Born and Wolf detection models. This is achieved by estimating the score and observed information matrix via SMC. We summarise the outcome of our numerical work by delineating the qualitative behaviours for the accuracy limits as functions of various experimental settings like collected photon count, molecule diffusion, etc. We also verify that we can recover known results from the static molecule case.
\, \par
\emph{Keywords} : Single-molecule microscopy, Fluorescence microscopy, Particle filtering, Particle smoothing, Sequential Monte Carlo (SMC), Fisher information matrix (FIM), Stochastic differential equations (SDEs)
\end{abstract}

\section{Introduction}

\subsection{Motivation}
In recent years, \textit{single-molecule microscopy} has become a powerful tool in cell biology \cite{saxton1997single, saxton1997single2}. It has allowed significant insight to be gained into the behaviour of single molecules in cellular environments using \textit{fluorescence microscopy}. Single-molecule fluorescence microscopy (see \cite{shashkova2017single, moerner2003methods} for reviews) consists of using a suitable fluorophore to label the molecule(s) of interest, exciting said fluorophore with a specific light source and capturing the fluorescence or photons emitted by the molecule(s) through an optical microscope system onto a detector during a fixed acquisition time. 
Many biological applications rely on being able to accurately track moving molecules (or localise them in the static case) and also estimate their model parameters.
Molecule location estimates, which are themselves useful, are also used to estimate the separation distance between two closely spaced molecules \cite{saxton1997single, michalet2003power}, which is needed to quantify the microscopy technique's resolution (discussed below). By model parameters, we mean the drift and diffusion coefficients that describe the motion of randomly moving molecules, but also more generally other parameters for any assumed statistical elements/model for the image acquisition pipeline (see \cref{ex:mol:1}, \cref{ex:smm:fimest}, \cref{sec:mol:exp}). In addition to solving these estimation problems by devising appropriate numerical techniques to compute them, it is also essential to quantify their accuracy, and tools from statistical estimation theory such as the  \textit{Cram\'{e}r-Rao lower bound} (CRLB) \cite{cramer1999mathematical, rao1992information, frechet1943extension, darmois1945limites} are popular in the microscopy literature \cite{ober2004localization, chao2016fisher, ram2006beyond}. Not only is the CRLB able to quantify the accuracy of the estimates, it can also provide the qualitative relationship between estimation accuracy and
various experimental settings, such as the average number of photons captured by the detector, the speed of one or multiple diffusing molecules, or the distance between molecules, which is particularly important in experimental design. For example, one might aim to evaluate how an increase in the speed (or diffusion) of a stochastically diffusing molecule might reduce the accuracy of estimates for its mean location, and whether this loss in accuracy can be mitigated by increasing the mean number of photons captured by the detector.
\subsection{Methods for assessing the accuracy of parameter estimates}
In the past, in the context of the resolution problem,  \textit{Rayleigh's criterion} \cite{born2013principles} has been used to define the minimum distance between two point sources such that they can be distinguished in the image. However, a drawback of employing Rayleigh's criterion is that it ignores the statistical aspect of the separation distance estimation problem. For example, it doesn't account for the fact that each new observation (taking the form of a captured emitted photon) brings new information on the separation distance. In contrast, in estimation theory, the CRLB establishes a lower bound on the variance of unbiased estimates, and is therefore often used as a benchmark for the quality of a given estimator. As a result, the CRLB plays an important part in experimental design for single-molecule microscopy \cite{ober2004localization, ram2006stochastic}. For example, in \cite{ram2006beyond, ram2013stochastic}, the authors present an improved microscope resolution measure in the form of the square root of the CRLB for the separation distance between two molecules, which is referred to as the \textit{limit of accuracy} with which the separation distance between the two objects can be estimated based on the observed data. A particular advantage of this new resolution measure is that it predicts that increasing the photon count makes it possible to estimate a separation distance between two molecules that is shorter than Rayleigh's criterion. In the context of localisation and estimation of parameters for models describing the motion of a single molecule, we also quantify the limits of accuracy for these model parameter estimates by computing the CRLB. \par 
Evaluating these limits of accuracy is a challenging task. In this paper, we consider the \textit{fundamental data model} \cite{ober2004localization, ram2006stochastic}, which is crucial in that it provides more easily computed lower bounds for the limits of accuracy of more realistic practical models, where factors such as pixelisation and readout noise come into play and make inference more challenging \cite{vahid2020effect}. Indeed, the limits of accuracy derived for the fundamental model are often known as the \textit{fundamental limits of accuracy}. In this model, the detection process of the emitted fluorescence already presents its own challenges, as it is intrinsically random both in time and location. While many methods \cite{calderon2015inferring, calderon2016motion, calderon2013quantifying} have assumed that the arrival times of the photons on the detector were uniformly distributed, \cite{ober2004localization, ram2006stochastic} suggest that the arrival times of photons follow a Poisson process. As for the arrival location of these photons on the detector, a wide range of measurement models exist $-$ corresponding to the various types of detector. The typical measurement model used for an in-focus source is the \textit{Airy profile} \cite{vahid2020fisher, chao2016fisher}. If the molecule is out of focus, 3D models are generally used instead, such as the \textit{Born and Wolf model} \cite{born2013principles}. Often, these models make parameter inference difficult, and researchers have often opted for a Gaussian approximation to these models, such as in \cite{berglund2010statistics, relich2016estimation, michalet2012optimal}. However, \cite{vahid2020fisher} argue that in practice, assuming Gaussian distributed photon locations on the detector is not an accurate approximation of the underlying model. \par
While it is important to be able to accurately study the behaviours and interactions of single molecules within a cell, it is especially challenging when those molecules have stochastic trajectories.
The motion of an object in a cellular environment is affected by a multitude of deterministic, as well as random factors \cite{briane2018statistical}, and in many applications \cite{vahid2020fisher, calderon2016motion}, the trajectories of single molecules are modelled by \textit{stochastic differential equations} (SDEs) \cite{oksendal2013stochastic}.
The CRLB is obtained by taking the inverse of the \textit{Fisher information matrix} (FIM), and analytical expressions for the FIM, and thus the limit of accuracy (given by the square root of the CRLB) for the location of an in-focus \textit{static} (or unmoving) molecule have been derived in \cite{ober2004localization, chao2016fisher}. Similar results for an out-of-focus static molecule are available in \cite{ober2020quantitativech19}, and analytical expressions have also been derived in the context of molecules with deterministic linear or circular trajectories in \cite{wong2010limit}. As for the resolution problem, it is addressed in \cite{ram2006beyond, ram2006stochastic} in a static molecule context and in \cite{lin2015limit} for two dynamic molecules with deterministic trajectories. However, when molecules have stochastic trajectories, the analytical tools to obtain the CRLB and tackle many of these problems are still for the most part missing. In this paper, we propose a numerical approach to address these problems. \par
In the context of stochastically moving molecules, \cite{vahid2020fisher} developed a method to obtain the FIM for a molecule whose trajectory is described by a linear SDE. For a 2D Gaussian approximation of the photon detection process, the authors take advantage of the Kalman filter formulae to obtain an analytical form  for the FIM for a specific set of photon detection times. However, if the Airy profile is used instead, the computational cost of performing numerical integration becomes prohibitive for more than a single photon. Among other things, we build on \cite{vahid2020fisher} and provide effective methodological advances which enable the estimation of the FIM for the hyperparameters of models with Airy and Born and Wolf distributed photon locations.\par 
\subsection{Contributions}
In this paper, we develop an effective and general numerical framework to obtain \textit{sequential Monte Carlo} (SMC) approximations of expectations of interest, including for stochastically moving molecules. The ability to approximate these expectations is important for estimating the \textit{score} and \textit{observed information matrix} (OIM) for the hyperparameters of interest, and can also be employed to obtain maximum likelihood (ML) estimates of said hyperparameters. Access to the score and/or OIM is vital in order to be able to estimate the FIM. To achieve this, the observation interval is first discretised and the problem reformulated as a discrete-time state space model, which takes into account the random arrival times of photons on the detector in the form of missing observations. Then, a particle filter is employed in conjunction with \textit{forward smoothing} methods \cite{del2010forward, olsson2017efficient} to obtain particle approximations of the expectations of interest. Our work complements \cite{ashley2015method}, in which the authors similarly employed time discretisation of the observation interval, but they did not attempt to estimate the CRLB for hyperparameters.  
With our approach, we are for the first time able to obtain the limits of accuracy for parameters of a single molecule whose trajectory follows an SDE, thus providing new insights beyond existing results for molecules that are static or following a deterministic trajectory. Our SMC-based methodology is also more general than the Kalman filter-based approach of \cite{vahid2020fisher}, and has no systemic limitations (i.e. variance in estimates of the limits of accuracy can always be reduced by increasing the number of Monte Carlo samples). We are also able to generalise results for the optical microscope resolution problem from considering the separation distance between two static molecules to that between two stochastically diffusing molecules.\par
The numerical experiments in this paper consist first of applying the methodology to estimate the limit of accuracy for a single stochastically moving molecule with 2D Gaussian, Airy, and Born and Wolf photon detection models by using estimates of the score and OIM obtained by forward smoothing. This is repeated for various expected mean photon counts to verify that for molecules with stochastic trajectories, the limit of accuracy exhibits an inverse square root decay with respect to mean photon count, i.e. the uncertainty of the hyperparameter estimates decreases as the expected number of photons increases. This has already been proven for static molecules \cite{ober2020quantitativech18, ober2020quantitativech19, ram2006beyond}. The methodology is also applied in the context of the optical microscope resolution problem to obtain estimates of the limit of accuracy for the mean separation distance between two closely spaced diffusing molecules. Thanks to our numerical approach, insights can be obtained into the generalisation to diffusing molecules of results proven in \cite{ram2013stochastic} on this resolution problem for two static molecules. For instance, in \cite{ober2004localization}, it was shown that the limit of accuracy for the location of a static molecule has a linear relationship with the standard deviation of the photon detection profile. From our numerical results, we show that when molecules are diffusing, the appropriate relationship behaves qualitatively with 
the diffusion coefficient standard deviation in a similar way, i.e. it can be translated into additional observation uncertainty. The qualitative relationships observed through our numerical experiments for stochastically moving molecules are summarised in \cref{tab:smm:dynamic}. \par
This paper is structured as follows. In \cref{sec:mol:model}, the model is presented, including the molecule trajectory, described by a stochastic differential equation (SDE), and the photon detection time and location processes. In \cref{sec:mol:hmm}, the model is formulated as a discrete-time state space model with a discretised observation interval. Then, \cref{sec:mol:parinf} establishes the main parameter inference aims and methods, which consist of particle filtering and smoothing of additive functionals in order to estimate the score and OIM for hyperparameters, and methods to estimate the FIM from the score and OIM. Numerical experiments are run in \cref{sec:mol:exp} to first estimate the limit of accuracy for the drift and diffusion coefficients of the SDE for all photon detection profiles and then estimate the limit of accuracy for the separation distance between two dynamic molecules. Finally, \cref{sec:smm:conc} provides concluding remarks.
\section{Model specification}
\label{sec:mol:model}
For the purpose of this paper, a basic optical system is considered, also known in \cite{vahid2020fisher, chao2016fisher} as the \textit{fundamental data model}. See \cref{fig:opt} for an overview of the optical system. Under the fundamental model, we assume that the photons are observed under ideal conditions, in which the detector $\mathcal{Y}=\mathbb{R}^2$ is non-pixelated. This model does not describe image data obtained from actual microscopy experiments the way more realistic, or \textit{practical} models do. However, the fundamental model is crucial, in that it offers an obtainable lower bound to the CRLB of parameters of the more realistic practical model, which is much more difficult to obtain. In this section, the various aspects of the model are described. These include the true molecule trajectory, occurring in the object space, the photon detection locations in the image space, and the times at which photons arrive on the detector.

\tikzset{every picture/.style={line width=0.75pt}} 
\begin{figure}[htbp!]
	\centering
	\begin{tikzpicture}[x=0.65pt,y=0.55pt,yscale=-1,xscale=1]
		
		\draw  [color={rgb, 255:red, 128; green, 128; blue, 128 }  ,draw opacity=1 ][fill={rgb, 255:red, 230; green, 230; blue, 230 }  ,fill opacity=1 ] (145.67,42.94) -- (145.67,216.6) -- (69.91,252.28) -- (69.91,78.63) -- cycle ;
		\draw  [color={rgb, 255:red, 128; green, 128; blue, 128 }  ,draw opacity=1 ][fill={rgb, 255:red, 126; green, 211; blue, 33 }  ,fill opacity=1 ] (95.8,143.21) .. controls (95.8,125.69) and (109.64,126.39) .. (112.41,143.21) .. controls (115.17,160.04) and (122.55,152.33) .. (120.71,160.04) .. controls (118.86,167.75) and (111.48,176.16) .. (99.49,166.35) .. controls (87.5,156.53) and (95.8,160.74) .. (95.8,143.21) -- cycle ;
		\draw  [fill={rgb, 255:red, 218; green, 217; blue, 217 }  ,fill opacity=1 ] (536.67,41.94) -- (536.67,215.6) -- (460.91,251.28) -- (460.91,77.63) -- cycle ;
		\draw    (310.33,13.83) -- (312.67,289) ;
		\draw  [draw opacity=0] (311.68,194.97) .. controls (299.35,181.23) and (291.67,161.79) .. (291.67,140.25) .. controls (291.67,118.97) and (299.17,99.74) .. (311.23,86.03) -- (355.83,140.25) -- cycle ; \draw   (311.68,194.97) .. controls (299.35,181.23) and (291.67,161.79) .. (291.67,140.25) .. controls (291.67,118.97) and (299.17,99.74) .. (311.23,86.03) ;
		\draw  [draw opacity=0] (311.84,86.23) .. controls (324.25,99.89) and (332,119.27) .. (332,140.76) .. controls (331.99,161.99) and (324.43,181.16) .. (312.27,194.79) -- (267.83,140.75) -- cycle ; \draw   (311.84,86.23) .. controls (324.25,99.89) and (332,119.27) .. (332,140.76) .. controls (331.99,161.99) and (324.43,181.16) .. (312.27,194.79) ;
		\draw    (228.33,20) -- (307.04,20) ;
		\draw [shift={(310.04,20)}, rotate = 180] [fill={rgb, 255:red, 0; green, 0; blue, 0 }  ][line width=0.08]  [draw opacity=0] (10.72,-5.15) -- (0,0) -- (10.72,5.15) -- (7.12,0) -- cycle    ;
		\draw    (24,20) -- (102.71,20) ;
		\draw [shift={(21,20)}, rotate = 0] [fill={rgb, 255:red, 0; green, 0; blue, 0 }  ][line width=0.08]  [draw opacity=0] (10.72,-5.15) -- (0,0) -- (10.72,5.15) -- (7.12,0) -- cycle    ;
		\draw    (519,20) -- (597.71,20) ;
		\draw [shift={(600.71,20)}, rotate = 180] [fill={rgb, 255:red, 0; green, 0; blue, 0 }  ][line width=0.08]  [draw opacity=0] (10.72,-5.15) -- (0,0) -- (10.72,5.15) -- (7.12,0) -- cycle    ;
		\draw    (314,20) -- (392.71,20) ;
		\draw [shift={(311,20)}, rotate = 0] [fill={rgb, 255:red, 0; green, 0; blue, 0 }  ][line width=0.08]  [draw opacity=0] (10.72,-5.15) -- (0,0) -- (10.72,5.15) -- (7.12,0) -- cycle    ;
		\draw  [color={rgb, 255:red, 0; green, 0; blue, 0 }  ,draw opacity=1 ][fill={rgb, 255:red, 65; green, 117; blue, 5 }  ,fill opacity=1 ] (476.6,120.97) .. controls (486.71,75.56) and (506.25,68.42) .. (519,113) .. controls (531.75,157.58) and (531.45,142.4) .. (531.19,166.33) .. controls (530.93,190.25) and (510.97,209.79) .. (484.69,183.33) .. controls (458.41,156.88) and (466.49,166.39) .. (476.6,120.97) -- cycle ;
		\draw [color={rgb, 255:red, 0; green, 0; blue, 0 }  ,draw opacity=1 ][fill={rgb, 255:red, 255; green, 255; blue, 255 }  ,fill opacity=1 ]   (152.33,105) -- (114.6,137.7) ;
		\draw [shift={(112.33,139.67)}, rotate = 319.09000000000003] [fill={rgb, 255:red, 0; green, 0; blue, 0 }  ,fill opacity=1 ][line width=0.08]  [draw opacity=0] (8.93,-4.29) -- (0,0) -- (8.93,4.29) -- cycle    ;
		\draw    (21,150) -- (598.67,150) ;
		\draw [color={rgb, 255:red, 0; green, 0; blue, 0 }  ,draw opacity=1 ][fill={rgb, 255:red, 255; green, 255; blue, 255 }  ,fill opacity=1 ]   (557,61) -- (512.2,89.39) ;
		\draw [shift={(509.67,91)}, rotate = 327.63] [fill={rgb, 255:red, 0; green, 0; blue, 0 }  ,fill opacity=1 ][line width=0.08]  [draw opacity=0] (8.93,-4.29) -- (0,0) -- (8.93,4.29) -- cycle    ;
		\draw  (109.24,158.95) -- (109.24,67.84)(32.33,184.71) -- (117.79,145.96) (104.24,77.11) -- (109.24,67.84) -- (114.24,72.57) (39.33,186.53) -- (32.33,184.71) -- (39.33,176.53)  ;
		\draw  (499.24,158.61) -- (499.24,67.51)(422.33,184.37) -- (507.79,145.63) (494.24,76.77) -- (499.24,67.51) -- (504.24,72.24) (429.33,186.2) -- (422.33,184.37) -- (429.33,176.2)  ;
		\draw  [fill={rgb, 255:red, 0; green, 0; blue, 0 }  ,fill opacity=1 ] (105.2,163.27) .. controls (105.2,163.05) and (105.38,162.87) .. (105.6,162.87) .. controls (105.82,162.87) and (106,163.05) .. (106,163.27) .. controls (106,163.49) and (105.82,163.67) .. (105.6,163.67) .. controls (105.38,163.67) and (105.2,163.49) .. (105.2,163.27) -- cycle ;
		\draw    (150.33,233) -- (108.76,169.34) ;
		\draw [shift={(107.67,167.67)}, rotate = 416.85] [color={rgb, 255:red, 0; green, 0; blue, 0 }  ][line width=0.75]    (4.37,-1.32) .. controls (2.78,-0.56) and (1.32,-0.12) .. (0,0) .. controls (1.32,0.12) and (2.78,0.56) .. (4.37,1.32)   ;
		\draw  [fill={rgb, 255:red, 0; green, 0; blue, 0 }  ,fill opacity=1 ] (504.6,177.5) .. controls (504.6,176.23) and (503.57,175.2) .. (502.3,175.2) .. controls (501.03,175.2) and (500,176.23) .. (500,177.5) .. controls (500,178.77) and (501.03,179.8) .. (502.3,179.8) .. controls (503.57,179.8) and (504.6,178.77) .. (504.6,177.5) -- cycle ;
		\draw    (424.33,230.33) -- (500.62,180.89) ;
		\draw [shift={(502.3,179.8)}, rotate = 507.05] [color={rgb, 255:red, 0; green, 0; blue, 0 }  ][line width=0.75]    (4.37,-1.32) .. controls (2.78,-0.56) and (1.32,-0.12) .. (0,0) .. controls (1.32,0.12) and (2.78,0.56) .. (4.37,1.32)   ;
		\draw [color={rgb, 255:red, 128; green, 128; blue, 128 }  ,draw opacity=1 ] [dash pattern={on 0.84pt off 2.51pt}]  (102,136.33) -- (301.39,97.57) ;
		\draw [shift={(304.33,97)}, rotate = 529] [fill={rgb, 255:red, 128; green, 128; blue, 128 }  ,fill opacity=1 ][line width=0.08]  [draw opacity=0] (5.36,-2.57) -- (0,0) -- (5.36,2.57) -- cycle    ;
		\draw [color={rgb, 255:red, 128; green, 128; blue, 128 }  ,draw opacity=1 ] [dash pattern={on 0.84pt off 2.51pt}]  (107.67,167.67) -- (303.35,188.68) ;
		\draw [shift={(306.33,189)}, rotate = 186.13] [fill={rgb, 255:red, 128; green, 128; blue, 128 }  ,fill opacity=1 ][line width=0.08]  [draw opacity=0] (5.36,-2.57) -- (0,0) -- (5.36,2.57) -- cycle    ;
		\draw [color={rgb, 255:red, 128; green, 128; blue, 128 }  ,draw opacity=1 ] [dash pattern={on 0.84pt off 2.51pt}]  (320,97) -- (489.34,111.41) ;
		\draw [shift={(492.33,111.67)}, rotate = 184.86] [fill={rgb, 255:red, 128; green, 128; blue, 128 }  ,fill opacity=1 ][line width=0.08]  [draw opacity=0] (5.36,-2.57) -- (0,0) -- (5.36,2.57) -- cycle    ;
		\draw [color={rgb, 255:red, 128; green, 128; blue, 128 }  ,draw opacity=1 ] [dash pattern={on 0.84pt off 2.51pt}]  (319,186.33) -- (493.36,164.05) ;
		\draw [shift={(496.33,163.67)}, rotate = 532.72] [fill={rgb, 255:red, 128; green, 128; blue, 128 }  ,fill opacity=1 ][line width=0.08]  [draw opacity=0] (5.36,-2.57) -- (0,0) -- (5.36,2.57) -- cycle    ;
		
		\draw (320.33,71.5) node [anchor=north west][inner sep=0.75pt]  [font=\small] [align=left] {\begin{minipage}[lt]{51.17pt}\setlength\topsep{0pt}
				\begin{center}
					lens system
				\end{center}	
		\end{minipage}};
		\draw (121,11.83) node [anchor=north west][inner sep=0.75pt]   [align=left] {object space};
		\draw (414.67,12.5) node [anchor=north west][inner sep=0.75pt]   [align=left] {image space};
		\draw (92.02,220.47) node  [font=\footnotesize,xslant=0.08] [align=left] {\begin{minipage}[lt]{26.746644000000003pt}\setlength\topsep{0pt}
				\begin{center}
					object \\plane
				\end{center}	
		\end{minipage}};
		\draw (152,88.67) node [anchor=north west][inner sep=0.75pt]   [align=left] {object};
		\draw (481.69,219.47) node  [font=\footnotesize,xslant=0.08] [align=left] {\begin{minipage}[lt]{27.200000000000003pt}\setlength\topsep{0pt}
				\begin{center}
					image \\plane
				\end{center}	
		\end{minipage}};
		\draw (560,49.67) node [anchor=north west][inner sep=0.75pt]   [align=left] {detector};
		\draw (37.33,184.07) node [anchor=north west][inner sep=0.75pt]    {$e_{o,1}$};
		\draw (114.67,64.73) node [anchor=north west][inner sep=0.75pt]    {$e_{o,2}$};
		\draw (555.67,129.67) node [anchor=north west][inner sep=0.75pt]   [align=left] {\begin{minipage}[lt]{35.054pt}\setlength\topsep{0pt}
				\begin{center}
					optical \\axis
				\end{center}	
		\end{minipage}};
		\draw (428.67,182.07) node [anchor=north west][inner sep=0.75pt]    {$e_{i,1}$};
		\draw (504.67,56.73) node [anchor=north west][inner sep=0.75pt]    {$e_{i,2}$};
		\draw (147.33,231.57) node [anchor=north west][inner sep=0.75pt]    {$X( t )$};
		\draw (124.33,254) node [anchor=north west][inner sep=0.75pt]  [font=\footnotesize] [align=left] {\begin{minipage}[lt]{57.120000000000005pt}\setlength\topsep{0pt}
				\begin{center}
					object location at time $\displaystyle t $ 
				\end{center}	
		\end{minipage}};
		\draw (400,230.23) node [anchor=north west][inner sep=0.75pt]    {$Y( t )$};
		\draw (350.33,249.33) node [anchor=north west][inner sep=0.75pt]  [font=\footnotesize] [align=left] {\begin{minipage}[lt]{75.70644pt}\setlength\topsep{0pt}
				\begin{center}
					location of detected photon at time $\displaystyle t $ 
				\end{center}	
		\end{minipage}};
	\end{tikzpicture}
	\caption[Illustration of an optical microscope.]{Illustration of an optical microscope. At time $t\geq t_0$, the molecule is located at $X(t)$ in the object space and might be moving along the object plane. If the molecule is out of focus, it will instead move along a plane parallel to the object plane but displaced along the optical axis. The molecule emits photons through the lens system into the image space and its image is acquired on the planar detector $\mathcal{Y}$ located on the image plane. The location of the detected photons at time $t$ is denoted by $Y(t)$.} \label{fig:opt}
\end{figure}
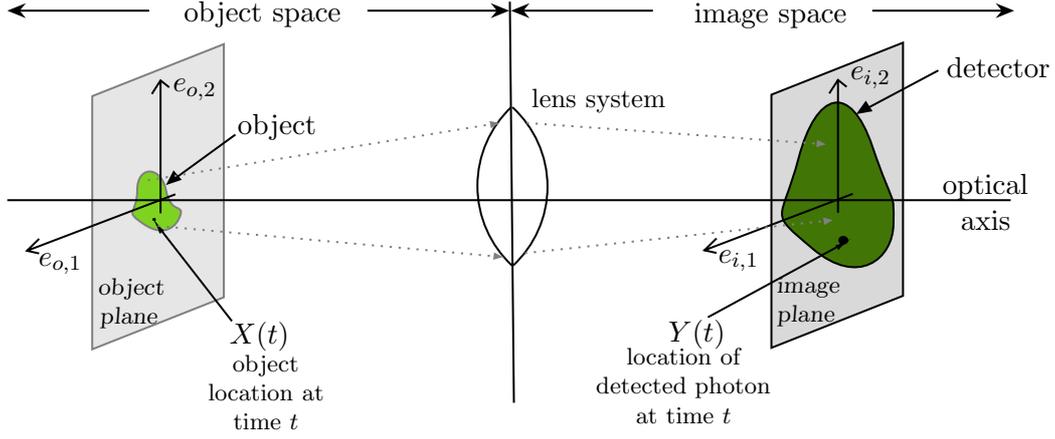

\subsection{Molecule trajectory}
\label{subsec:x}
For notational simplicity, let $X_t := X(t)\in \mathbb{R}^d$ denote the true, $d$-dimensional location of the molecule at time $t$. Given hyperparameters $\theta$, let $f_{s, t}^{\theta}(x_t | x_s)$ denote the probability density function of $X_t$ given the previous location $X_s$. Assume that the molecule trajectory $(X_t)_{t_0\leq t\leq T}$ follows a linear \textit{stochastic differential equation} (SDE)
\begin{equation}
	\label{eq:sde}
	\text{d}X_t = b(t, X_t)\text{d}t+\sigma(t, X_t)\text{d}B_t,
\end{equation}
where $b(t, X_t):=b_0+b_1(t)X_t$ and $\sigma(t, X_t) := \sigma(t)$ represent the \textit{drift} and \textit{diffusion} coefficients, respectively, $b_0$ is the zero order drift coefficient, and $(\text{d}B_t)_{t_0\leq t\leq T}$ is a Wiener process with $\mathbb{E}\left[\text{d}B_t \text{d}B_t^{\intercal}\right]=\mathbb{I}_{d\times d}$. According to \cite{jazwinski2007stochastic, evans2012introduction} the solution to the SDE in \eqref{eq:sde} at discrete time points $t_0 < t_1 < \ldots$ is given by
\begin{equation}
	\label{eq:x:sol}
	X_{t_{i+1}} = \Phi(t_i, t_{i+1})X_{t_i} + a(t_i, t_{i+1}) + W_g(t_i, t_{i+1}),
\end{equation}
where the \textit{fundamental matrix function} $\Phi \in \mathbb{R}^{d\times d}$ satisfies the following for all $s, t, u \geq t_0$
\begin{align}
	\label{eq:phi}
	\frac{\text{d}\Phi(s, t)}{\text{d}t}&=b_1(t)\Phi(s, t), \\
	\Phi(t, t) = \mathbb{I}_{d\times d},&\qquad \Phi(s, t)\Phi(t, u) = \Phi(s, u) \nonumber,
\end{align}
the vector $a(t_i, t_{i+1}) \in \mathbb{R}^d$ is given by
\begin{equation*}
	a(t_i, t_{i+1}) = \int_{t_i}^{t_{i+1}}b_0\Phi(t_i, t)\text{d}t,
\end{equation*}
and finally the process $\left(W_g(t_i, t_{i+1}) = \int_{t_i}^{t_{i+1}}\Phi(t_i, t)\sigma(t)\text{d}B_t \right)_{i=1}^{\infty}$ is a white noise sequence with mean zero and covariance
\begin{equation}
	\label{eq:r}
	R(t_i, t_{i+1}) = \int_{t_{i}}^{t_{i+1}}\Phi(t_i, t)\sigma(t)\sigma^{\intercal}(t)\Phi^{\intercal}(t_i, t)\text{d}t.
\end{equation}
Therefore, the transition density $f_{t_{i+1},t_{i}}^{\theta}(x'|x)$ can be expressed as a Gaussian with mean $\mu(x, t_i, t_{i+1}) = \Phi(t_i, t_{i+1})x+a(t_i, t_{i+1})$ and covariance $R(t_i, t_{i+1})$:
\begin{equation}
	\label{eq:x}
	X_{t_{i+1}} | (X_{t_i}=x) \sim \mathcal{N}\left(\mu(x, t_i, t_{i+1}), R(t_i, t_{i+1})\right).
\end{equation}

\begin{exmp}
	\label{ex:mol:1}
	Let the trajectory of a molecule be given by the following SDE
	\begin{equation} \label{eq:sde:ex1}
		\text{d}X_t = b\mathbb{I}_{d\times d}X_t\text{d}t + \sqrt{2}\sigma\text{d}B_t,
	\end{equation}
	where in the drift term $b \in \mathbb{R}$, in the diffusion term $\sigma>0$, and $(\text{d}B_t)_{t_0\leq t\leq T}$ is a Wiener process and let $\theta = (\sigma^2, b)$. Assuming the time points $t_0, t_1, \ldots$ are equidistant, i.e. $t_{i+1} - t_{i} = \Delta$ for all $i = 0, 1, \ldots$, let the fundamental matrix $\Phi_{\Delta} := \varphi^{\theta}_{\Delta}\mathbb{I}_{d\times d}$ where $\varphi^{\theta}_{\Delta} \in \mathbb{R}$ and the covariance matrix $R_{\Delta} := r^{\theta}_{\Delta}\mathbb{I}_{d\times d}$ where $r^{\theta}_{\Delta}>0$. Then, by solving \eqref{eq:phi} and plugging the result into \eqref{eq:r}, we obtain 
	\begin{equation*}
		\varphi^{\theta}_{\Delta} = 
		\begin{cases}
			e^{\Delta b}, & \mathrm{if}\;b\neq0,\\
			1, & \mathrm{if}\;b=0,
		\end{cases} 
		\qquad\text{and}\qquad r^{\theta}_{\Delta} = 
		\begin{cases}
			\frac{\sigma^2}{b}\left(e^{2\Delta b}-1\right), & \mathrm{if}\;b\neq0,\\
			2\sigma^2\Delta, & \mathrm{if}\;b=0.
		\end{cases}
	\end{equation*}
	The initial distribution $X_{t_0}\sim\mathcal{N}(x_{0}, P_0)$ has covariance matrix $P_0 = p_0^2\mathbb{I}_{d\times d}$ where $p_0 \in \mathbb{R}$. \par 
	In a 2D setting (i.e. $d=2$), let the drift $b=-10$ s$^{-1}$, the diffusion $\sigma^2 = 1$ $\mu$m$^2/$s and the initial covariance $p_0^2 = 10^{-2}$ $\mu$m$^2$ and mean $x_0 = (4.4, 4.4)^\intercal$ $\mu$m. Note that for the purpose of this example, the initial covariance matrix is diagonal, but there is no restriction to employing a more general, non-diagonal initial covariance matrix. By simulating the molecule trajectory for the time interval $[0, 0.1]$ seconds, we obtain the trajectory in \cref{fig:mol:true}.
\end{exmp}
\begin{figure*}[htbp!]
	\centering
	\includegraphics[width=.35\textwidth]{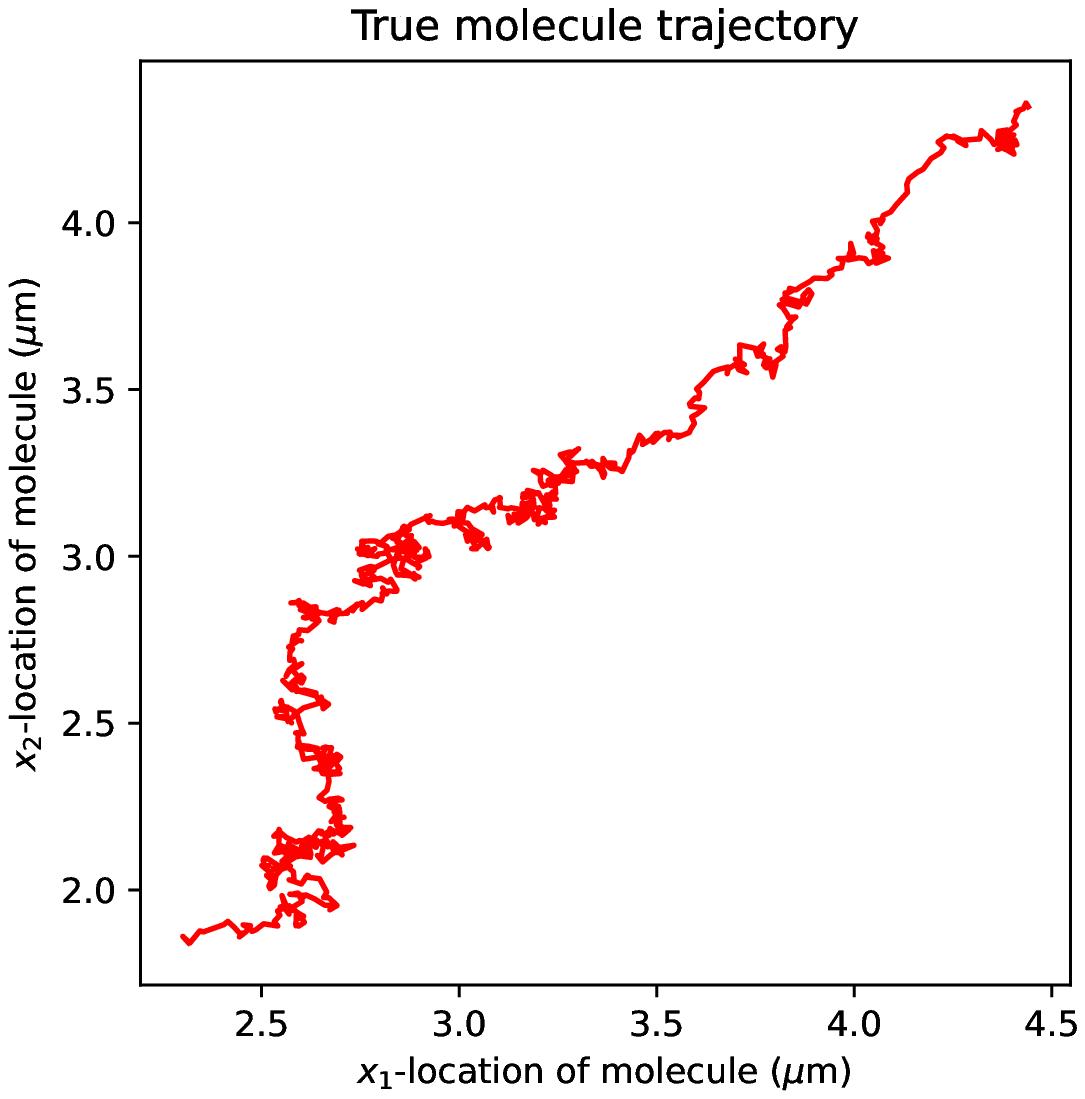}
	\caption[Trajectory of a molecule in the object space with stochastic trajectory.]{Trajectory of a molecule in the object space with stochastic trajectory described in \eqref{eq:sde:ex1} and with diffusion and drift coefficients $\sigma^2 = 1$ $\mu$m$^2/$s and $b=-10$ s$^{-1}$, respectively. The molecule moves during an interval of $[0, 0.1]$ seconds and its initial location is Gaussian distributed with mean $x_0 = (4.4, 4.4)^\intercal$ $\mu$m and covariance $P_0 = 10^{-2}\mathbb{I}_{2\times 2}$ $\mu$m$^2$.}
	\label{fig:mol:true}
\end{figure*}
\subsection{Photon detection locations}
The true molecule trajectory cannot be observed directly. Instead, a fluorescence microscope is used: the molecule of interest is labelled using a suitable fluorophore, magnified through a lens system and the photons it emits arrive on a detector $\mathcal{Y}:=\mathbb{R}^2$ for a fixed time period (see \cref{fig:opt}). The arrival location of a photon on the detector is random, and using the typical approximation of the optical microscope from \cite{goodman2005introduction}, it can be described as follows. Let $Y \in \mathcal{Y}$ denote the observed location of a detected photon. For an object located at $(x_{0, 1}, x_{0, 2}, z_{0}) \in \mathbb{R}^3$ in the object space, its \textit{photon distribution profile} \cite{ram2006stochastic} is given by the density
\begin{equation}
	\label{eq:image}
	g_{\theta}(y\vert x) := \frac{1}{\vert M \vert}q_{z_{0}}\left(M^{-1}y-(x_{0,1}, x_{0,2})^{\intercal}\right), \quad y \in \mathbb{R}^2,
\end{equation}
where $M \in \mathbb{R}^{2\times2}$ is an invertible \textit{lateral magnification matrix} and the \textit{image function} $q_{z_{0}}:\mathbb{R}^2 \rightarrow \mathbb{R}$ describes the image of an object in the detector space when that object is located at $(0, 0, z_{0})$ in the object space. Note that the subscript $\theta$ is used in the left-hand side of \eqref{eq:image} to include dependence on hyperparameters. Depending on the model considered and inference aims, the hyperparameter(s) of interest can be $(x_{0, 1}, x_{0, 2})$ if the object is static and/or $z_{0}$ if an out-of-focus molecule is considered. \par
Three types of image functions are considered. 
First of all, according to optical diffraction theory from \cite{born2013principles}, an in-focus point source (i.e. when $z_{0}=0$) will typically generate an image that follows the Airy profile, given by
\begin{equation}
	\label{eq:airy}
	q(x_{1}, x_{2}) = \frac{J_1^2\left(\frac{2\pi n_{\alpha}}{\lambda_e}\sqrt{x_1^2+x_2^2}\right)}{\pi(x_1^2+x_2^2)}, \quad (x_1, x_2)\in \mathbb{R}^{2},
\end{equation}
where $n_{\alpha}$ is the numerical aperture of the objective lens, $\lambda_e$ is the emission wavelength of the molecule and $J_1(\cdot)$ represents the first order Bessel function of the first kind.\par
Often, to simplify the problem, the 2D Gaussian approximation to the Airy profile has been used instead (see \cite{cheezum2001quantitative, thompson2002precise, zhang2007gaussian, stallinga2010accuracy}):
\begin{equation}
	\label{eq:gausapp}
	q(x_1, x_2) = \frac{1}{2\pi\sigma_a^2} \exp \left[-\frac{x_1^2+x_2^2}{2\sigma_a^2}\right], \quad (x_1, x_2)\in \mathbb{R}^{2}.
\end{equation}
If the point source of interest is out of focus, then a 3D Born and Wolf model  \cite{born2013principles} is used instead:
\begin{equation}
	\label{eq:bw}
	q_{z_{0}}(x_1, x_2) = \frac{4\pi n^2_{\alpha}}{\lambda_e^2}\left\lvert\int_{0}^1J_0\left(\frac{2\pi n_{\alpha}}{\lambda_e}\sqrt{x_1^2+x_2^2}\rho\right)\exp{\left(\frac{j\pi n^2_{\alpha}z_{0}}{n_o\lambda_e}\rho^2\right)}\rho d\rho\right\rvert^2, \quad (x_1, x_2)\in \mathbb{R}^{2},
\end{equation}
where $z_{0} \in \mathbb{R}$ is the location of the object on the optical axis, $n_o$ is the refractive index of the objective lens immersion medium and $J_0(\cdot)$ is the zero-th order Bessel function of the first kind. Note that the Airy profile is a special case of the Born and Wolf model. Indeed, if the object is in focus, then $z_{0}=0$ on the optical axis and \eqref{eq:airy} and \eqref{eq:bw} coincide.
\subsection{Photon detection times}
Just like the photon detection locations, the times at which the photons arrive on the detector $\mathcal{Y}$ are random. More specifically, in \cite{vahid2020fisher, chao2016fisher}, the arrival of the photons on the detector, or \textit{photon detection process}, can be modelled as a Poisson process. Let $N(t)$ be the number of photons detected at time $t \geq t_0$ for initial time $t_0 \in \mathbb{R}$ and let $\lambda(t)$ be the \textit{photon detection rate}, representing the rate at which the photons emitted by the object hit the detector at any given time $t$.
For example, the detection rate of an object that has high photostability will simply be constant, while an exponentially decaying $\lambda(t)$ can indicate that the object image is photobleaching, or fading over time. The arrival times of the photons on the detector $\mathcal{Y}$ are denoted $t_1, t_2, \ldots$ where $t_i$ denotes the arrival time of the $i$-th photon.
\subsection{The observed data}
Let $n_p = N(T)-N(t_0)$ be the number of photons detected in the interval $[t_0, T]$. We have now established the two aspects of the data that can be observed in a basic optical system during this interval, namely the detection times $t_1, t_2, \ldots, t_{n_p}$ of photons and the location of those detected photons $Y_{t_1}, Y_{t_2}, \ldots, Y_{t_{n_p}}$ on the detector $\mathcal{Y}$. Assume that, conditionally on the current object location $X_{t_i}$, the location of the $i$-th detected photon $Y_{t_i}$ at time $t_i$ is independent of the previous locations and time points of the detected photons, i.e. for $x_{t_i} \in \mathcal{X}$,
\begin{equation}
	\label{eq:assump:phot}
	p_{\theta}(y_{t_i}|x_{t_i}, y_{t_{i-1}}, \ldots, y_{t_0}) = p_{\theta}(y_{t_i}|x_{t_i}) =: g_{\theta}(y_{t_i}|x_{t_i}), \quad y_{t_i} \in \mathcal{Y},
\end{equation} 
where the density $g_{\theta}$ is the photon distribution profile from \eqref{eq:image}. This is a reasonable assumption, as at any given time, processes such as photon emission and image formation only depend on the state of the emitting fluorescent molecule at that time, and not on any prior event. 
\begin{exmp} \label{ex:smm:obs}
	Let the trajectory of a molecule be given by the SDE in \cref{ex:mol:1} and simulated using the same parameters and for the same time interval. Let $\mathcal{Y}$ be a non-pixelated detector. Then, let the photon detection rate be constant such that the mean number of photons is $500$, and the photon distribution profile be given by \eqref{eq:image}, where the magnification matrix $M = m\mathbb{I}_{2\times2}$ with $m=100$. The image functions for the Airy, 2D Gaussian and Born and Wolf profiles are given by \eqref{eq:airy}, \eqref{eq:gausapp} and \eqref{eq:bw} respectively, where $n_{\alpha}=1.4$, $\lambda_e = 0.52$ $\mu$m, $n_o = 1.515$, $\sigma_a^2=49 \times 10^{-4}$ $\mu$m$^2$ and $z_{0}=1$ $\mu$m. By simulating the detected photon locations based on the same molecule trajectory and according to these three models, we obtain the observed photon trajectories in \cref{fig:mol:est}. Note that the parameters of the Airy and 2D Gaussian profiles have been chosen so that the Gaussian profile approximates the Airy profile.
\end{exmp}
\begin{figure*}[htbp!]
	\centering
	\includegraphics[width=.9\textwidth]{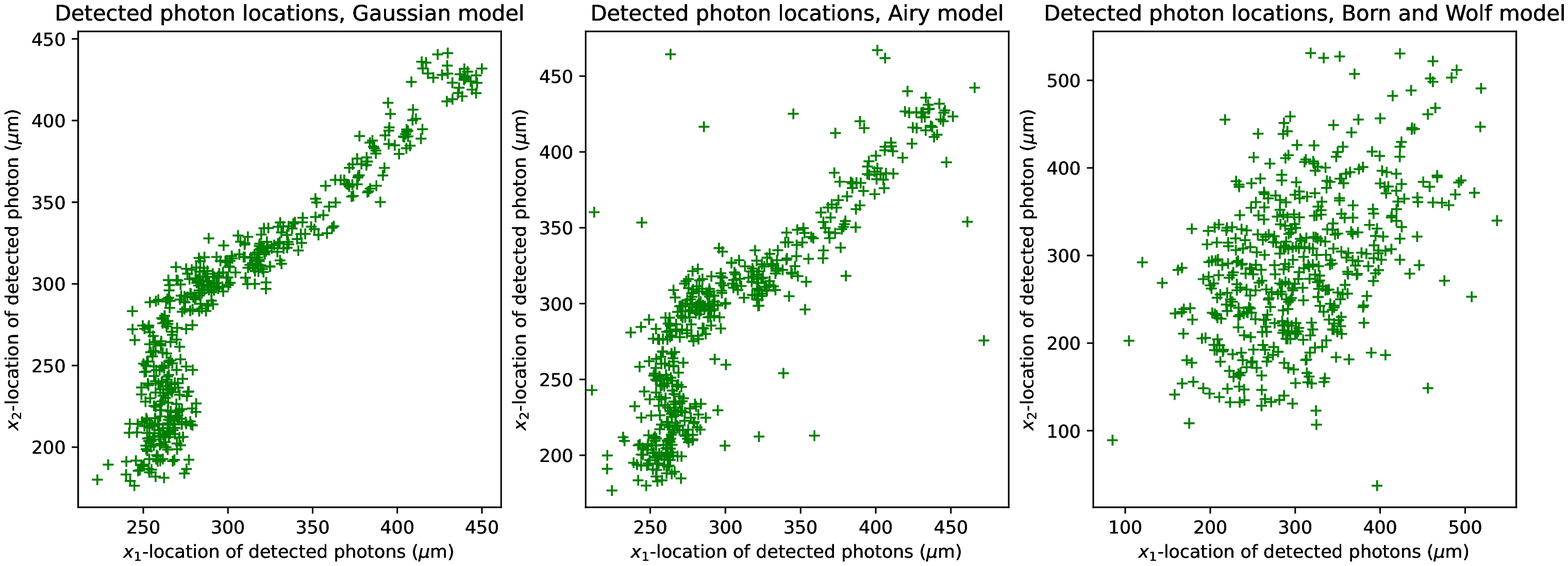}
	\caption[Detected photon locations of a moving molecule with stochastic trajectory for the 2D Gaussian, Airy profiles and Born and Wolf model.]{Detected photon locations of a moving molecule with stochastic trajectory for the 2D Gaussian (\textit{left}), Airy (\textit{middle}) profiles and Born and Wolf model (\textit{right}).}
	\label{fig:mol:est}
\end{figure*}
\section{The model as a state space model}
\label{sec:mol:hmm}
It is possible to reformulate this model as a state space model that takes into account the random arrival times of photons. This is achieved by discretising the time interval during which photons are recorded.

\subsection{Reformulation}
For simplicity, we assume for the rest of this paper (unless stated otherwise) that the photon detection rate is constant, i.e. $\lambda(t) = \lambda \in [0, 1]$ for all $t \geq t_0$. First of all, let $X_t = (x_{t, 1}, x_{t, 2}) \in \mathcal{X}$ where $\mathcal{X} := \mathbb{R}^2$ denotes the state of the molecule at time $t \geq t_0$, which includes its location $x_{t, 1:2}$ on the object plane. The location of the object on the optical axis is assumed to be constant and equal to the initial location parameter, i.e. $z_{0}$ for all $t \geq t_0$.
The probability of recording an observation, i.e. detecting a photon in the small interval $(t, t+h]$ is
\begin{equation*}
	\mathbb{P}\left[N(t+h)-N(t)>0\right] = \lambda h + o(t), \quad \lambda \in [0, 1], \, t \geq t_0.
\end{equation*}
Let $t_i$ denote the arrival time of the $i$-th photon on a detector $\mathcal{Y}$ for $i = 1, 2, \ldots$ and $Y_{t_i} \in \mathcal{Y}$ be the location of the captured photon on the detector. Assume the location of a detected photon is distributed according to the probability density function 
\begin{equation*}
	Y_{t_{i}}\vert\left(X_{t_{i}}=x\right)\sim g_{\theta}(\cdot\vert x), \quad i=1, 2, \ldots,
\end{equation*}
where $g_{\theta}$ is the photon distribution profile given in \eqref{eq:image}. The recorded data in the time interval $[t_0,T]$, $0 \leq t_0 < T$ comprises of $n_p$ observations with arrival times $t_0 < t_{1}<\ldots<t_{n_p}\leq T$ and photon locations $y_{t_{1}},\ldots,y_{t_{n_p}}$. The inference objective is to estimate the trajectory of the molecule $\left(X_{t}\right)_{t_0\leq t\leq T}$ given data $(t_{i},y_{t_{i}})$,
$i=1,\ldots,n_p$. As seen in \cref{subsec:x}, the molecule evolves according to the probability density function 
\begin{equation*}
	X_{t_{i+1}}\vert\left(X_{t_{i}}=x\right)\sim f_{t_{i}, t_{i+1}}^{\theta}(\,\cdot\,\vert x), \quad i=1, 2, \ldots,n_p,
\end{equation*}
where $\theta$ denotes the model parameters and $f_{s, t}^{\theta}$ for $t>s\geq t_0$ is the homogeneous continuous-time Markov transition density given by the the Gaussian distribution in \eqref{eq:x} for $d=2$.
\subsubsection{Non-constant photon detection rate} \label{sec:smm:nonconst}
If the photon detection rate $\lambda(t)$ is not assumed to be constant, then we redefine the state of an object at time $t \geq t_0$ as $X_t = (x_{t, 1}, x_{t, 2}, \lambda_t) \in \mathcal{X}$ where $\mathcal{X} := \mathbb{R}^2\times[0, 1]$. The state at time $t$ now includes the location of the molecule $(x_{t, 1}, x_{t, 2})$ as well as the probability $\lambda_t$ of detecting a photon it emits. The Markov transition density $p_{\theta}(x'|x)$ can be defined as follows
\begin{equation*}
	p_{\theta}(x_{t_{i+1}}|x_{t_i}) = f_{t_{i},t_{i+1}}^{\theta}(x_{t_{i+1}, 1:2}\vert x_{t_i, 1:2})l_{\theta}(\lambda_{t_{i+1}}|\lambda_{t_i}), \quad x_{t_{i+1}}, x_{t_i} \in \mathcal{X},
\end{equation*}
where $t_i$ and $t_{i+1}$ denote the arrival times of the $i$-th and $(i+1)$-th photons, respectively, $f_{t_{i},t_{i+1}}^{\theta}$ is the Markov transition density for the object location defined above and $l_{\theta}$ is the Markov transition density for the photon detection rate.
\subsection{Time discretisation}
\label{sec:tdiscret}
Let $(t_1, y_{t_1}), \ldots, (t_{n_p}, y_{t_{n_p}})$ be a realisation of the photon arrival times and locations observed in the time interval $[t_0,T]$. Setting $t_0 := 0$ for convenience, we adopt a discrete time formulation where $[0,T]$ is divided into segments of length $\Delta$. Let $x_{k} \in \mathcal{X}$ denote the state of the molecule at time $t=\left(k-1\right)\Delta$ where $k=1,\ldots, n$ for $n :=\left\lceil T/\Delta\right\rceil $. We assume the discretisation is fine enough so that an interval $(k\Delta,k\Delta+\Delta]$ contains
at most one arrival time $t_{i}$. Then, for $k=1,\ldots, n$, let
\[
y_{k}=\begin{cases}
	\emptyset, & \mathrm{if}\quad t_{i}\notin(k\Delta-\Delta,k\Delta], \quad \forall i=0, 1, \ldots, n_p,\\
	y_{t_{i}}, & \mathrm{if}\quad t_{i}\in(k\Delta-\Delta,k\Delta], 
\end{cases}
\]
where $y_{t_i} \in \mathcal{Y}$ denotes the location of the $i$-th detected photon on the detector $\mathcal{Y}$. The vector $y_{k}$ is assigned $\emptyset$ to indicate the absence of
an observation in the corresponding interval. See \cref{sec:discretisation} for details on why the time discretisation is a valid approximation of the Poisson process. If $x = (x_1, x_2, \lambda) \in \mathcal{X}$, let 
\begin{equation}
	\label{eq:pot:fun}
	G_{k}^{\theta}(x)=\begin{cases}
		1-\Delta\lambda, & \mathrm{if}\;y_{k}=\emptyset,\\
		\lambda g_{\theta}(y_{t_{i}}\vert x_{1:2}), & \mathrm{if}\;y_{k}=y_{t_{i}},
	\end{cases}
\end{equation}
where $g_\theta$ is the photon distribution profile \eqref{eq:image}, then $G_{k}^{\theta}(x)$ is the so called potential function. 
The potential $G_{k}^{\theta}(x)$ plays the role of the likelihood in Bayesian estimation problems. In the above context, the expression for $y_k=\emptyset$ corresponds to the probability of no photon being observed during that time interval. When a photon is observed in the interval, with observation time $t_i$ and observation location $y_{t_i}$ on the detector, the expression for $G_k^{\theta}$ is the product of the probability $\Delta\lambda$ of receiving one photon, with the uniform probability density $1/\Delta$ for the arrival time $t_i$ in that interval and the density of the location of the observation given that the molecule is situated at $x_{1:2}$ in the object space (the $\Delta$ terms then cancel out).\par 
For $k=1, \ldots, n$, the probability density function of $X_{k+1}$ given the previous state $X_k$ is $f_{\Delta}^{\theta}(x_{k+1}\vert x_{k}):=f_{k\Delta, k\Delta+\Delta}^{\theta}(x_{k+1}\vert x_{k})$ from \eqref{eq:x}, thus transforming \eqref{eq:x:sol} into
\begin{equation*}
	X_{k+1} = \Phi_\Delta X_{k} + a_\Delta + W_x, \quad W_x \sim \mathcal{N}\left(0, R_{\Delta}\right),
\end{equation*}
where $\Phi_\Delta = \Phi(k\Delta, k\Delta+\Delta)$ is now constant and similarly for $a_\Delta$ and $R_{\Delta}$.\par
To summarise, $\left(X_k\right)_{k=1}^{\infty}$ and $\left(Y_k\right)_{k=1}^{\infty}$ are $\mathcal{X}$- and $\mathcal{Y}\cup\emptyset$-valued stochastic processes where the molecule trajectory in the object space $\left(X_k\right)_{k=1}^{\infty}$ corresponds to the unobserved latent Markov process  with Markov transition density $f_{\Delta}^{\theta}(x'\vert x)$ and initial density $\nu_{\theta}(x)$, and the photon detection locations (or lack of) $\left(Y_k\right)_{k=1}^{\infty}$ represent the observed process with conditional density or potential function $G_{k}^{\theta}(x)$, i.e.
\begin{align}	
	X_1 \sim \nu_{\theta}(\cdot), &\quad X_{k+1}\vert\left(X_{k}=x\right) \sim f_{\Delta}^{\theta}(\cdot \vert x), \label{eq:mol:hmmx} \\
	Y_k\vert\left(X_{k}=x\right) \sim G_{k}^{\theta}(x), &\quad k=1, 2, \ldots.  
	\label{eq:mol:hmmy}
\end{align}
Note that if the object is static, so that the drift and diffusion coefficient in \eqref{eq:sde} are zero, the model simplifies from a state space model to a basic inference problem with independent observations. The observed process is still described by \eqref{eq:mol:hmmy} but the location of the object $x_0$ becomes part of the hyperparameters.
\section{Parameter inference}
\label{sec:mol:parinf}
\subsection{Inference aim}

Now that we have formulated the problem in \eqref{eq:mol:hmmx} and \eqref{eq:mol:hmmy} as a state space model, the first aim is going to be to estimate the posterior probability density function of $X_{1:n}:=\left\lbrace X_{1},\ldots, X_n\right\rbrace$, $n \in \mathbb{N}$, given the observations $Y_{1:n}$, also known as the \textit{joint smoothing distribution}, which is given by
\begin{equation}
	\label{eq:post}
	p_{\theta}(x_{1:n}|y_{1:n})=\dfrac{p_{\theta}(x_{1:n},y_{1:n})}{p_{\theta}(y_{1:n})},
\end{equation}
where the numerator represents the \textit{joint density}
\begin{equation}
	\label{eq:joint}
	p_{\theta}(x_{1:n},y_{1:n})=\nu_{\theta}(x_{1})\prod_{k=2}^n f_{\Delta}^{\theta}(x_{k}\vert x_{k-1}) \prod_{k=1}^nG_{k}^{\theta}(x_{k}),
\end{equation}
where $\nu_{\theta}(x_1)$ is the initial distribution of $X_1$, and the denominator represents the \textit{marginal likelihood} of the observed data
\begin{equation}
	\label{eq:lik1}
	p_{\theta}(y_{1:n})=\int_{\mathcal{X}^n} p_{\theta}(x_{1:n},y_{1:n}) \text{d}x_{1:n}.
\end{equation}
Estimating $p_{\theta}(x_{1:n}|y_{1:n})$ is what allows the molecule to be tracked and is done using a particle filter. The second aim is to obtain particle approximations of smoothed additive functionals, which in turn will allow for estimation of the score and OIM for of the hyperparameters $\theta$, as well as other applications such as ML estimation of said hyperparameters via gradient ascent and Expectation-Maximization (EM). Finally, the third aim is to use the estimates of the score and OIM of the hyperparameters to obtain an approximation of their FIM.

\subsection{Tracking the molecule using a particle filter}
The particle approximation of the marginal posterior of $X_{1},\ldots, X_n$ defined in \eqref{eq:post} is given by
\begin{equation*}
	\hat{p}(x_{1:n}|y_{1:n}) = \sum_{i=1}^N \omega_n^{(i)}\delta_{X_{1:n}^{(i)}}(x_{1:n}),
\end{equation*}
where $X_{1:n}^{(1:N)}$ are the particles, $\omega_n^{(1:N)}$ their corresponding normalised \textit{importance weights}, i.e. $\sum_{i=1}^N \omega_n^{(i)} = 1$ and $\delta_{v_0}(v)$ denotes the dirac delta mass located at $v_0$.
To obtain this particle approximation, we employ sequential Monte Carlo (SMC) methods in the form of a \textit{particle filter} (see \cite{cappe2006inference, douc2014nonlinear, doucet2001sequential, chopin2020introduction} for comprehensive reviews of SMC methods). There is flexibility in the specific choice of particle filter, but the general form they take follows three key steps, namely \textit{resample}$\rightarrow$\textit{propagate}$\rightarrow$\textit{weight}. For $k = 2, \ldots, n$, the \textit{resampling} step avoids weight degeneracy \cite{doucet2009tutorial, kong1994sequential} and consists of drawing indices $\iota_{k-1}^{(1:N)}$ with probabilities corresponding to the normalised weights $\omega_{k-1}^{(1:N)}$, then, depending on the resampling algorithm considered, resetting the weights accordingly, e.g. $\omega_{k-1}^{(1:N)} := \frac{1}{N}$.
The propagation and weighting steps consist of advancing the (resampled) particle population $(X_{k-1}^{\iota_{k-1}^{(1:N)}}, \omega_{k-1}^{(1:N)})$ forward in time via the \textit{proposal density} $\eta_k(x_k|x_{k-1})$ (\textit{propagate}) and updating the importance weights (\textit{weight}) as follows (see \cref{sec:supp:pf} for more details.):
\begin{equation*}
	\omega_k^{(i)} = \frac{\omega_{k-1}^{(i)}\tilde{w}\left(X^{(i)}_{k-1}, X^{(i)}_{k}\right)}
	{\sum_{j=1}^N \omega_{k-1}^{(j)}\tilde{w}\left(X^{(j)}_{k-1}, X^{(j)}_{k}\right)},
\end{equation*}
where $\tilde{w}(x_{k-1}, x_k)$ is known as the \textit{incremental weight} and is given by
\begin{equation*}
	\tilde{w}(x_{k-1}, x_k) = \frac{G^\theta_k(x_k)f^\theta_\Delta(x_{k}|x_{k-1})}{\eta_k(x_{k}|x_{k-1})}.
\end{equation*}
The proposal density is user-defined. For example, if $\eta_t(x_k|x_{k-1}) = f^\theta_\Delta(x_k|x_{k-1})$, the particle filter becomes the well-known \textit{bootstrap filter}, introduced in \cite{gordon1993novel} and the computation of the incremental weights simplifies to $\tilde{w}(x_{k}) = G_t^\theta(x_k)$. A generic particle filter is summarised in \cref{alg:smc:PF}.\par
Given weighted particle sample $(X_{k-1}^{(1:N)}, \omega_{k-1}^{(1:N)})$ at step $k$, we denote an iteration of running the particle filter (steps \ref{alg:step:resamp}-\ref{alg:step:weight} of \cref{alg:smc:PF}) as
$$\left(X_{k}^{(1:N)}, \omega_{k}^{(1:N)}\right) := \textsf{PF}_\Delta\left(X_{k-1}^{(1:N)}, \omega_{k-1}^{(1:N)}\right).$$ 
For this particular problem, we must also take into account the missing observations introduced by the time discretisation. Since a lack of observation does not bring any new information, it suffices to only run the particle filter at segments which contain an observation. A typical iteration of this approach is summarised in \cref{alg:smm:pf-missing}. The interval counter is initialised at $c_0 := 1$ and counts the number of discrete intervals since (and including) the last observation. An example of particle filtering for stochastically moving molecules observed through the 2D Gaussian, Airy and Born and Wolf models is available in \cref{ex:smm:pf}.
\begin{algorithm}[htbp]
	\caption{Particle filter for SDE with missing observations}\label{alg:smm:pf-missing}
	\begin{algorithmic}[1]	
		\Statex \textbf{Input}: weighted particle sample $\left(X_{k-1}^{(1:N)}, \omega_{k-1}^{(1:N)}\right)$ and interval counter $c_{k-1}$ at step $k-1$.
		\If{$y_k = \emptyset$}
		\State $c_k := c_{k-1} + 1$
		\State Do not run the particle filter
		$$\left(X_{k}^{(1:N)}, \omega_{k}^{(1:N)}\right):= \left(X_{k-1}^{(1:N)}, \omega_{k-1}^{(1:N)}\right).$$
		\Else
		\State Run the particle filter with updated interval length, i.e.
		$$\left(X_{k}^{(1:N)}, \omega_{k}^{(1:N)}\right) := \textsf{PF}_{c_k\Delta}\left(X_{k-1}^{(1:N)}, \omega_{k-1}^{(1:N)}\right).$$
		\State $c_k := 1$
		\EndIf
		\Statex \textbf{Output}: updated particle sample $\left(X_{k}^{(1:N)}, \omega_{k}^{(1:N)}\right)$.
	\end{algorithmic}
\end{algorithm}
\subsection{Particle approximations of expectations of additive functionals} \label{sec:smm:smcfs}
The second inference aim is to obtain estimates of the score and observed information matrix (OIM) for the hyperparameters $\theta$. To achieve these aims, we make use of smoothed additive functionals. Assume that there exists a real-valued function $S^{\theta}_k$, $k\geq 0$ such that it is an \textit{additive functional} given by
\begin{equation}
	\label{eq:add:fun}
	S_k^{\theta}(x_{1:k}) = \sum_{j=1}^{k} s^{\theta}_j(x_{j-1}, x_{j}),
\end{equation}
where $s^{\theta}_1(x_{0}, x_{1}):=s^{\theta}_1(x_{1})$ and $\left\{s^{\theta}_k\right\}_{k\geq 0}$ is a sequence of \textit{sufficient statistics} which may depend on the value of the observations $y_{0:k}$. The main aim is to compute the posterior or \textit{smoothing expectation}, given by
\begin{equation}
	\label{eq:add:exp}
	\mathcal{S}_k(\theta) := \mathbb{E}_{\theta}\left[S^{\theta}_k\left(X_{1:k}\right)|y_{1:k}\right] = \int_{\mathcal{X}} S^{\theta}_k\left(x_{1:k}\right) p_{\theta}(x_{1:k}|y_{1:k}) \text{d}x_{1:k}.
\end{equation} 
If the model in question is linear and Gaussian or the state space $\mathcal{X}$ is finite, then the expectation $\mathcal{S}_k(\theta)$ can be computed exactly by recursion. However, this is not the case if the Airy or Born and Wolf profiles are used to describe photon distribution. In this case, SMC methods can again be employed to approximate the expectation as follows
\begin{equation*}
	\hat{\mathcal{S}}_k(\theta) := \sum_{i=1}^N \omega_k^{(i)}S^{\theta}_k\left(X^{(i)}_{1:k}\right),
\end{equation*} 
where the weighted sample $(X^{(1:N)}_{1:k}, \omega_k^{(1:N)})$ is a particle approximation of the joint smoothing distribution $p_{\theta}(x_{1:k}|y_{1:k})$ obtained using a particle filter.\par 
A simple way of estimating the smoothing expectation $\mathcal{S}_n(\theta)$ for a set of $n$ observations $y_{1:n}$ is to run the desired particle filter in a `forward pass' through the whole data to obtain the particle approximation $(X^{(1:N)}_{n}, \omega_{n}^{(1:N)})$ at the final step $n$, followed then by a `backward smoothing' pass through the data, starting from the latest sample $y_n$. This is the case of algorithms such as the fixed-lag smoother by \cite{kitagawa2001monte, olsson2008sequential, olsson2011particle}, forward-filtering backward smoothing (FFBSm) by \cite{doucet2000sequential, hurzeler1998monte, kitagawa1996monte} and forward-filtering backward simulation (FFBSi) by \cite{godsill2004monte}.
However, if one wishes to avoid multiple passes through the data, it is also possible to take advantage of the form of the additive functional in \eqref{eq:add:fun} to estimate $\mathcal{S}_k(\theta)$ in an online or `forward-only' fashion, as proposed in \cite{del2010forward} and further developed in \cite{olsson2017efficient}. Introducing the \textit{auxiliary function}
\begin{equation*}
	T_k^{\theta}(x_{k}) := \int_{\mathcal{X}^{k-1}}S^{\theta}_k(x_{1:k})p_{\theta}(x_{1:k-1}|y_{1:k-1}, x_k)\text{d}x_{1:k-1},
\end{equation*}
the following recursion is then created:
\begin{equation}
	\label{eq:ex:rec}
	T_k^{\theta}(x_{k}) = \int_{\mathcal{X}} \left[T_{k-1}^{\theta}(x_{k-1}) + s_k^{\theta}(x_{k-1}, x_k)\right]p_{\theta}(x_{k-1}|y_{1:k-1}, x_k)\text{d}x_{k-1},
\end{equation}
where $T_0^{\theta}:=0$ and its particle approximation given the weighted sample $(X^{(1:N)}_{1:k}, \omega_k^{(1:N)})$ and previous state particle approximation $\hat{T}^{\theta}_{k-1}(X_{k-1}^{(1:N)})$ is given by
\begin{equation}
	\label{eq:smc:rec}
	\hat{T}^{\theta}_{k}\left(X_{k}^{(i)}\right)=\sum_{j=1}^N
	\Psi_k^{\theta}(i, j)
	\left[\hat{T}^{\theta}_{k-1}\left(X_{k-1}^{(j)}\right) + s_k^{\theta}\left(X_{k-1}^{(j)}, X_{k}^{(i)}\right)\right]
\end{equation}
for all $ i \in \{1, \ldots, N\}$, and where 
\begin{equation}
	\label{eq:psi}
	\Psi_k^{\theta}(i, j):=\frac{\omega_{k-1}^{(j)}f_{\Delta}^{\theta}\left(X_{k}^{(i)}|X_{k-1}^{(j)}\right)}
	{\sum_{j=1}^N\omega_{k-1}^{(j)}f_{\Delta}^{\theta}\left(X_{k}^{(i)}|X_{k-1}^{(j)}\right)}.
\end{equation}
Finally, using the recursion on the auxiliary function $T_k^{\theta}$, the smoothing expectation in \eqref{eq:add:exp} can be rewritten as
\begin{equation}
	\label{eq:fs:exp}
	\mathcal{S}_k(\theta) = \int_{\mathcal{X}} T_k^{\theta}(x_k)p_{\theta}(x_k| y_{1:k})\text{d}x_k,
\end{equation}
and its particle approximation is
\begin{equation}
	\label{eq:fs:exp:part}
	\hat{\mathcal{S}}_k(\theta) = \sum_{i=1}^N \omega_{k}^{(i)} \hat{T}_k^{\theta}\left(X_{k}^{(i)}\right).
\end{equation}
This algorithm is known as \textit{Forward smoothing SMC} (SMC-FS) and is summarised in the context of our experiments in \cref{alg:smc:smc-fs}. 
\begin{algorithm}[htbp]
	\caption{Forward smoothing SMC (SMC-FS)}
	\label{alg:smc:smc-fs}
	\begin{algorithmic}[1]
		\Statex \textit{Where $(i)$ or $(j)$ appears, the operation is performed for all $i, j \in \{1, \ldots, N\}$.}
		\Statex At $k=1$,
		\State Initialise the particle filter to obtain the weighted particle sample $\left(X^{(1:N)}_{1}, \omega_1^{(1:N)}\right)$.
		\State Initialise the interval counter $c_{0} := 1$.
		\State Set $\hat{T}^\theta_1\left(X_1^{(i)}\right) := 0$.
		\For{$k = 2, \ldots, n$}
		\If{$y_k = \emptyset$}
		\State $c_k := c_{k-1} + 1$
		\Else 
		\State Use the particle filter to update the weighted particle sample, i.e. 
		$$\left(X^{(1:N)}_{k}, \omega_{k}^{(1:N)}\right) := \textsf{PF}_{c_k\Delta}\left(X^{(1:N)}_{k-1}, \omega_{k-1}^{(1:N)}\right).$$
		\State Evaluate 
		$$\Psi^\theta_k(i, j):=\frac{\omega_{k-1}^{(j)}f^\theta_{c_k\Delta}\left(X_{k}^{(i)}|X_{k-1}^{(j)}\right)}
		{\sum_{j=1}^N\omega_{k-1}^{(j)}f^\theta_{c_k\Delta}\left(X_{k}^{(i)}|X_{k-1}^{(j)}\right)}.$$
		\State Update the auxiliary function estimate
		$$\hat{T}^\theta_{k}\left(X_{k}^{(i)}\right)=\sum_{j=1}^N
		\Psi_k^\theta(i, j)
		\left[\hat{T}^\theta_{k-1}\left(X_{k-1}^{(j)}\right) + s^\theta_k\left(X_{k-1}^{(j)}, X_{k}^{(i)}\right)\right].$$
		\State Update the smoothing expectation estimate
		$$\hat{\mathcal{S}}_k(\theta) = \sum_{i=1}^N \omega_k^{(i)}\hat{T}^\theta_{k}\left(X_{k}^{(i)}\right).$$
		\State Reset the interval counter $c_k := 1$.
		\EndIf
		\EndFor
		\Statex \textbf{Output}: smoothing expectation estimate $\hat{\mathcal{S}}_n$.
	\end{algorithmic}
\end{algorithm}
\subsection{Estimation of the score and observed information matrix (OIM)}
\label{sec:smm:scoreoim}
The score and OIM have important applications to ML estimation, e.g. see \cite{le1997recursive, poyiadjis2011particle}. They can also be instrumental in assessing the performance of such an estimator, either directly, as argued by \cite{efron1978assessing}, or as tools to estimate the FIM when the latter cannot be computed exactly, as we will see in this section. We aim to compute, recursively in time, the score vector $\mathcal{G}_k(\theta):=\nabla\log{p_{\theta}(y_{1:k})}$ and OIM $\mathcal{H}_k(\theta):=-\nabla^2\log{p_{\theta}(y_{1:k})}$ where $p_{\theta}(y_{1:k})$ denotes the marginal likelihood at step $1\leq k\leq n$ defined in \eqref{eq:lik1}, $\nabla$ denotes the gradient and $\nabla^2$ the Hessian. 
\subsubsection{Establishing the sufficient statistics}
The key to obtaining the particle approximation \eqref{eq:fs:exp:part} of a smoothing expectation \eqref{eq:fs:exp} of interest is to establish the relevant additive functionals and sufficient statistics. First of all, assume that the regularity conditions allowing for differentiation and integration to be switched around in expressions are satisfied. Let us establish the Fisher and Louis identities for the score and OIM, respectively, from \cite{cappe2006inference, douc2014nonlinear}:
\begin{align}
	\label{eq:FL}
	\mathcal{G}_k(\theta) &= \int_{\mathcal{X}}\nabla\log{p_{\theta}(x_k, y_{1:k})}p_{\theta}(x_k| y_{1:k})\text{d}x_k,\\
	\mathcal{H}_k(\theta) &= \nabla\log p_{\theta}(y_{1:k}) \nabla\log p_{\theta}(y_{1:k})^{\intercal}-
	\frac{\nabla^2 p_{\theta}(y_{1:k})}{p_{\theta}(y_{1:k})},\nonumber
\end{align}
where
\begin{align}
	\frac{\nabla^2 p_{\theta}(y_{1:k})}{p_{\theta}(y_{1:k})} &= 
	\int_{\mathcal{X}}\nabla\log p_{\theta}(x_k, y_{1:k}) \nabla\log p_{\theta}(x_k, y_{1:k})^{\intercal}p_{\theta}(x_k|y_{1:k})\text{d}x_k \nonumber \\
	&+ \int_{\mathcal{X}}\nabla^2\log p_{\theta}(x_k, y_{1:k})p_{\theta}(x_k|y_{1:k})\text{d}x_k,
	\label{eq:hess}
\end{align}
and note that \eqref{eq:FL} and \eqref{eq:hess} can be rewritten as
\begin{align}
	\label{eq:ps:smexp:score}
	\nabla\log{p_{\theta}(y_{1:k})} &= \mathbb{E}\left[\alpha_k^{\theta}(X_k)|y_{1:k}\right],\\
	\frac{\nabla^2 p_{\theta}(y_{1:k})}{p_{\theta}(y_{1:k})} &= 
	\mathbb{E}\left[\alpha_k^{\theta}(X_k)\alpha_k^{\theta}(X_k)^{\intercal}|y_{1:k}\right] 
	+ \mathbb{E}\left[\beta_k^{\theta}(X_k)|y_{1:k}\right],
	\label{eq:ps:smexp:oim}
\end{align}
where the expectations here are with respect to the density $p(x_k|y_{1:k})$, and correspond to the smoothing expectations in \eqref{eq:fs:exp}, with the functions $\alpha_k^{\theta}(x_k) := \nabla\log p_{\theta}(x_k, y_{1:k})$ and $\beta_k^{\theta} := \nabla^2\log p_{\theta}(x_k, y_{1:k})$ acting as the auxiliary functions of interest. A recursion for $\alpha_k^{\theta}$ and $\beta_k^{\theta}$ is straightforward to obtain, more details in \cite{poyiadjis2011particle}. For $\alpha_k^{\theta}$ and $\beta_k^{\theta}$, \eqref{eq:ex:rec} becomes 
\begin{align*}
	\alpha_k^{\theta}(x_k) &= \int_{\mathcal{X}}\left[\alpha^{\theta}_{k-1}(x_{k-1}) + s_k^{\alpha}(x_{k-1}, x_k)\right]p_{\theta}(x_{k-1}|y_{1:k-1}, x_{k})\text{d}x_{k-1},\\
	\beta_k^{\theta}(x_k) &= \int_{\mathcal{X}}\left[\beta^{\theta}_{k-1}(x_{k-1}) + s_k^{\beta}(x_{k-1}, x_k)\right]p_{\theta}(x_{k-1}|y_{1:k-1}, x_{k})\text{d}x_{k-1} - \alpha_k^{\theta}(x_k)\alpha_k^{\theta}(x_k)^{\intercal},
\end{align*}
where the sufficient statistics are given by
\begin{align}
	\label{eq:ex:rec:alpha}
	s_k^{\alpha}(x_{k-1}, x_k) &:= \nabla \log G^{\theta}_k(x_k) + \nabla \log f_{\Delta}^{\theta}(x_k|x_{k-1}),\\
	s_k^{\beta}(x_{k-1}, x_k) &:= \left[\alpha^{\theta}_{k-1}(x_{k-1}) + s_k^{\alpha}(x_{k-1}, x_k)\right]\left[\alpha^{\theta}_{k-1}(x_{k-1}) + s_k^{\alpha}(x_{k-1}, x_k)\right]^{\intercal} \nonumber \\
	&+ \nabla^2\log G^{\theta}_k(x_k) + \nabla^2 \log f_{\Delta}^{\theta}(x_k|x_{k-1}).
	\label{eq:ex:rec:beta}
\end{align}
Finally, to approximate the score and OIM, adapt the particle approximation in \eqref{eq:smc:rec} to the recursions in \eqref{eq:ex:rec:alpha} and \eqref{eq:ex:rec:beta} to obtain the score estimate, given by a weighted sum \eqref{eq:fs:exp:part} approximating the smoothing expectation \eqref{eq:ps:smexp:score}, i.e.
\begin{equation*}
	\hat{\mathcal{G}}_k({\theta}) = \sum_{i=1}^N \omega_k^{(i)}\hat{\alpha}_k^{\theta}\left(X_k^{(i)}\right)
\end{equation*} 
and OIM estimate
\begin{equation*}
	\hat{\mathcal{H}}_k({\theta}) = \hat{\mathcal{G}}_k({\theta})\hat{\mathcal{G}}_k(\theta)^{\intercal}-\sum_{i=1}^N \omega_k^{(i)}\left[ \hat{\alpha}_k^{\theta}\left(X_k^{(i)}\right)\hat{\alpha}_k^{\theta}\left(X_k^{(i)}\right)^{\intercal} + \hat{\beta}_k^{\theta}\left(X_k^{(i)}\right)\right],
\end{equation*}
where the weighted sum is the particle approximation \eqref{eq:fs:exp:part} of the smoothed expectation in \eqref{eq:ps:smexp:oim}. In \cref{ex:score}, we apply this framework to a possible application of the single-molecule tracking model. We focus for now on the case where the photon distribution is described by the Airy or 2D Gaussian profile.
\begin{exmp}
	\label{ex:score}
	Let the trajectory of a molecule be given by the following SDE
	\begin{equation*}
		\text{d}X_t = b\mathbb{I}_{2\times 2}X_t\text{d}t + \sqrt{2}\sigma\text{d}B_t,
	\end{equation*}
	where in the drift term, $b\neq 0$, in the diffusion term, $\sigma>0$, and $(\text{d}B_t)_{t_0\leq t\leq T}$ is a Wiener process. Let the photon detection process be described by the Airy or 2D Gaussian profile. Then, the parameters of interest are $\theta = (\sigma^2, b)$. Recall from \cref{sec:tdiscret} and \cref{ex:mol:1} that the solution to the SDE can be written as 
	\begin{equation}
		\label{eq:ex:score}
		X_k = e^{\Delta b} X_{k-1} + W_x, \quad W_x \sim \mathcal{N}\left(0, \frac{\sigma^2}{b}\left(e^{2\Delta b}-1\right)\mathbb{I}_{2\times2}\right),
	\end{equation}	
	and since the potential function $G_k$ does not depend on $\theta$ in this case, it can be dropped from \eqref{eq:ex:rec:alpha} and \eqref{eq:ex:rec:beta} and the components of the sufficient statistic $s^{\alpha}_k(x_{k-1}, x_k)$ for the additive functional $\alpha_k^{\theta}$ are
	\begin{align*}
		\frac{\partial}{\partial \sigma^2}\log{f_{\Delta}^{\theta}(x_k| x_{k-1})}&=
		-\frac{1}{\sigma^2}+\frac{b\norm{x_k-e^{\Delta b}x_{k-1}}^2}{2\sigma^4\left(e^{2\Delta b}-1\right)},\\
		\frac{\partial}{\partial b}\log{f_{\Delta}^{\theta}(x_k| x_{k-1})}&= 
		\frac{1}{b} -\frac{2\Delta e^{2\Delta b}}{\left(e^{2\Delta b}-1\right)} -\frac{\norm{x_{k}-e^{\Delta b}x_{k-1}}^2}{2\sigma^2(e^{2\Delta b}-1)}\\
		&+\frac{\Delta be^{\Delta b}(x_k- e^{\Delta b}x_{k-1})^{\intercal}x_{k-1}}{\sigma^2(e^{2\Delta b}-1)}
		+\frac{\norm{x_k - e^{\Delta b}x_{k-1}}^2\Delta be^{2\Delta b}}
		{\sigma^2(e^{2\Delta b}-1)^2}.
	\end{align*}
	The components of the sufficient statistic $s^{\beta}_k(x_{k-1}, x_k)$ for $\beta_k^{\theta}$ are given in \cref{app:smm:suffstat}. Note that these derivatives can be evaluated for any value of $\Delta$, and it is therefore possible to adapt them in order to only compute sufficient statistics when an observation is recorded as in \cref{alg:smm:pf-missing}. This is reflected in \cref{alg:smc:smc-fs}.
\end{exmp}
\subsection{Estimating the Fisher information matrix (FIM)}
\label{sec:fim}
The \textit{Fisher information matrix} (FIM) is widely used in estimation problems as an indicator of the performance of a given estimator. Indeed, it is a key element of the Cram\'{e}r-Rao inequality, or \textit{Cram\'{e}r-Rao Lower Bound} (CRLB) derived by \cite{cramer1999mathematical, rao1992information, frechet1943extension, darmois1945limites}, which states that for an unbiased estimate $\hat{\theta}$ of the parameter $\theta$, its covariance has lower bound
\begin{equation*}
	\text{Cov}(\hat{\theta}) \succeq \mathcal{I}_n(\theta)^{-1},
\end{equation*} 
where given matrices $A$ and $B$, the inequality $A \succeq B$ indicates that $A-B$ is a positive semi-definite matrix, and $\mathcal{I}_n(\theta)$ denotes the FIM in a random sample $Y_1, \ldots, Y_n$ of size $n$ \cite{degroot2012probability}, defined as
\begin{align}
	\label{eq:fim:sam}
	\mathcal{I}_n(\theta) &= \mathbb{E}_{\theta}\left[\nabla\log p_{\theta}(Y_{1:n}) \nabla\log p_{\theta}(Y_{1:n})^{\text{T}}\right]\\
	&= \mathbb{E}_{\theta}\left[-\nabla^2\log p_{\theta}(Y_{1:n})\right],
	\label{eq:fim:oim}
\end{align}
where the second equality is proven in \cite{duchi2016lecture}. When the expectations in \eqref{eq:fim:sam} and \eqref{eq:fim:oim} are intractable $-$ which is the case when the Airy profile is used to describe the photon detection locations in the single-molecule tracking model $-$ there are several ways one can go about estimating the FIM.
\subsubsection{Estimating the FIM for a single large sample using the OIM}
Firstly, note that from \eqref{eq:fim:oim}, the relationship between the FIM and OIM is simply
\begin{equation}
	\label{eq:fim:oim2}
	\mathcal{I}_n(\theta) = \mathbb{E}_{\theta}\left[\mathcal{H}_n(\theta)\right],
\end{equation}
where $\mathcal{H}_n(\theta) = -\nabla^2\log p_{\theta}(y_{1:n})$ denotes the OIM. Then, for a general state space model, in \cite{bickel1998asymptotic}, it was proven that under mild assumptions, 
\begin{equation*}
	\frac{1}{n}\mathcal{H}_n(\theta) \rightarrow \mathcal{I}(\theta)  \quad \text{as } n \rightarrow \infty,
\end{equation*}
where $\mathcal{I}(\theta)$ is the asymptotic FIM. See \cite{houssineau2019identification} for the corresponding result for multiple targets. So for a large enough sample size $n$, i.e. if the interval during which the molecule(s) of interest are observed is long enough, the OIM and FIM can be used interchangeably, i.e. for $n\gg1$,
\begin{equation} \label{eq:fim:oim1}
	\mathcal{H}_n(\theta) \approx \mathcal{I}_n(\theta).
\end{equation}
See \cref{fig:fim:stat-airy} for an illustration. Therefore, the first way of estimating the \textit{asymptotic} FIM in the single-molecule tracking model is simply to obtain the OIM for a large sample size. For more details on the OIM as an estimate of the FIM, see \cite{douc2004asymptotic}.
\subsubsection{Estimating the FIM using the mean outer product of the score}
If the molecule(s) of interest are only observed for a short interval, then the size $n$ of the sample of interest is not large enough to estimate the FIM using the OIM. It is then also possible to instead obtain a particle approximation of the expectation in \eqref{eq:fim:sam} using the score as follows: generate $D$ datasets $y^{(1:D)}_{1:n}$ of (smaller) size $n$ where $y^{(d)}_{1:n} := \{y^{(d)}_{1}, \ldots, y^{(d)}_{n}\}$, and according to the same parameters $\theta$. The outer product of the score can then be used in the estimate of the FIM as follows:
\begin{equation} \label{eq:fim:op}
	\hat{\mathcal{I}}_n(\theta) = \frac{1}{D} \sum_{j=1}^{D} \mathcal{G}^{(d)}_n(\theta)\mathcal{G}^{(d)}_n(\theta)^{\intercal},
\end{equation}
where for $d = 1, \ldots, D$, the vector $\mathcal{G}^{(d)}_n(\theta) := \nabla\log p_{\theta}(y^{(d)}_{1:n})$ is the score for the $d$-th dataset of size $n$. An advantage of this approach is that the OIM need not be computed. \par 
\subsubsection{Estimating the FIM using the mean OIM}
When multiple datasets are available, the OIM can also similarly be averaged over $D$ datasets to estimate the FIM as follows:
\begin{equation}
	\label{eq:fim:oim3}
	\hat{\mathcal{I}}_n(\theta) = \frac{1}{D}\sum_{j=1}^D \mathcal{H}^{(d)}_n(\theta).
\end{equation}
This third approach is the Monte Carlo estimator of the expectation in \eqref{eq:fim:oim2}, and can be seen as averaging the first estimation method in \eqref{eq:fim:oim1}.\par
Now that the various methods for estimating the FIM have been established, it can be used in an experimental design setting to plan experiments with the aim of returning the most accurate parameter estimates. See \cref{sec:smm:em} for details on how ML estimates can similarly be obtained via EM and gradient ascent methods with the use of  smoothed additive functionals and SMC-FS.
\begin{exmp} \label{ex:smm:fimest}
	To verify these approaches to estimate the FIM, consider the straightforward special case of estimating the FIM for the location $x_0=(x_{0,1}, x_{0,2})$ parameters of a static molecule emitting photons at a constant rate. In \cite{ober2004localization, chao2016fisher}, the analytical expression for the FIM is derived for 
	the Airy profile, and its diagonal components given observations $y_{1:n}$ are given by 
	\begin{align*}
		\mathcal{I}_n^{\text{Airy}}(x_{0,1}) = \mathcal{I}_n^{\text{Airy}}(x_{0,2}) &= N_{phot} \alpha^2,
	\end{align*}
	where $\alpha = \frac{2\pi n_a}{\lambda_e}$, $N_{phot}$ denotes the expected photon count, and $\mathcal{I}^{\text{Airy}}(x_{0,i})$ denotes the $(i, i)$-th element of the FIM, corresponding to parameter component $x_i$, for the Airy profile. 
	As mentioned in \cref{sec:tdiscret}, having a static molecule simplifies the model. Since we have independent data, the true values of score $\mathcal{G}$ and OIM $\mathcal{H}$ can be derived as follows. Given a set of observations $y_{1:n}$ distributed according to the Airy profile,
	\begin{align*}
		\mathcal{G}_n^{\text{Airy}}(x_0) &= \sum_{k=1}^n \gamma_k (M^{-1}y_k-x_0)\mathbbm{1}_{y_k\neq \emptyset}, \\
		\mathcal{H}_n^{\text{Airy}}(x_0) &= \sum_{k=1}^n \left(\chi_k (M^{-1}y_k-x_0)(M^{-1}y_k-x_0)^{\intercal} + \gamma_k \mathbb{I}_{2\times2}\right) \mathbbm{1}_{y_k\neq \emptyset},
	\end{align*}
	where 
	\begin{equation*}
		\gamma_k =  \frac{2\alpha}{r} \frac{J_2(\alpha r_k)}{J_1(\alpha r_k)}, \qquad
		\chi_k =  -\frac{2\alpha^2}{r_k^2}
		\left[\frac{J_3(\alpha r_k)}{J_1(\alpha r_k)} -
		\frac{J_2^2(\alpha r_k)}{J_1^2(\alpha r_k)}\right],
	\end{equation*}
	and $r_k = \sqrt{(M^{-1}y_k-x_0)^\intercal(M^{-1}y_k-x_0)}$. See \cref{app:smm:airy-static} for the full derivation.\par
	Using the same settings as in \cref{ex:smm:obs}, we simulate $D_l=40$ `large' datasets according to the Airy 
	profile consisting of observations obtained during the interval $[0, 0.2]$ seconds. We also simulate $D_s=400$ `short' datasets consisting of observations obtained during the shorter interval $[0, 0.02]$ seconds. The score and OIM are obtained for all datasets and the FIM for the large and short datasets is estimated in three ways: \textit{(i)} using the OIM returned from a single dataset selected at random \eqref{eq:fim:oim1}, \textit{(ii)} using the mean outer product of the score \eqref{eq:fim:op} over all datasets and \textit{(iii)} using the mean OIM across all datasets \eqref{eq:fim:oim3}. Finally, the square root of the CRLB, also known as the (fundamental) \textit{limit of accuracy} and defined as
	$$\delta_\vartheta = \sqrt{CRLB_\vartheta}$$
	for parameter $\vartheta$ is obtained. This is repeated for various expected photon counts in order to compare the evolution of the estimated limit of accuracy as the expected number of photons increases to the true limit of accuracy obtained using the true FIM. In \cref{fig:fim:stat-airy}, it is apparent that, apart from very low photon counts, all approaches are able to return accurate estimates of the limit of accuracy. Comparing \cref{fig:sub-first-airy} and \cref{fig:sub-second-airy}, it also becomes apparent that for long datasets, approach \textit{(i)} is slightly more accurate than \textit{(ii)}, and the opposite is true for short datasets. In both cases, approach \textit{(iii)} is the most accurate. Similar results can be obtained for the 2D Gaussian profile and Born and Wolf model, as analytical expressions for the FIM are also available for a static object \cite{ober2020quantitativech18, ober2020quantitativech19}.
	\begin{figure}
		\centering
		\begin{subfigure}{.49\textwidth}
			\centering
			\includegraphics[width=.9\textwidth]{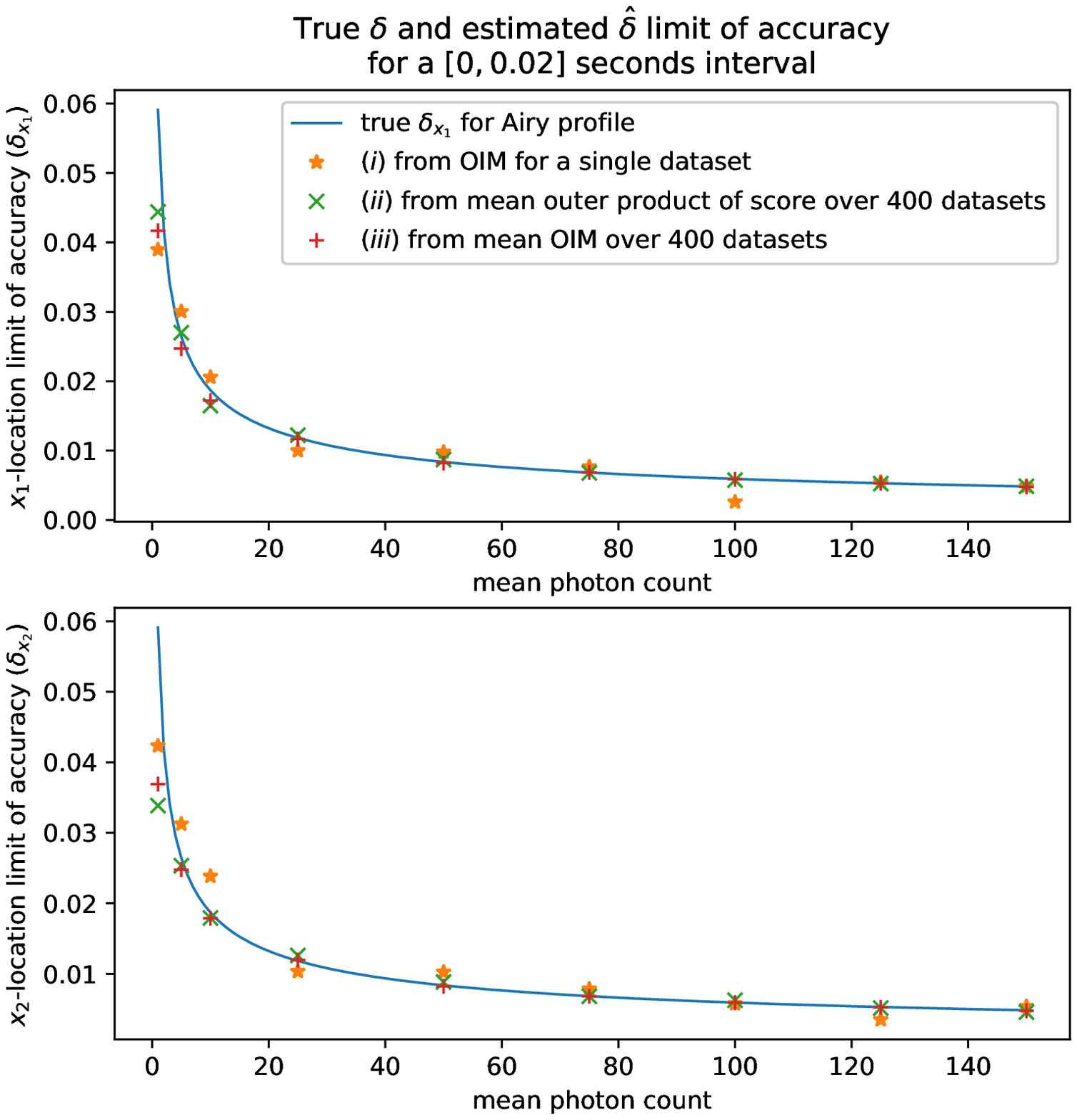} 
			\includegraphics[width=\textwidth]{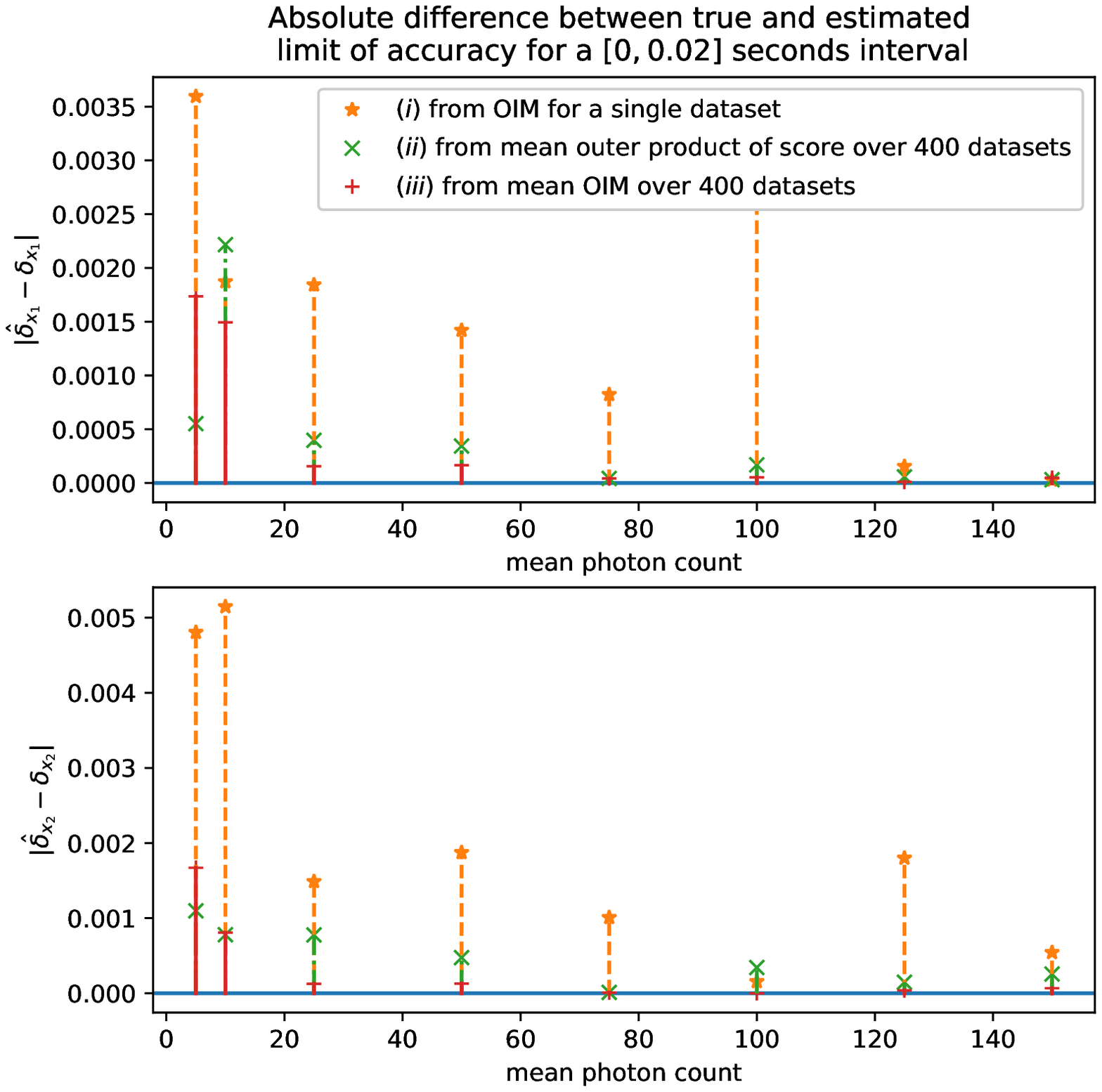} 
			\caption{Many short datasets, Airy profile}
			\label{fig:sub-first-airy}
		\end{subfigure}
		\begin{subfigure}{.49\textwidth}
			\centering
			\includegraphics[width=.95\textwidth]{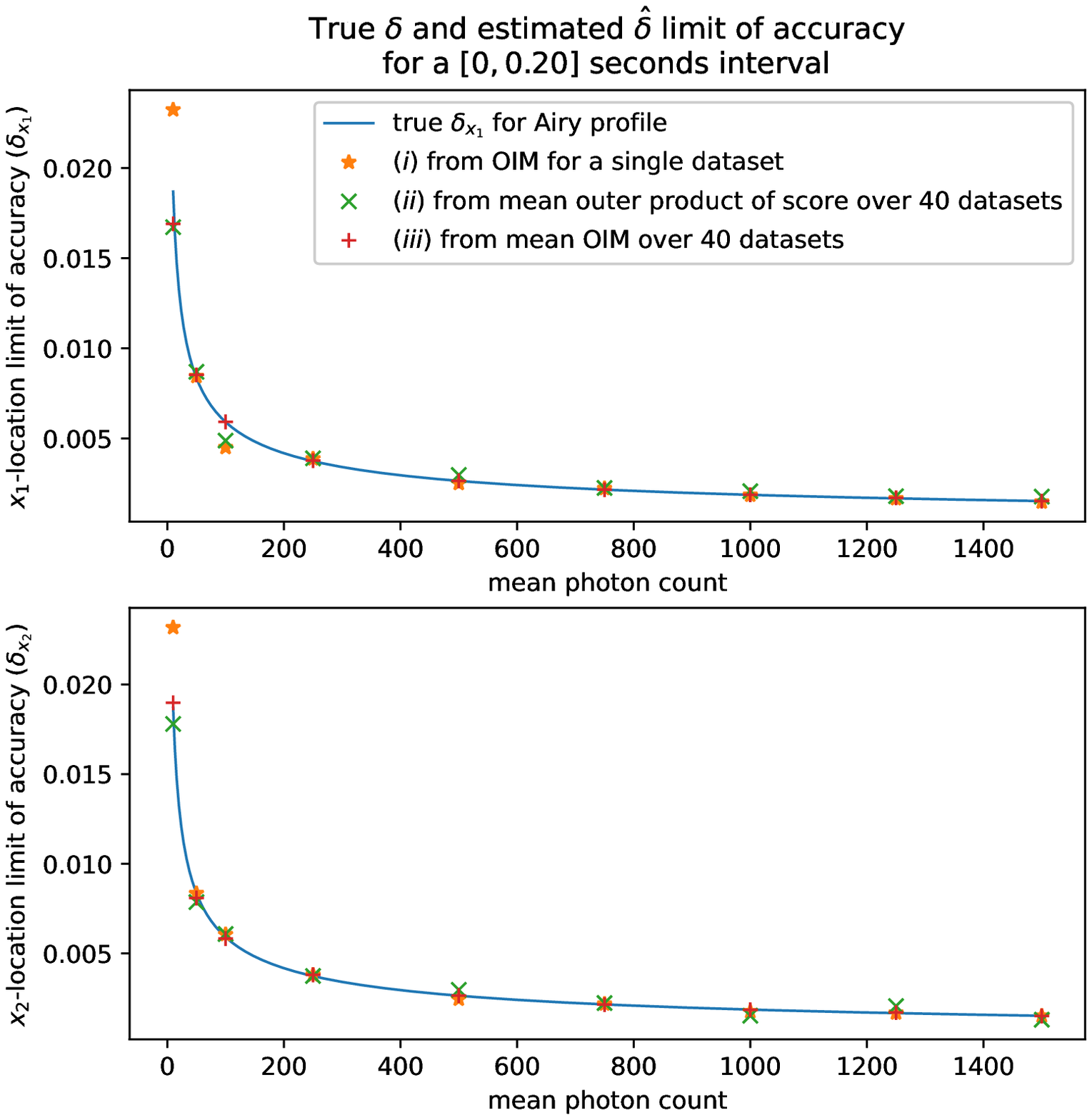} 
			\includegraphics[width=\textwidth]{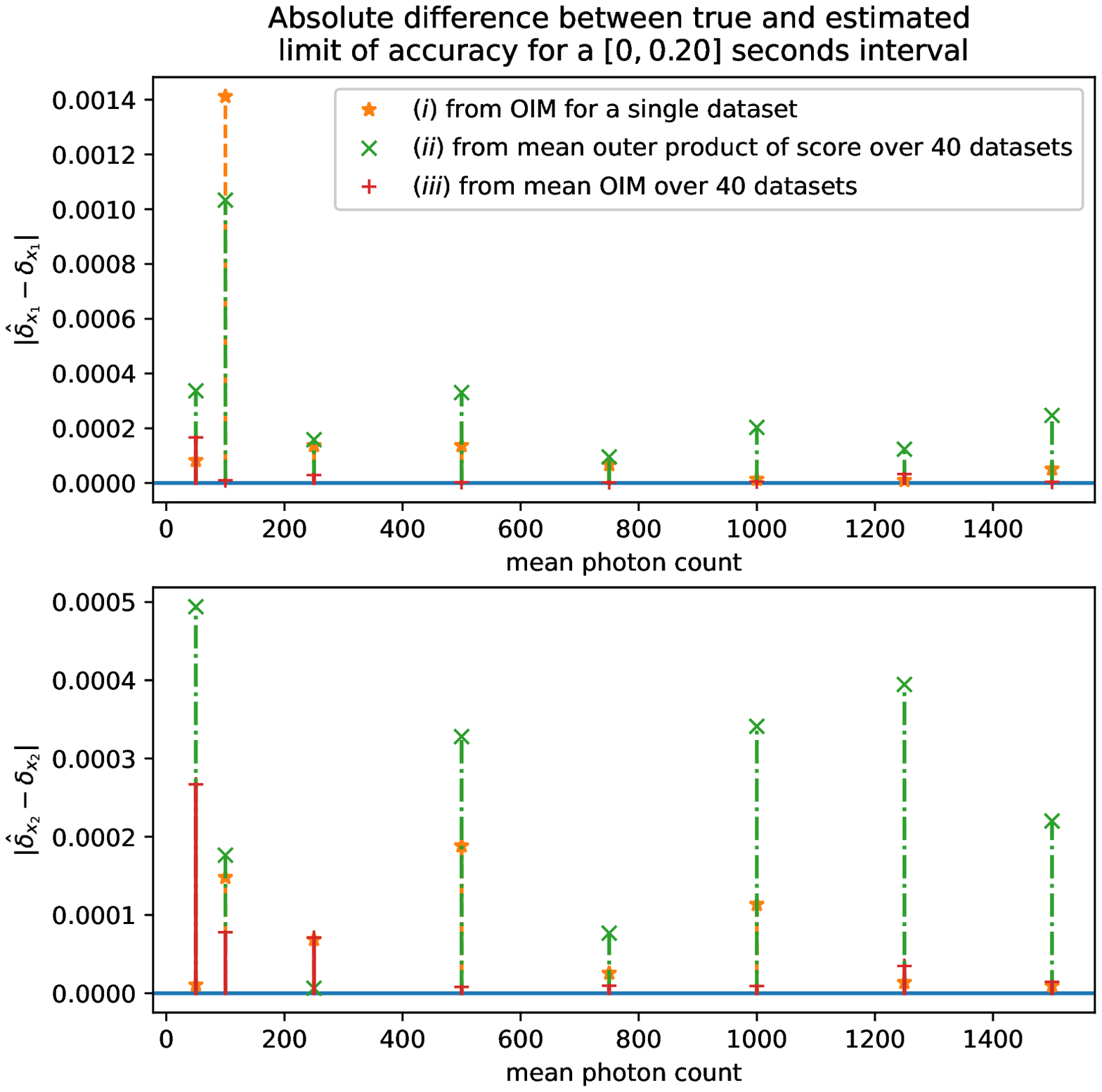} 
			\caption{Few long datasets, Airy profile}
			\label{fig:sub-second-airy}
		\end{subfigure}
		\caption[True and estimated limit of accuracy for short and long datasets distributed according to the Airy profile.]{True and estimated limit of accuracy for mean photon counts ranging from \textbf{(a)} 1 to 150 \textbf{(b)} 10 to 1500. The limit of accuracy is estimated for the location parameters $(x_1, x_2)$ of a static in-focus molecule. The estimates are obtained by taking the square root of the inverse of the FIM, obtained for \textbf{(a)} 400 `short' \textbf{(b)} 40 `long' simulated datasets using approaches \textit{(i)} \textcolor{orange}{$\star$}, \textit{(ii)} \textcolor{green}{$\times$} and \textit{(iii)} \textcolor{red}{$+$} for comparison purposes. To generate each dataset, the photon detection times are simulated according to a Poisson process with constant rate corresponding to the expected mean photon count for \textbf{(a)} $[0, 0.02]$  \textbf{(b)} $[0, 0.2]$ seconds and the intervals are discretised. The  photon detection locations are generated according to the Airy profile, with parameters as in \cref{ex:smm:obs}. The true limit of accuracy (\textit{blue solid line}) is also computed as it is available analytically \cite{ober2020quantitativech18}. Estimates of the limit of accuracy based on a single dataset (approach \textit{(i)}) are more accurate when the dataset is long, while taking the mean outer product of the score over all datasets (approach \textit{(ii)}) yields more accurate estimates for a large number of short datasets. Approach \textit{(iii)} provides a good balance between the two. In general, estimates of the limit of accuracy are relatively poor for very low mean photon counts but quickly improve as it increases.}
		\label{fig:fim:stat-airy}
	\end{figure}
\end{exmp}
\section{Numerical experiments}
\label{sec:mol:exp}
In this section, we apply the particle smoother known as SMC-FS to estimate the FIM, and thus the limit of accuracy, for various parameters in the context of one or multiple moving molecules with stochastic trajectories. Experiments are first run with photon detection locations described by the Gaussian and Airy profiles, and then the Born and Wolf model, where an additional hyperparameter, namely the optical axis location, must be considered as well. The methodology is then applied to the optical microscope resolution problem, where the limit of accuracy for the mean separation distance between two closely spaced diffusing molecules is assessed.\par 
Unless stated otherwise, the FIM for any given settings is estimated according to  \eqref{eq:fim:oim3}, i.e. by generating several datasets according to the same settings, estimating the OIM for each dataset using the SMC-FS algorithm (\cref{alg:smc:smc-fs}) and averaging the estimated OIM over all generated datasets. The particle filter employed in the SMC-FS algorithm is the bootstrap filter. A large number of datasets is needed to minimise Monte Carlo error in FIM estimates, so to speed up computations we adopt a distributed computing approach: the datasets and repeat runs of the SMC-FS algorithm to estimate the OIMs are divided evenly among 60 to 64 CPUs and run in parallel.  
We note that for our methodology, access to a large number of CPUs is beneficial to both the accuracy of estimates and the speed at which they can be obtained. The wall clock speed of the SMC-FS algorithm is also affected by the mean photon count $N_{phot}$ considered. Indeed, as described in \cref{alg:smc:smc-fs}, the filtering and smoothing steps only occur in segments where a photon is observed, so the expected complexity of a full run of the SMC-FS algorithm is $\mathcal{O}(N_{phot}N^2)$ where $N$ is the size of the SMC particle population (generally $N=500$).
\subsection{Limit of accuracy of drift and diffusion coefficients for the Gaussian and Airy profiles} \label{sec:fim:gaussian-airy}
Consider a molecule with trajectory described by the SDE in \cref{ex:mol:1}. In \cite{vahid2020fisher}, the authors took advantage of the Kalman filter formulae to evaluate the FIM for the diffusion ($\sigma^2$) and drift ($b$) coefficients. However, it was only possible to obtain an analytic solution for a particular set of detection times $t_1, t_2, \ldots$ and for the 2D Gaussian photon distribution profile. Otherwise, the computational cost of performing numerical integration was too high for more than one photon. \par 
In our particle filtering framework, it is also possible to take advantage of the Kalman filter formulae when considering the 2D Gaussian model in order to obtain an accurate approximation of the true score and OIM by numerical differentiation, and for any detection times schedule. An estimate of the FIM is therefore obtained by evaluating the true OIM for 3000 datasets and taking their mean, as described in \cref{sec:fim}. The molecule trajectories are simulated for $[0, 0.2]$ seconds, with diffusion coefficient $\sigma^2=1$ $\mu$m$^2/$s, drift coefficient $b=-10$ s$^{-1}$, and initial location Gaussian distributed with mean $x_0=(5.5, 5.5)^\intercal$ $\mu$m and covariance $P_0 = 10^{-2}\mathbb{I}_{2\times2}$ $\mu$m$^2$. The observations for the first experiment are generated according to the 2D Gaussian profile \eqref{eq:gausapp} with parameters as in \cref{ex:smm:obs}. It is not possible to employ the Kalman filter formulae for the Airy and Born and Wolf profiles, and we must resort to using the SMC-FS algorithm instead. First of all, to evaluate the performance of the SMC-FS algorithm, the algorithm is employed using $N=500$ particles to estimate the score and OIM for the same 3000 2D Gaussian profile datasets, and we similarly take the mean OIM over all datasets to estimate the FIM. Next, we move on to the Airy profile, for which it was too computationally costly in \cite{vahid2020fisher} to obtain the FIM for more than a single photon. We estimate the OIM for the diffusion and drift coefficients using the SMC-FS algorithm with $N=500$ particles for 2040 datasets, where the molecule trajectories are simulated using the same parameters as for the 2D Gaussian profile, and the observations are generated according to the Airy profile \eqref{eq:airy} with parameters as in \cref{ex:smm:obs}. This is repeated for various mean photon counts ranging from 10 to 1250. Then, the limit of accuracy estimate, denoted $\hat{\delta}_{\vartheta}$ for hyperparameter $\vartheta$, is computed, and the results are displayed in \cref{fig:fim:gaussian-airy}.\par
Both \cref{fig:sub-fim-gaussian} and \cref{fig:sub-fim-airy} display an inverse square root decay of the limit of accuracy with respect to the mean photon count. This is consistent with the results for a static molecule from \cref{ex:smm:fimest}, and means that the quality of diffusion and drift estimates improves as the mean photon count increases. In addition to that, comparing the limit of accuracy obtained from the estimated and true OIM for the 2D Gaussian profile in \cref{fig:sub-fim-gaussian} indicates that the SMC-FS algorithm is able to return accurate estimates of the score and FIM for a stochastically moving molecule. Indeed, apart from a very slight discrepancy for very low photon counts for the drift coefficient, the estimates of the limit of accuracy are almost indistinguishable. 
\begin{figure}[htbp!]
	\centering
	\begin{subfigure}{.49\textwidth}
		\centering
		\includegraphics[width=\textwidth]{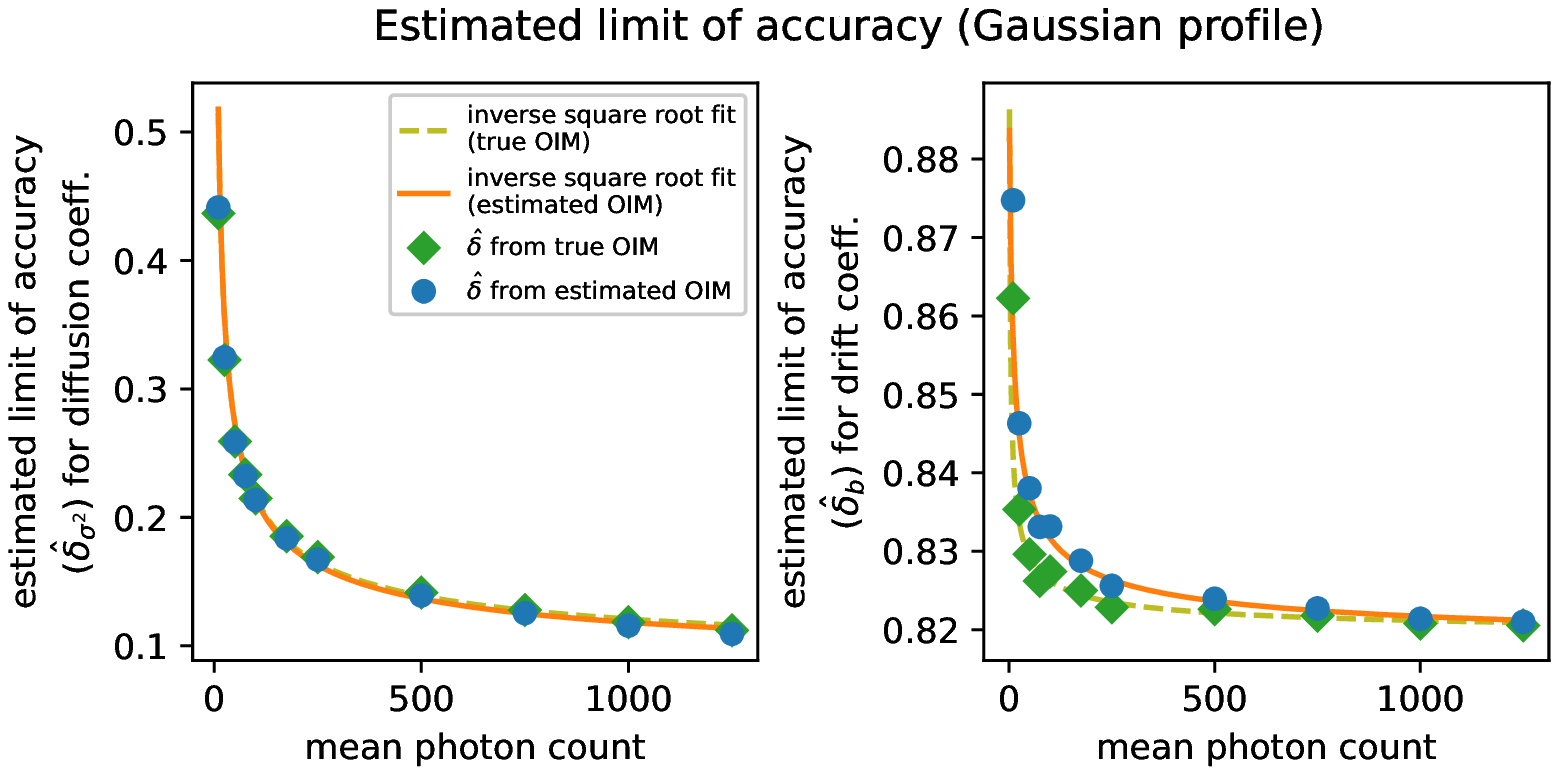} 
		\caption{2D Gaussian profile}
		\label{fig:sub-fim-gaussian}
	\end{subfigure}
	\begin{subfigure}{.49\textwidth}
		\centering
		\includegraphics[width=\textwidth]{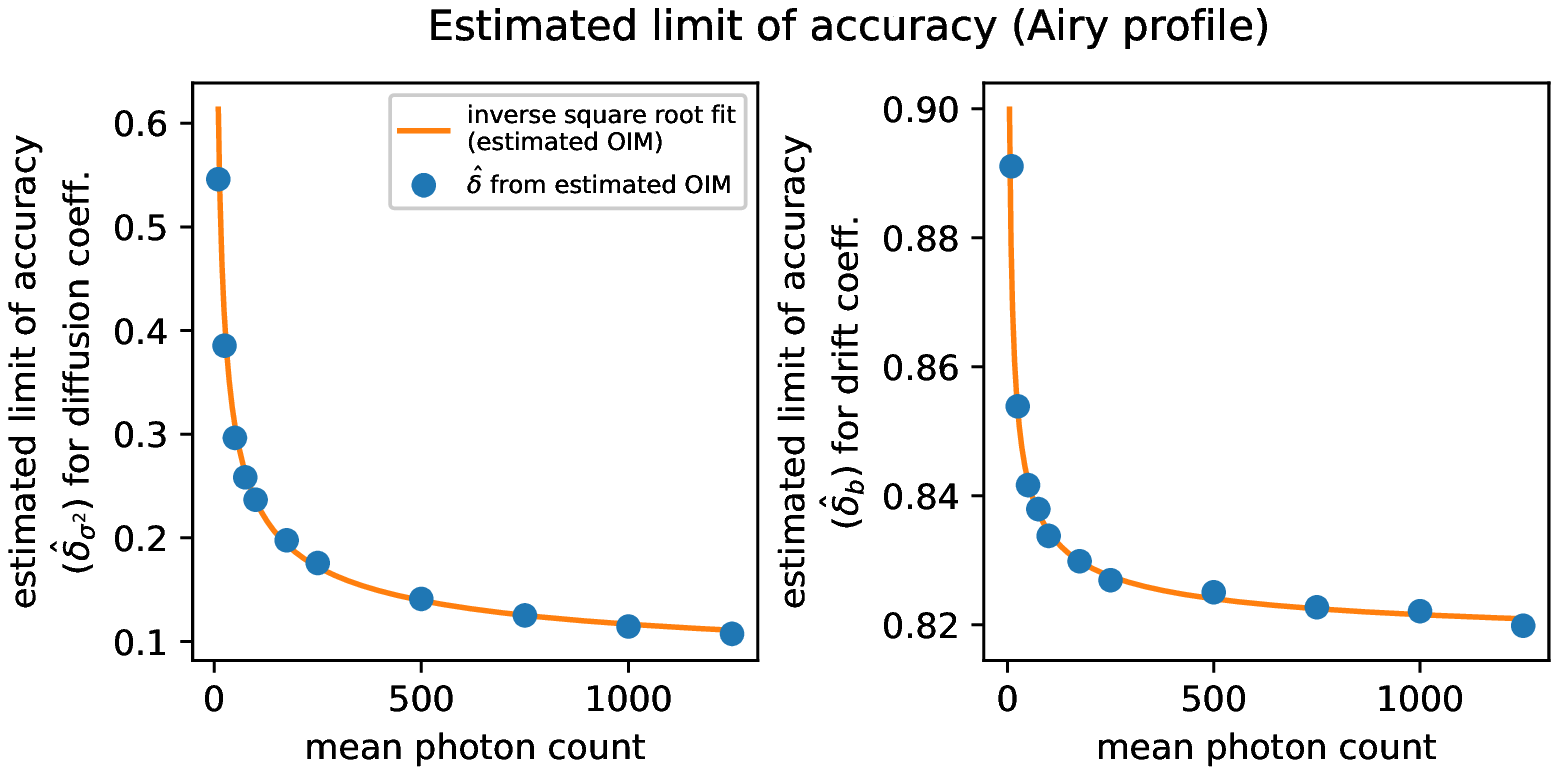} 
		\caption{Airy profile}
		\label{fig:sub-fim-airy}
	\end{subfigure}
	\caption[Estimated limit accuracy for the diffusion ($\sigma^2$) and drift ($b$) coefficients for an in-focus molecule with stochastic trajectory and photon detection locations described by the \textbf{(a)} 2D Gaussian and \textbf{(b)} Airy profiles.]{Evolution of the estimated limit accuracy for mean photon counts ranging from 10 to 1250. The limit of accuracy is estimated for the diffusion ($\sigma^2$) and drift ($b$) coefficients for an in-focus molecule with stochastic trajectory. The estimates are obtained by taking the square root of the inverse of the FIM, obtained by estimating the OIM using the SMC-FS algorithm with 500 particles for \textbf{(a)} 3000 and \textbf{(b)} 2040 simulated datasets. To generate each dataset, the molecule's trajectory was simulated according to the SDE in \cref{ex:mol:1} for the interval $[0, 0.2]$ seconds, with $\sigma^2=1$ $\mu$m$^2/$s, $b=-10$ s$^{-1}$, and initial location Gaussian distributed with mean $x_0=(5.5, 5.5)^\intercal$ $\mu$m and covariance $P_0 = 10^{-2}\mathbb{I}_{2\times2}$ $\mu$m$^2$. The observations are generated according to the \textbf{(a)} 2D Gaussian and \textbf{(b)} Airy profiles, with parameters as in \cref{ex:smm:obs}. For the \textbf{(a)} 2D Gaussian profile, the limit of accuracy is also estimated by using the true OIM obtained using numerical differentiation applied to the Kalman filter. An inverse square root curve (\textit{orange} and \textit{green dashed}) is fitted to the resulting estimated limits of accuracy for comparison.}
	\label{fig:fim:gaussian-airy}
\end{figure}
\subsection{Limit of accuracy of drift, diffusion and optical axis location for the Born and Wolf model} \label{sec:fim:bw}
When the molecule is out of focus, which means the photon detection locations are distributed according to the Born and Wolf model \eqref{eq:bw}, the FIM components for the diffusion and drift coefficients can be obtained as for the Airy and Gaussian profiles. However, a new hyperparameter must be considered, namely the optical axis location, denoted $z_0$. While previously, differentiating the log potential function was not needed, the vector of hyperparameters is now $\theta = (\sigma^2, b, z_0)$, and $G^{\theta}_k(x_k)$ depends on $z_0$ for $k = 1, \ldots, n$. \par 
While it requires numerical integration, differentiating $\log q_{z_0}(x_{1}, x_{2})$ for a given $x=(x_1, x_2) \in \mathbb{R}^2$ with respect to $z_0$ is not impossible. For notational simplicity, let $\alpha := \frac{2\pi n_{\alpha}}{\lambda_e}$, $r := \sqrt{x_1^2+x_2^2}$ and $W := \frac{\pi n^2_{\alpha}}{n_o\lambda_e}$ and rewrite \eqref{eq:bw} as 
\begin{equation*}
	q_{z_{0}}(x_1, x_2) = \frac{\alpha^2}{\pi}\left(U_{z_0}^2 + V_{z_0}^2\right),
\end{equation*}
where 
\begin{align*}
	U_{z_0} &:= \int_{0}^1J_0\left(\alpha r\rho\right)\cos{\left(Wz_0\rho^2\right)}\rho d\rho,&\qquad
	V_{z_0} &:= \int_{0}^1J_0\left(\alpha r\rho\right)\sin{\left(Wz_0\rho^2\right)}\rho d\rho.
\end{align*}
The first derivative was derived in \cite{ober2020quantitativech19} and is given by
\begin{equation*}
	\frac{\partial\log q_{z_{0}}(x_1, x_2)}{\partial z_0} = 2 \frac{U_{z_0}\dot{U}_{z_{0}} + V_{z_0}\dot{V}_{z_{0}}}
	{U_{z_0}^2 + V_{z_0}^2},
\end{equation*}
where
\begin{align*}
	\dot{U}_{z_{0}}  := \frac{\partial U_{z_{0}}}{\partial z_0} &= \int_{0}^1 
	J_0\left(\alpha r \rho\right)
	\cos{\left(Wz_0\rho^2\right)} W \rho^3d\rho, \\
	\dot{V}_{z_{0}} := \frac{\partial V_{z_{0}}}{\partial z_0} &= -\int_{0}^1
	J_0\left(\alpha r \rho\right)
	\sin{\left(Wz_0\rho^2\right)} W \rho^3d\rho.
\end{align*} 
The second derivative with respect to $z_0$ is given by
\begin{equation*}
	\frac{\partial ^2\log q_{z_{0}}(x_1, x_2)}{\partial z_0^2} =  2\frac{U_{z_0}\ddot{U}_{z_{0}} + \dot{U}_{z_{0}}^2 
		+ V_{z_0}\ddot{V}_{z_{0}} + \dot{V}_{z_{0}}^2}
	{U_{z_0}^2 + V_{z_0}^2} 
	- \left(\frac{\partial\log q_{z_{0}}(x_1, x_2)}{\partial z_0}\right)^2,
\end{equation*}
where
\begin{align*}
	\ddot{U}_{z_{0}} := \frac{\partial^2 U_{z_{0}}}{\partial z_0^2} &= -\int_{0}^1
	J_0\left(\alpha r \rho\right)
	\cos{\left(W z_0\rho^2\right)} W^2\rho^5 d\rho, \\
	\ddot{V}_{z_{0}} := \frac{\partial^2 V_{z_{0}}}{\partial z_0^2} &= - \int_{0}^1 
	J_0\left(\alpha r \rho\right)
	\sin{\left(W z_0\rho^2\right)} W^2 \rho^5 d\rho.
\end{align*}
The potential function only depends on $z_0$, so any cross terms in the FIM and OIM between $z_0$ and either $\sigma^2$ or $b$ will be zero. \par
The OIM is estimated for the diffusion ($\sigma^2$), drift ($b$) coefficients and optical axis location ($z_0$) using the SMC-FS algorithm with 500 particles for 2040 datasets, where the molecule trajectories are simulated using the same parameters as for the 2D Gaussian and Airy profiles, and the observations are generated according to the Born and Wolf model with parameters as in \cref{ex:smm:obs} (i.e. $z_0 = 1$ $\mu$m). Then, the limit of accuracy for mean photon counts ranging from 10 to 1250 is computed, and the results are displayed in \cref{fig:fim:bw}. Once again, there is an inverse square root decay of the limit of accuracy with respect to the mean photon count for all hyperparameters considered.
\begin{figure}[htbp!]
	\centering
	\includegraphics[width=\textwidth]{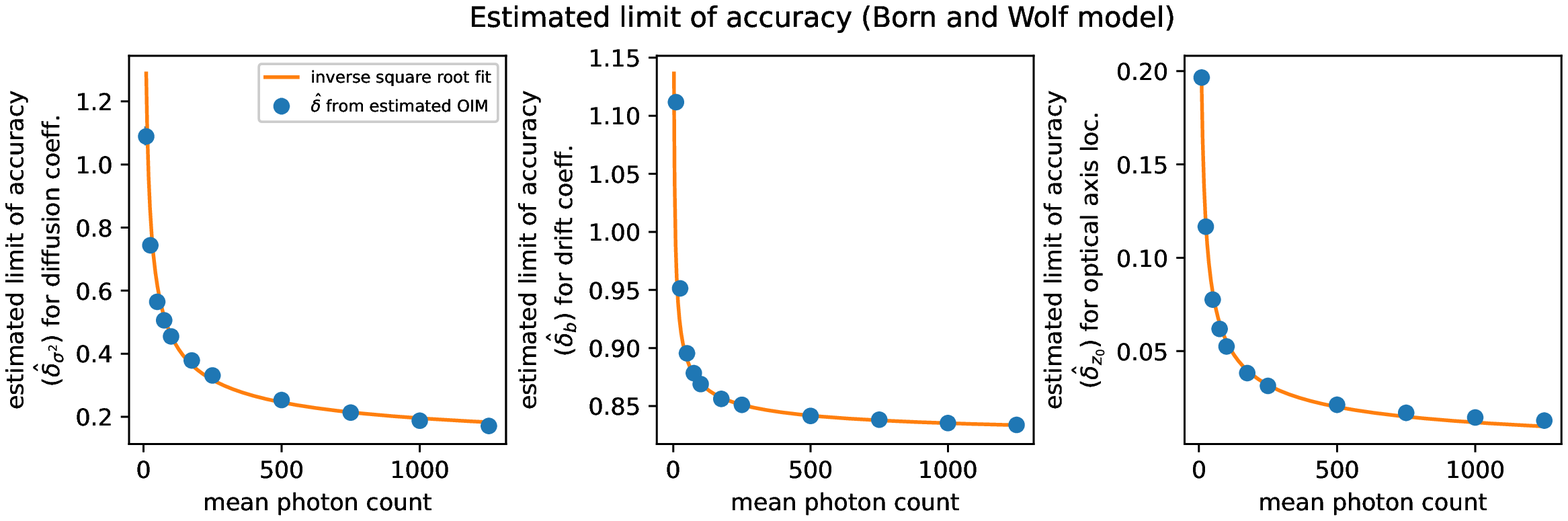} 
	\caption[Estimated limit accuracy for the diffusion ($\sigma^2$), drift ($b$) coefficients and optical axis location ($z_0$) for an out of focus molecule with stochastic trajectory and photon detection locations described by the Born and Wolf model.]{Evolution of the estimated limit accuracy for mean photon counts ranging from 10 to 1250. The limit of accuracy is estimated for the diffusion ($\sigma^2$), drift ($b$) coefficients and optical axis location ($z_0$) for an out-of-focus molecule with stochastic trajectory. The estimates are obtained by taking the square root of the inverse of the FIM, obtained by estimating the OIM using the SMC-FS algorithm with 500 particles for 2040 simulated datasets. To generate each dataset, the molecule trajectories are simulated according to the SDE in \cref{ex:mol:1} for the interval $[0, 0.2]$ seconds, with $\sigma^2=1$ $\mu$m$^2/$s, $b=-10$ s$^{-1}$, and initial location Gaussian distributed with mean $x_0=(5.5, 5.5)^\intercal$ $\mu$m and covariance $P_0 = 10^{-2}\mathbb{I}_{2\times2}$ $\mu$m$^2$. The observations are generated according to the Born and Wolf model with parameters as in \cref{ex:smm:obs}, where $z_0=1$ $\mu$m. An inverse square root curve (\textit{orange}) is fitted to the resulting estimated limits of accuracy for comparison.}
	\label{fig:fim:bw}
\end{figure}
\subsection{Limit of accuracy of the separation distance between two molecules for the Airy profile} \label{sec:sep_dist}

Being able to estimate the distance of separation between two closely spaced molecules is an important aspect of single-molecule microscopy. In the past, Rayleigh's criterion \cite{born2013principles} has been used to define the minimum distance between two point sources such that they can be distinguished in the image. However, \cite{ram2006beyond} treated the separation distance problem as a statistical estimation task and derived the CRLB (or inverse of the FIM) for the mean square error of the separation distance estimate. It was shown that Rayleigh's minimum distance can be surpassed by capturing more photons, e.g. by observing the molecules for a longer period. So far, the limit of accuracy has only been derived for static molecules. In this experiment, we apply our methodology to estimate the limit of accuracy for the locations and separation distance between two molecules that are not static, but diffusing independently at their respective stationary distributions, as illustrated in \cref{fig:smm:sepdist:ov}.\par
\begin{figure}[htbp!]
	\centering
	\begin{subfigure}{.35\textwidth}
		\centering
		\includegraphics[width=\textwidth]{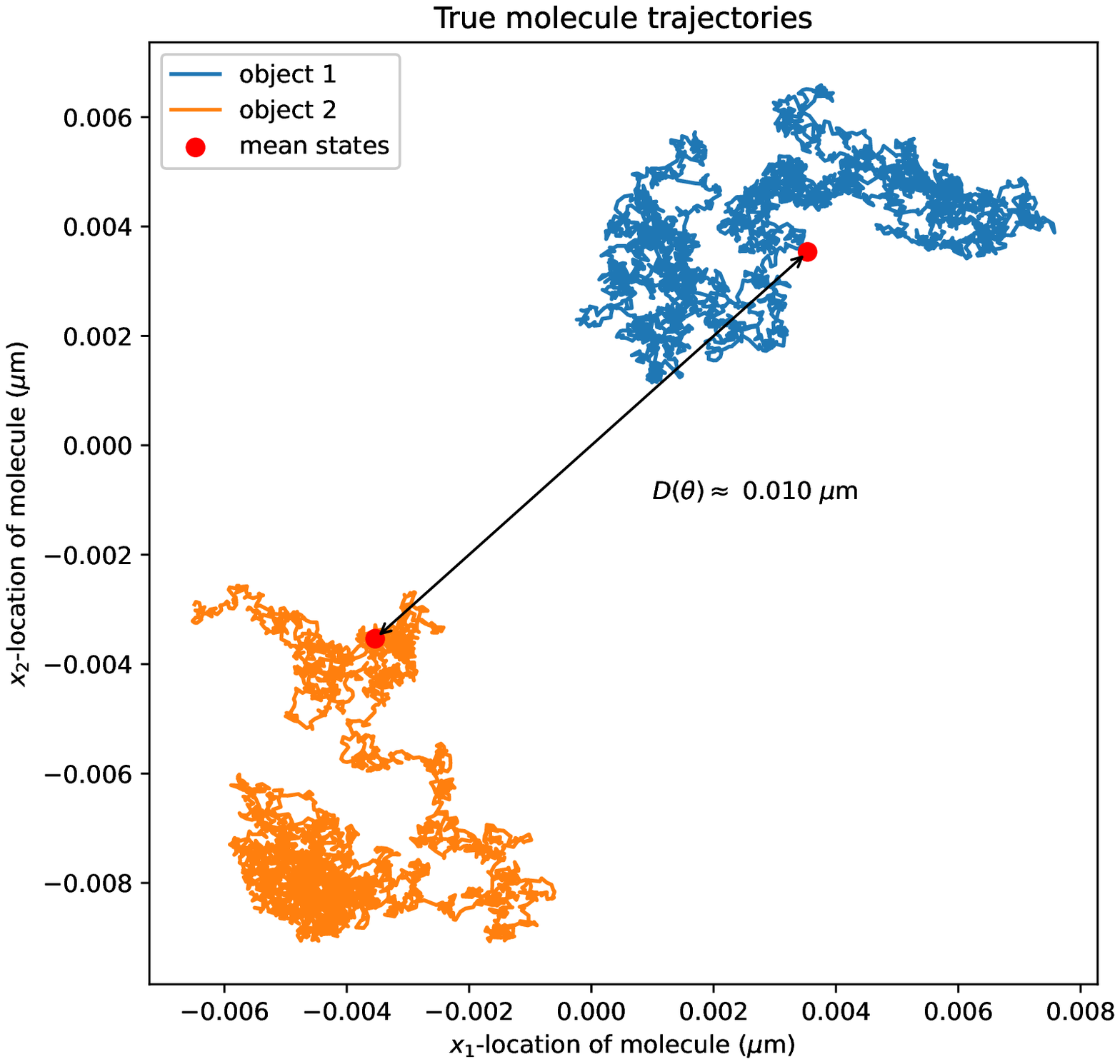} 
		\caption{}
	\end{subfigure}
	\begin{subfigure}{.3395\textwidth}
		\centering
		\includegraphics[width=\textwidth]{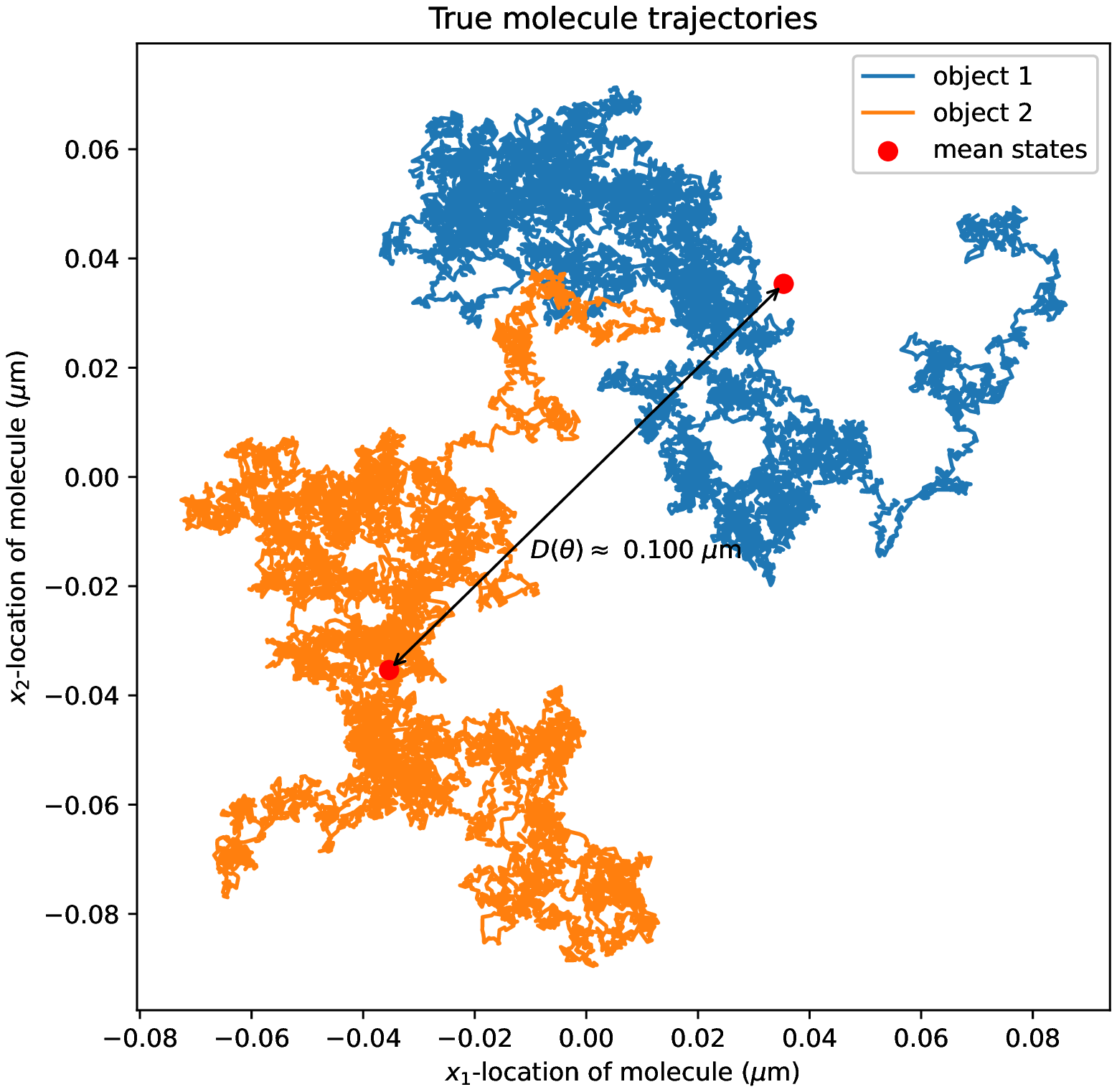} 
		\caption{}
	\end{subfigure}
	\caption[Two molecules diffusing independently.]{Examples of two molecules diffusing independently at a mean separation distance of \textbf{(a)} $0.01$ $\mu$m with diffusion coefficient $\sigma^2 = 10^{-4}$ $\mu$m$^2$/s \textbf{(b)} $0.1$ $\mu$m with $\sigma^2 = 10^{-3}$ $\mu$m$^2$/s. For an Airy distributed photon detection profile with $n_{\alpha}=1.4$ and $\lambda_e = 0.52$ $\mu$m, Rayleigh's resolution limit is $\approx$ $0.227$ $\mu$m. Increasing the value of the diffusion coefficient $\sigma^2$ will often lead to the molecule trajectories overlapping.}
	\label{fig:smm:sepdist:ov}
\end{figure}
Let $X_{t}=(X_{t, 1},X_{t, 2})^{\intercal}$ be the cartesian coordinates of a moving molecule with stationary distribution $\mathcal{N}(x_0,\sigma^2 \mathbb{I}_{2\times2})$
for all $t$, where $x_0$ is referred to as the \textit{mean state}. The continuous time dynamics are given by 
\begin{equation} \label{eq:smm:sepdist:x}
	\text{d}X_{t}  = (x_0-X_{t})\text{d}t + \sqrt{2}\sigma\text{d}B_{t}.
\end{equation}
From \cref{subsec:x}, it is straightforward to establish the solution to this SDE, which yields the conditional pdf $f^{x_0}_{\Delta}$ of $X_{k+1}$ at the $(k+1)$-th discrete segment, given $X_{k}=x$ at the $k$-th segment, as
\begin{equation*}
	X_{k+1} | (X_{k} = x) = \Phi_\Delta x + a_\Delta  + W_x, \quad W_x \sim \mathcal{N}(0, R_{\Delta}),
\end{equation*}
where $\Phi_\Delta = e^{-\Delta}$, $a_\Delta = x_0(1-e^{-\Delta})$ and $R_\Delta = \sigma^2 (1-e^{-2\Delta})\mathbb{I}_{2\times2}$. \par
In this experiment, consider two independently diffusing molecules whose states are $\left(X_{t},V_{t}\right)$, where $X_{t}$ is the state of the first molecule and $V_{t}$ is the state of the second. Assume that the initial state of each molecule is the same as its corresponding mean state, i.e.  $\left(x_{0},v_{0}\right)=:\theta=(\theta_{1},\theta_{2},\theta_{3},\theta_{4})^{\intercal}$, and is non-random but unknown and to be estimated. The conditional probability density function of $\left(X_{k+1}, V_{k+1}\right)$ given $\left(X_{k}, V_{k}\right)=\left(x_k,v_k\right)$ is $f^{x_0}_{\Delta}(x_{k+1}\vert x_k)f^{v_0}_{\Delta}(v_{k+1}\vert v_k)$ owing to their independent motions. \par
Let $\hat{\theta}=\left(\hat{\theta}_{1}(Y_{1:n}),\hat{\theta}_{2}(Y_{1:n}),\hat{\theta}_{3}(Y_{1:n}),\hat{\theta}_{4}(Y_{1:n})\right)^{\intercal}$
denote an estimate of $\theta$ given observations $Y_{1:n}$. Recall that the FIM, denoted $\mathcal{I}_{n}(\theta)$, is given by 
\begin{equation*}
	\mathcal{I}_{n}(\theta)=\mathbb{E}\left[ \nabla\log p_\theta(Y_{1:n})\;\nabla\log p_\theta(Y_{1:n})^{T}\right].
\end{equation*}
For any scalar-valued function $D(\theta)\in\mathbb{R}$, we can estimate $D(\theta)$ using $D(\hat{\theta})$ where $\hat{\theta}$ is the estimate of $\theta$. Assuming the estimate is unbiased, we have the following CRLB for the function $D$, 
\begin{equation} \label{eq:smm:sepdist}
	\mathbb{E}\left[ \left(D(\hat{\theta})-D(\theta)\right)^{2}\right] \geq\nabla D(\theta)^{\intercal}\mathcal{I}_{n}(\theta)^{-1}\nabla D(\theta),
\end{equation}
where $\nabla D(\theta):=(\partial D/\partial\theta_{1},\ldots,\partial D/\partial\theta_{4})^{\intercal}$. For
example, to estimate the separation between the two molecules we have
$D(\theta)=\sqrt{\left(\theta_{1}-\theta_{3}\right)^{2}+\left(\theta_{2}-\theta_{4}\right)^{2}}$, and as a result
\begin{equation*}
	\nabla D(\theta)= \frac{1}{D(\theta)}\left(\begin{array}{c}
		\theta_1 - \theta_3\\
		\theta_2 - \theta_4\\
		- (\theta_1 - \theta_3)\\
		- (\theta_2 - \theta_4)
	\end{array}\right).
\end{equation*}
This experiment is essentially the dynamic version of the experiments on estimating the separation of two static molecules by \cite{ram2013stochastic}. The key difference here is that the molecules are diffusing.  The observations $Y_{1:n}$ are generated as in \cite{ram2013stochastic}, i.e. according to the following mixture
\begin{align} %
	G_k(x_k, v_k)  
	&= \begin{cases} \label{eq:smm:mixtg:s}
		1-\Delta \lambda_\theta, & \mathrm{if}\;y_{k}=\emptyset,\\
		\lambda_xg(y_k | x_k) + \lambda_v g(y_k | v_k) , & \mathrm{otherwise},
	\end{cases}
\end{align}
where $g$ is the photon distribution profile given in \eqref{eq:image} and $\lambda_\theta = \lambda_x + \lambda_v$. The measurement model considered in this experiment is the Airy profile \eqref{eq:airy}, but it is straightforward to also apply the methodology to the 2D Gaussian profile and Born and Wolf model. \par
In the first part of the experiment, we analytically replicate results similar to those in \cite{ram2006beyond, ram2013stochastic} for two static molecules, then observe how introducing diffusion affects the progression of the limit of accuracy $\delta_{D(\theta)}$ for the separation distance (obtained using \eqref{eq:smm:sepdist}), as this separation distance between the two molecules increases. We set $\lambda_x = \lambda_v = \lambda$ for simplicity.
Evaluating $\delta_{D(\theta)}$ analytically for the static case is performed as in \cite{ram2006beyond}, with a mean photon count, denoted $N_{phot}$, of $3000$. For the dynamic case, the molecules are observed during an interval of $[0, 1]$ seconds with the same mean photon count, and for diffusion coefficients $\sigma^2$ varying from $5\times10^{-3}$ to $10^{-4}$ $\mu$m$^2$/s. The parameters of the Airy profile are unchanged (i.e. $n_{\alpha}=1.4$, $\lambda_e = 0.52$ $\mu$m), as is the lateral magnification matrix ($M = 100\mathbb{I}_{2\times2}$). The estimate of the limit of accuracy is obtained by estimating the OIM for the mean locations $x_0$ and $v_0$ via the SMC-FS algorithm for 640 to 1024 datasets then applying \eqref{eq:smm:sepdist}.  The resulting estimated limits of accuracy $\hat{\delta}_{D(\theta)}$ are given in \cref{fig:smm:sepdist:sig}. 
The second part of the experiment involves similarly estimating the limits of accuracy $\delta_{D(\theta)}$ for various separation distances, but this time the diffusion coefficient remains fixed, i.e. $\sigma^2=10^{-4}$ $\mu$m$^2$/s, and the mean photon count $N_{phot}$ is set to vary between 100 and 4500. The resulting estimated limits of accuracy are given in \cref{fig:smm:sepdist:nphot}. \par 
As the separation distance $D(\theta)$ gets closer to zero, the limit of accuracy increases, indicating that estimates would become less accurate. Additionally, an inverse square root curve was fit to each set of estimated limits of accuracy in \cref{fig:smm:sepdist:sig} and \cref{fig:smm:sepdist:nphot}. This is consistent with results in \cite{ram2006beyond} that showed an inverse square root relationship between separation distance and $\delta^{static}_{D(\theta)}$ for two static molecules, and indicates that these results can be generalised to dynamic molecules. 
Additionally, in \cite{ober2004localization}, it is suggested that the limit of accuracy for the location of a static molecule, known as \textit{localisation accuracy} and denoted $\delta^{loc}$, is of the form $\frac{\sigma_a}{\sqrt{N_{phot}}}$ where $N_{phot}$ is the mean photon count and $\sigma_a$ the standard deviation of the photon detection profile. 
The interpretation for this is that the quality of location estimates of a single static molecule deteriorates as the measurement uncertainty $\sigma_a$ increases. Now in \cite{ram2013stochastic}, it is proven that the limit of accuracy for the separation distance between two molecules $\delta^{static}_{D(\theta)}$ and the localisation accuracy for each of these molecules are related as follows:
\begin{equation}
	\label{eq:smm:locd}
	H^{sta}_{N_{phot}}:=\lim_{D(\theta)\rightarrow\infty}\delta^{static}_{D(\theta)} = \sqrt{\left(\delta^{sta, loc}_{x_0}\right)^2+\left(\delta^{sta, loc}_{v_0}\right)^2},
\end{equation} 
where $\delta^{sta, loc}_{x_0}$ and $\delta^{sta, loc}_{v_0}$ denote the localisation accuracy for the first and second (static) molecule observed independently with cumulative mean photon count $N_{phot}$, respectively. Even though the separation distance goes to infinity, its limit of accuracy $\delta_{D(\theta)}$ remains finite. This means that as $D(\theta)\rightarrow\infty$, evaluating the limit of accuracy for the separation distance between two (static) molecules becomes equivalent to two independent localisation accuracy problems. It also means that $\delta^{static}_{D(\theta)}$ is similarly affected by measurement uncertainty $\sigma_a$ as are the localisation accuracies for the two molecules. \par
In this experiment, the introduction of diffusion negatively affects the improvement in estimation accuracy as the mean distance of separation between the two molecules increases. This is evidenced in \cref{fig:smm:sepdist:sig} by the more and more slowly decaying limits of accuracy as the value of $\sigma^2$ increases, and in \cref{fig:smm:sig:comp} by the linearly increasing trend in $\hat{\delta}_{D(\theta)}$ for all values of $D(\theta)$ as $\sigma$ increases. As a result, the diffusion coefficient in the dynamic model can be translated into additional observation uncertainty which affects $\delta_{D(\theta)}$ in a way reminiscent of how $\sigma_a$ affects $\delta^{static}_{D(\theta)}$. More generally, from our numerical results, we observe the relationship for our dynamic application behaves qualitatively as 
\begin{equation*}
	\sqrt{\frac{\sigma_a^2+ \sigma^2}{N_{phot}}},
\end{equation*}
where, as above, $\sigma_a$ is the standard deviation of the photon detection process, also known as measurement uncertainty. \par
We now investigate the relationship between $\delta_{D(\theta)}$ and the dynamic equivalent to the localisation accuracy, namely the limit of accuracy for the mean locations $x_0$ and $v_0$ of each individual, stochastically moving molecule, denoted $\delta^{sto, loc}_{x_0}$ and $\delta^{sto, loc}_{v_0}$, respectively. The limits $\delta^{sto, loc}_{x_0}$ and $\delta^{sto, loc}_{v_0}$ can be estimated independently by repeatedly taking the mean estimated OIM for $x_0$ and $v_0$ based on two separate sets of 640 simulated datasets (one for each molecule) for mean photon counts ranging from 50 to 2250 (half of $N_{phot}$ each, given we have $\lambda_x = \lambda_v = \lambda$ under current settings). The distance $$H^{sto}_{N_{phot}}:=\sqrt{\left(\delta^{sto, loc}_{x_0}\right)^2+\left(\delta^{sto, loc}_{v_0}\right)^2}$$ between the limits of accuracy $\delta^{sto, loc}_{x_0}$ and $\delta^{sto, loc}_{v_0}$ of each individual object with various (cumulative) mean photon counts $N_{phot}$ is illustrated as horizontal lines in \cref{fig:smm:sepdist:nphot}, which appear to act as asymptotes, thus indicating that the relationship in \eqref{eq:smm:locd} can be generalised to stochastically moving molecules. While the introduction of diffusion leads to less accurate estimates, \cref{fig:smm:sepdist:nphot} displays a stronger decay in the limit of accuracy as the mean photon count $N_{phot}$ increases, thus indicating that increasing the mean photon count $N_{phot}$ improves those estimates, as was the case for static molecules in \cite{ram2006beyond}. This is reinforced in \cref{fig:smm:nphot:comp}, which also suggests that the relationship between $\delta_{D(\theta)}$ and $N_{phot}$ is an inverse square root. This is also a generalisation to the dynamic case of results in \cite{ram2006beyond} which showed an inverse square root relationship between $\delta_{D(\theta)}^{static}$ and $N_{phot}$ for two static molecules. \par
In summary, this experiment employs the numerical framework developed in this paper for estimating the FIM of parameters of dynamic molecules using SMC in order to gain insights into generalising results from \cite{ram2006stochastic, ram2013stochastic} about the effects of separation distance, measurement uncertainty and mean photon count to a context in which the two molecules considered follow a SDE rather than being static. These effects, as well at that of the measurement uncertainty, can all be observed by applying our methodology and are summarised in \cref{tab:smm:dynamic}. We also summarise in \cref{tab:smm:dynamic} the results on the limits of accuracy for the drift and diffusion coefficients of a single stochastically moving molecule observed via the 2D Gaussian, Airy profiles and the Born and Wolf model from \cref{sec:fim:gaussian-airy} and \cref{sec:fim:bw}. 
Note that the limits of accuracy for the mean locations of each molecule, denoted $\delta_\theta :=(\delta_{\theta_1}, \delta_{\theta_2}, \delta_{\theta_3}, \delta_{\theta_4})^\intercal$, can also be estimated as part of our methodology (as their FIM is required for \eqref{eq:smm:sepdist}) and return similar relationships with separation distance, mean photon count, diffusion coefficient and measurement uncertainty as $\delta_{D(\theta)}$ (not reported here).\par
In this section, results on the relationship between the limits of accuracy for various parameters and the mean photon count $N_{phot}$ have been extended from a single static \cite{ober2004localization, chao2016fisher, ober2020quantitativech19} or deterministically moving molecule \cite{wong2010limit} to a molecule whose trajectory follows an SDE. Additionally, insights have been gained into generalising results for the optical microscope resolution problem, which considers the separation distance between two static molecules \cite{ram2006beyond, ram2006stochastic}, to two stochastically diffusing molecules. The qualitative relationships observed and summarised in \cref{tab:smm:dynamic} are important in an experimental design context, as they provide information on how the accuracy of parameter estimates is affected by various experimental setups. For example, the $\mathcal{O}(N^{-1/2}_{phot})$ relationship between limits of accuracy and mean photon count indicates that quadrupling the number of photons can help halve the standard deviation of parameter estimates.
\begin{table}[htbp!]
	{\footnotesize
		\centering
		\begin{tabular}{cccc}	
			\toprule
			\textbf{Limit of accuracy $\delta$} & \multicolumn{2}{c}{\textbf{Qualitative Dependence}} & \textbf{Reference}\\
			$\delta_\vartheta = \text{std}(\hat{\vartheta})$& Parameter & Relationship & \\
			\midrule
			$\delta_{D(\theta)}, \delta_{\theta}$ & $D(\theta)$ separation distance  &  $\mathcal{O}\left(D(\theta)^{-1/2}\right)$ & \cref{fig:smm:sepdist}\\		
			$\delta_{D(\theta)}, \delta_{\theta}$ & $\sigma^2$ diffusion coefficient  &$\mathcal{O}\left(\sigma\right)$ & \cref{fig:smm:sig:comp} \\
			$\delta_{D(\theta)}, \delta_{\theta}$ & $\sigma_a^2$ measurement uncertainty  &$\mathcal{O}\left(\sigma_a\right)$ & \cite{ram2006beyond, ober2004localization} \\
			$\delta_{D(\theta)}, \delta_{\theta}$ & $N_{phot}$ mean photon count & $\mathcal{O}\left(N_{phot}^{-1/2}\right)$ & \cref{fig:smm:nphot:comp}\\
			$\delta_{\sigma^2}, \delta_{b}, \delta_{z_0}$ & $N_{phot}$ mean photon count  & $\mathcal{O}\left(N_{phot}^{-1/2}\right)$ & \cref{fig:fim:gaussian-airy}, \cref{fig:fim:bw} \\
			\bottomrule
		\end{tabular}
		\caption{Summary of the qualitative relationships between the limits of accuracy (or standard deviation of parameter estimates) $\delta_\theta :=(\delta_{\theta_1}, \delta_{\theta_2}, \delta_{\theta_3}, \delta_{\theta_4})^\intercal$ and $\delta_{D(\theta)}$ for the mean locations $\theta = (x_0, v_0) = (\theta_1, \theta_2, \theta_3, \theta_4)^\intercal$ and separation distance $D(\theta)$, respectively, of two stochastically diffusing molecules observed simultaneously. Also included in the table is the relationship between mean photon count and the limits of accuracy for the hyperparameters of the SDE trajectory (drift $b$ and diffusion $\sigma^2$ coefficients) and photon detection process (optical axis location $z_0$) of a single molecule. Note that when we increase the mean photon count $N_{phot}$, the observation interval length remains fixed.} \label{tab:smm:dynamic}
	}
\end{table}
\begin{figure}[htbp!]
	\centering
	\begin{subfigure}{.49\textwidth}
		\centering
		\includegraphics[width=\textwidth]{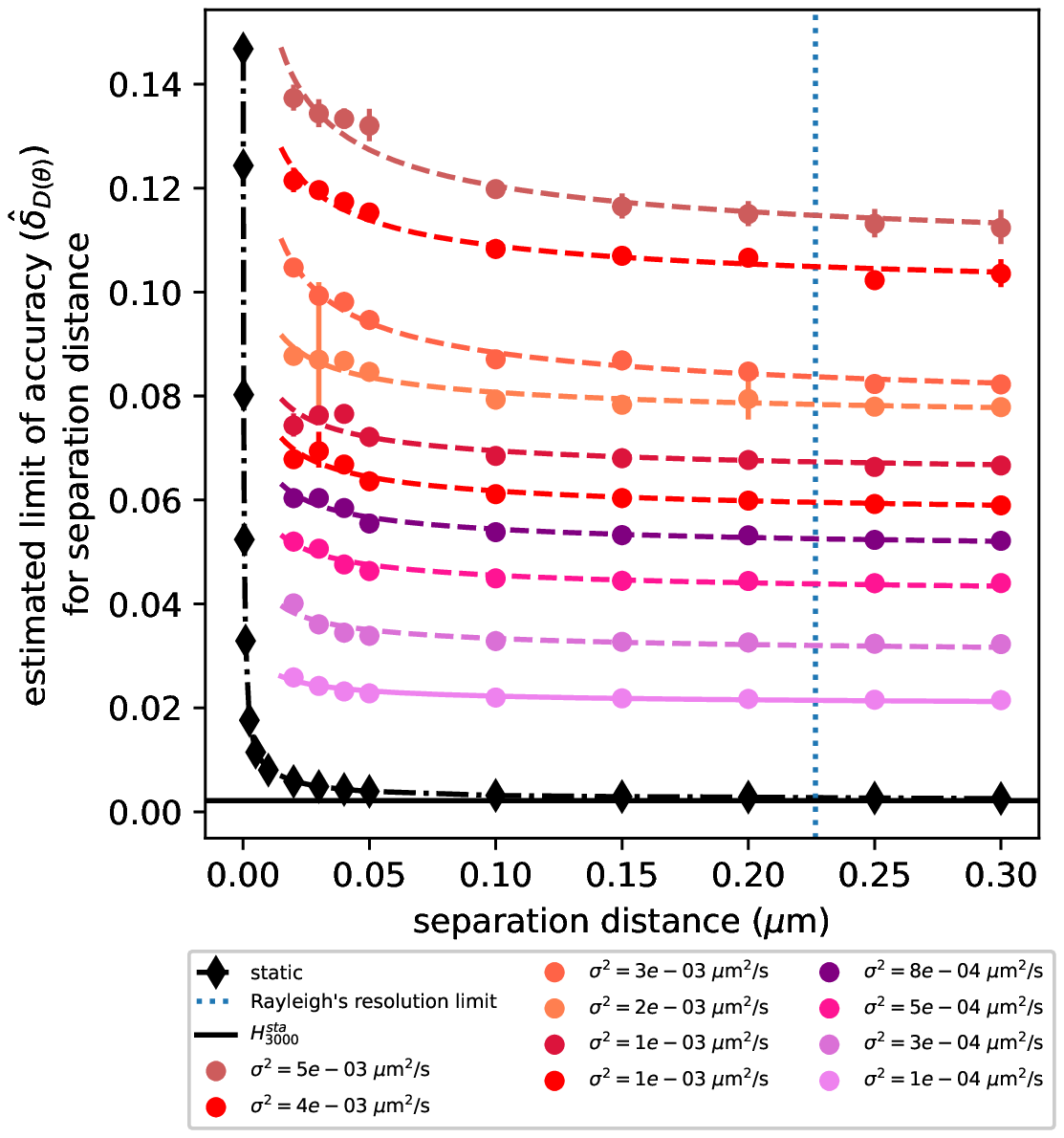} 
		\caption{}
		\label{fig:smm:sepdist:sig}
	\end{subfigure}
	\begin{subfigure}{.49\textwidth}
		\centering
		\includegraphics[width=\textwidth]{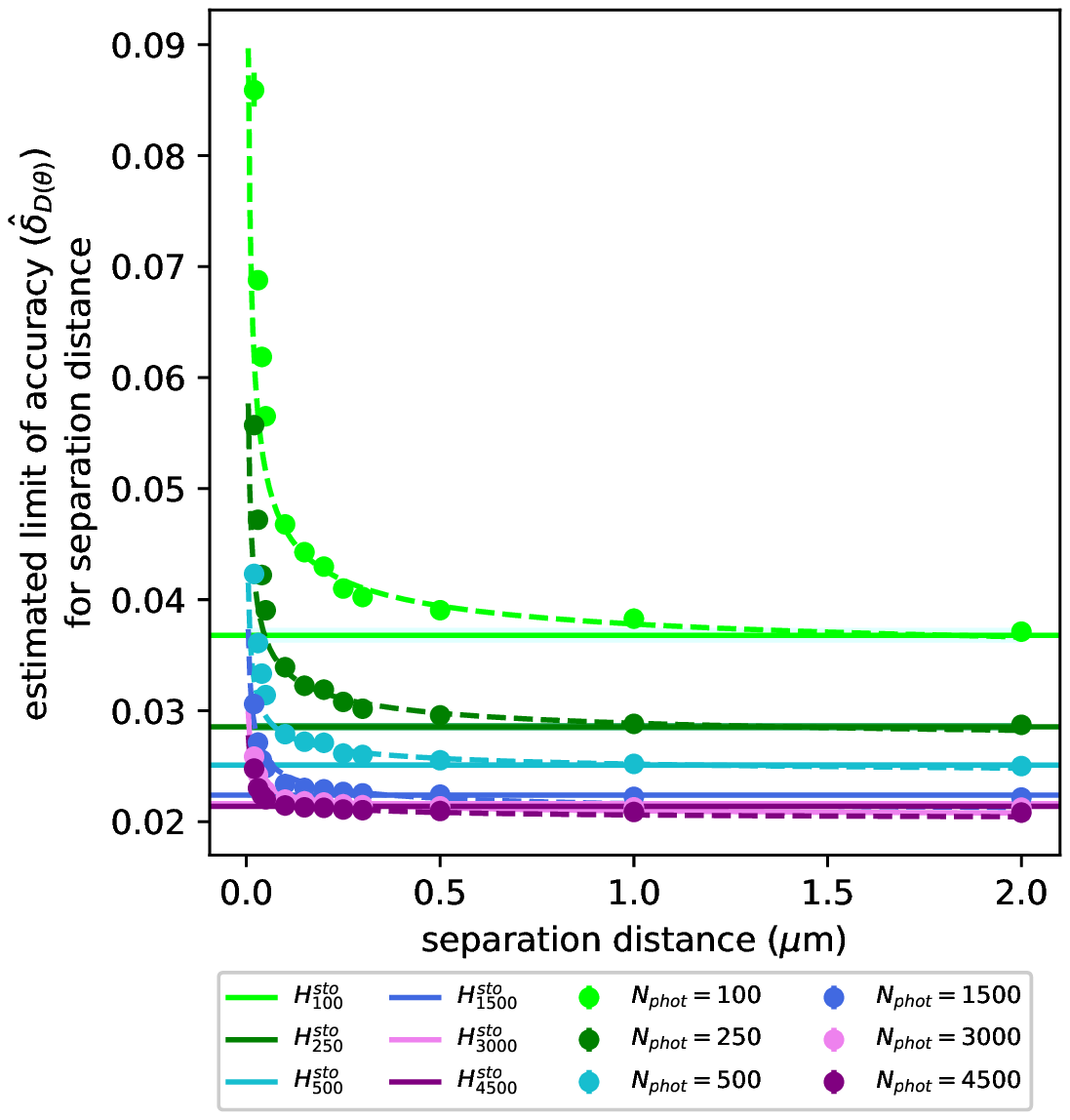} 
		\caption{}
		\label{fig:smm:sepdist:nphot}
	\end{subfigure}
	\caption[Estimated limit accuracy for the mean location ($\theta$) of two molecules with photon detection locations described by the Airy profile.]{\small Comparison of the evolution of the estimated limit accuracy for separation distances ranging from $20\times 10^{-3}$ to $2$ $\mu$m for various  \textbf{(a)} diffusion coefficient ($\sigma^2$) values \textbf{(b)} mean photon counts ($N_{phot}$). The limit of accuracy for the separation distance $\delta_{D(\theta)}$, where $\theta = (x_0, v_0) = (\theta_1, \theta_2, \theta_3, \theta_4)$, is estimated using the square root of the CRLB obtained using \eqref{eq:smm:sepdist} (in the dynamic case) and evaluated using analytical results from \cite{ram2006beyond} (in the static case). The estimates of $\mathcal{I}_n(\theta)$ in \eqref{eq:smm:sepdist} are obtained by running the SMC-FS algorithm with 500 particles for 640 to 1024 simulated datasets. For the dynamic case, the molecule trajectories are initialised at their respective mean locations $x_0$ and $v_0$ and each is propagated according to its corresponding SDE \eqref{eq:smm:sepdist:x} during an interval of $[0, 1]$ seconds with \textbf{(a)} fixed and mean photon count $N_{phot}=3000$ \textbf{(b)} fixed diffusion coefficient $\sigma^2 = 10^{-4}$ $\mu$m$^2$/s. The observations are generated according to a mixture of Airy profiles \eqref{eq:smm:mixtg:s} with parameters as in \cref{ex:smm:obs}. This is repeated for \textbf{(a)} $\sigma^2$ varying from $10^{-3}$ to $10^{-4}$ $\mu$m$^2$/s \textbf{(b)} $N_{phot}$ varying from $100$ to $4500$. Finally, an inverse square root curve is fitted to each of the resulting sets of estimated limits of accuracy for comparison purposes.  Note that the \textit{pink} set of estimates and their corresponding \textit{solid} fitted curve in \textbf{(a)} coincide with those in \textbf{(b)}. In \textbf{(b)}, the horizontal lines correspond to the equivalent mean photon counts and represent the distances $H^{sto}_{N_{phot}}$ between the limits of accuracy $\hat{\delta}^{sto, loc}_{x_0}$ and $\hat{\delta}^{sto, loc}_{v_0}$ for the mean locations $x_0$ and $v_0$ of each individual object, estimated independently for each molecule using the SMC-FS algorithm. Note that any variation in estimates for low separation distances is due to Monte Carlo error, and can be reduced by increasing the number of simulated datasets.}
	\label{fig:smm:sepdist}
\end{figure}
\begin{figure}[htbp!]
	\centering
	\begin{subfigure}{.39\textwidth}
		\centering
		\includegraphics[width=\textwidth]{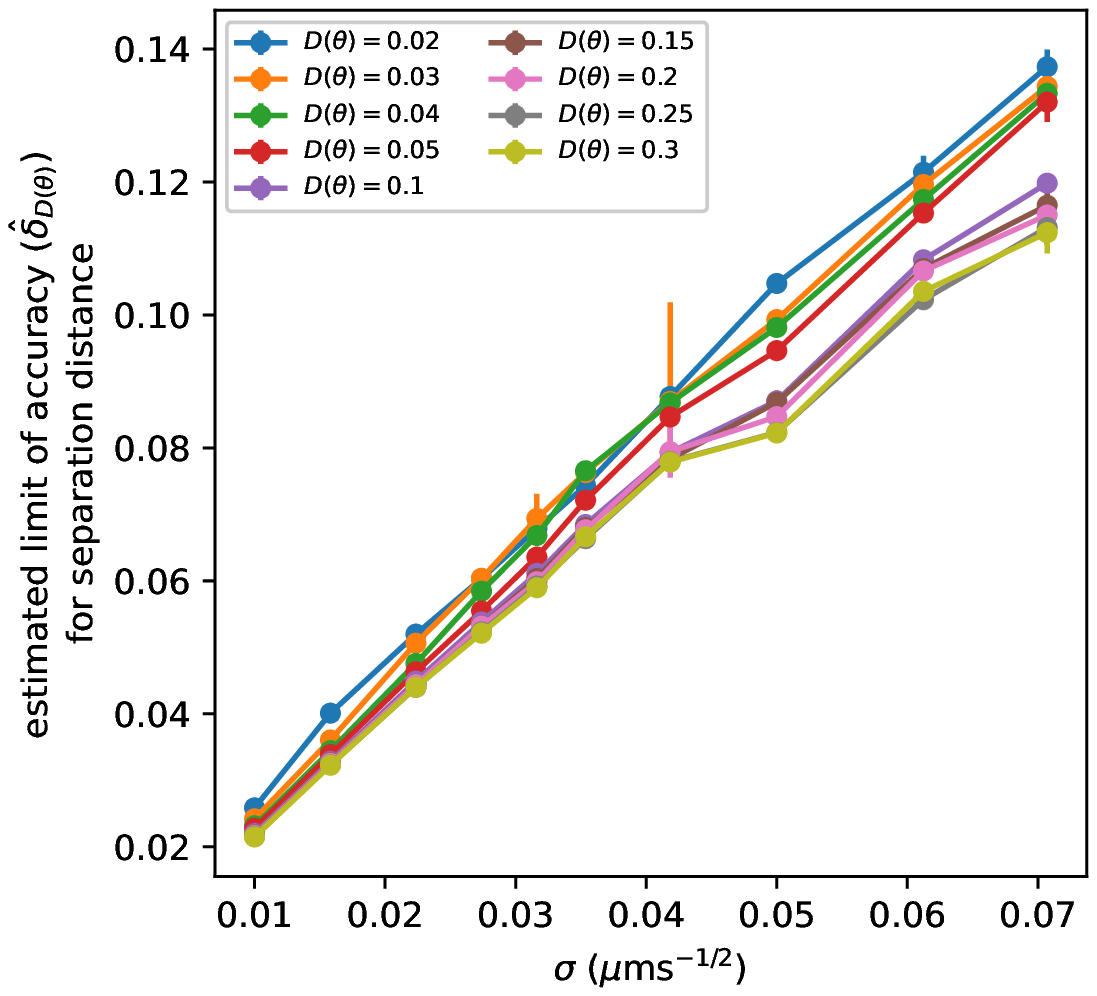} 
		\caption{}
		\label{fig:smm:sig:comp}
	\end{subfigure}
	\begin{subfigure}{.39\textwidth}
		\centering
		\includegraphics[width=\textwidth]{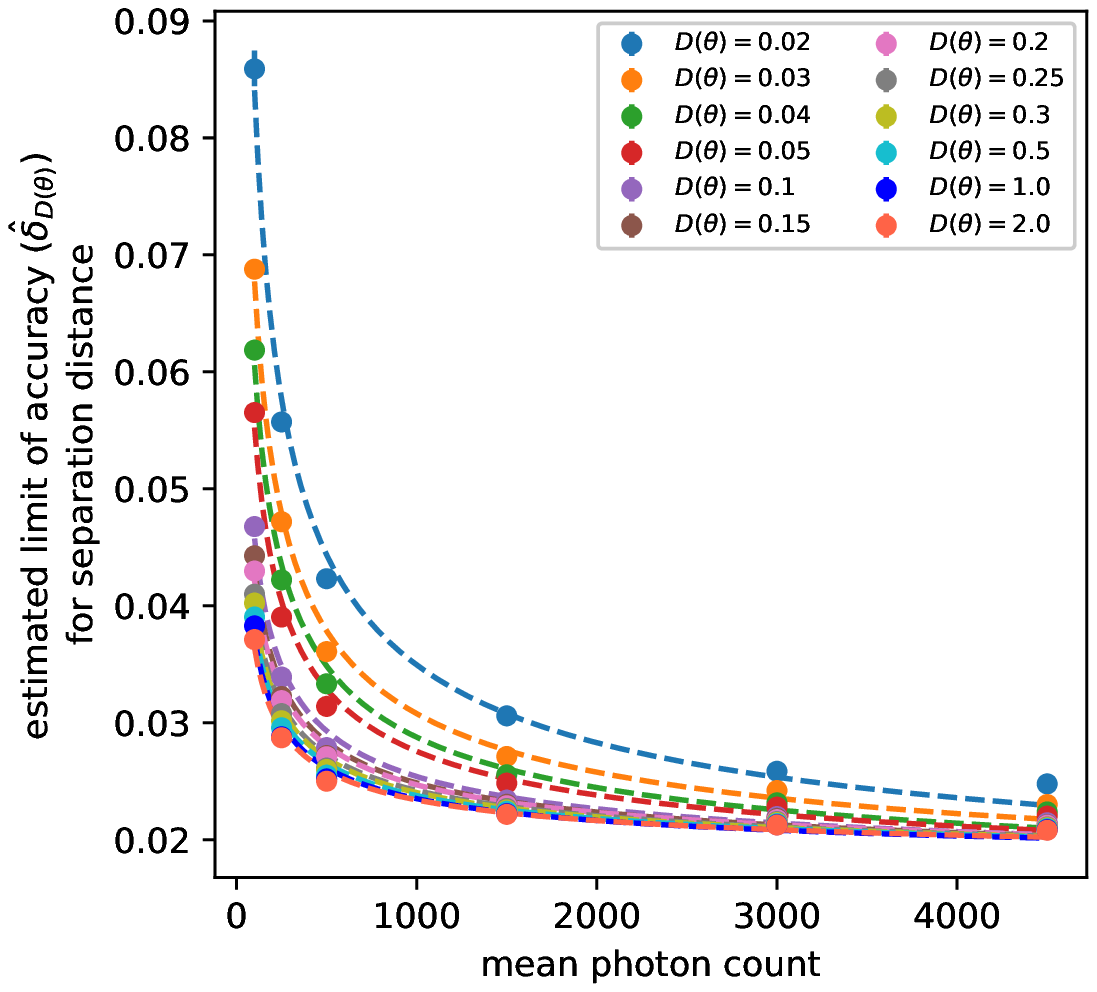} 
		\caption{}
		\label{fig:smm:nphot:comp}
	\end{subfigure}
	\caption[Evolution of the estimated limit of accuracy for the separation distance (obtained using \eqref{eq:smm:sepdist}) between two diffusing molecules for $\sigma$ ranging from $\sqrt{10^{-4}}$ to $\sqrt{5\times10^{-3}}$ $\mu$m s$^{-1/2}$.]{Evolution of the estimated limit of accuracy for the separation distance $\delta_{D(\theta)}$ (obtained using \eqref{eq:smm:sepdist}) between two stochastically moving molecules observed simultaneously for \textbf{(a)} $\sigma$ ranging from $\sqrt{10^{-4}}$ to $\sqrt{5\times10^{-3}}$ $\mu$m s$^{-1/2}$ \textbf{(b)} $N_{phot}$ ranging from 100 to 4500. Estimates are obtained through the same algorithm and parameters as in \textbf{(a)} \cref{fig:smm:sepdist:sig} \textbf{(b)} \cref{fig:smm:sepdist:nphot}, with separation distances ranging from $20\times 10^{-3}$ to $2$ $\mu$m. In \textbf{(b)}, inverse square root curves are fitted to the resulting estimates $\hat{\delta}_{D(\theta)}$ for comparison.}
	\label{fig:smm:comp}
\end{figure}
\section{Conclusion} \label{sec:smm:conc}
In this paper, we introduced an SMC approach to performing parameter inference when tracking a molecule with stochastic trajectory for a fixed time interval. The three main aspects of this fundamental model in single-molecule microscopy were the true location of the molecule in the object space, which follows a linear SDE, the Poisson distributed arrival process of the photons it emits on the detector in the image space, and the arrival location of those photons on the detector, which follows either a 2D Gaussian, Airy profile, or Born and Wolf model. \par 
First of all, we discretised the time interval in order to formulate the problem as a discrete-time state space model, in which all states are equally spaced in time, but a number of observations are marked as missing. From this, SMC methods were applied for parameter inference. A general forward smoothing algorithm was employed to estimate the score and OIM of the data regardless of the distribution of the photon locations. For the first time, this allowed for the estimation of the FIM and hence the limit of accuracy (square root of the CRLB), which could not be done before for the Airy profile and Born and Wolf model, and could only be achieved analytically for a specific set of photon detection times for the 2D Gaussian profile. The methodology was subsequently applied to characterise the precision limits for estimating the separation distance between two moving molecules, thus providing new insights into results for the static case from \cite{ram2013stochastic}. The outcome of our numerical work was summarised in \cref{tab:smm:dynamic}, which sums up the qualitative behaviours of the limits of accuracy  as functions of the mean photon count, separation distance, diffusion coefficient and measurement uncertainty. \par
Although for the first time a method has been described to estimate the limit of accuracy for the hyperparameters of dynamic single molecules with non-uniform observation times and complex measurement models, such as the Airy profile or Born and Wolf model, there is scope to use the techniques developed here to provide a wider range of more computationally efficient approaches. Indeed, an advantage of the straightforward state space model formulation of the problem is access to the vast range of filtering and smoothing algorithms available. While we employed forward smoothing, any kind of particle smoothing algorithm would be suitable, and indeed, the SMC-FS algorithm of \cite{del2010forward} employed for forward smoothing, even though it mitigates issues related to path degeneracy, is of $\mathcal{O}(N^2)$ complexity. For example, the PaRIS algorithm of \cite{olsson2017efficient} can reduce the complexity of the algorithm to linear. 
\bibliography{references}

\begin{appendices}
	\crefalias{section}{appendix}
\section{Validity of the time discretisation} \label{sec:discretisation}
Given a realisation of the observations $(t_1, y_{t_1}), \ldots, (t_n, y_{t_n})$ observed in the time interval $[t_0,T]$, we adopt a discrete time formulation in our methodology where $[t_0,T]$ is divided into segments of length $\Delta$. We assume the discretisation is fine enough so that an interval $(k\Delta,k\Delta+\Delta]$ contains at most one arrival time $t_{i}$. We now prove in \cref{prop:discret} that this discretisation is a valid approximation of the homogeneous Poisson process which is used in \cite{ober2004localization, ober2020quantitative} to describe photon detection.
\begin{proposition} \label{prop:discret}
	Let the photon detection process $\left\{N(t)\right\}_{t\geq t_0}$ be a homogeneous Poisson process with photon detection rate $\lambda>0$. The probability of observing $k$ photons during the time interval $[t, t+h]$ for $t\geq t_0$ and $h>0$ is 
	\begin{equation} \label{eq:poisson}
		\mathbb{P}\left(N(t+h)-N(t) = k\right) = \frac{\exp\left[-\lambda h\right](\lambda h)^k}{k!}.
	\end{equation}
	Discretise the interval $[t, t+h]$ into segments of length $\Delta$, so that a single segment contains at most one arrival time. Then as $\Delta \rightarrow 0$, the probability of observing $k$ is also given by \eqref{eq:poisson}.
\end{proposition}
\begin{proof}
	Discretising the interval into segments of length $\Delta$ as such, we now have for this interval a Binomial random variable $M_{\Delta} \sim \text{Binomial}(\left\lceil\frac{h}{\Delta}\right\rceil, \lambda \Delta)$ with $\left\lceil\frac{h}{\Delta}\right\rceil$ trials, with probability of success (i.e. a photon is observed) $\lambda \Delta$ and probability of failure (i.e. no photon is observed) $1-\lambda \Delta$. The probability observing $k$ photons in the interval $[t, t+h]$ is
	\begin{equation*} 
		\mathbb{P}\left(M_{\Delta}  = k\right) = \begin{pmatrix}
			\left\lceil \frac{h}{\Delta}\right\rceil \\ k
		\end{pmatrix}
		(\lambda \Delta)^{k} (1-\lambda \Delta)^{\left\lceil\frac{h}{\Delta}\right\rceil-k}.
	\end{equation*}
	Taking the limit as $\Delta \rightarrow 0$, we have
	\begin{align*}
		\lim_{\Delta \rightarrow 0} \mathbb{P}\left(M_{\Delta}  = k\right)
		&= \lim_{\Delta \rightarrow 0} \frac{\left\lceil\frac{h}{\Delta}\right\rceil!}{\left(\left\lceil\frac{h}{\Delta}\right\rceil-k\right)!k!}
		(\lambda \Delta)^{k} (1-\lambda \Delta)^{\left\lceil\frac{h}{\Delta}\right\rceil-k} \nonumber \\
		&= \lim_{\Delta \rightarrow 0} \frac{\left\lceil\frac{h}{\Delta}\right\rceil^{k} + O\left(\left\lceil\frac{h}{\Delta}\right\rceil^{k-1}\right)}{k!}
		(\lambda \Delta)^{k} (1-\lambda \Delta)^{\left\lceil\frac{h}{\Delta}\right\rceil-k} \nonumber \\
		&= \lim_{\Delta \rightarrow 0} \frac{\Delta^k\left\lceil\frac{h}{\Delta}\right\rceil^{k}}{k!}
		\lambda^{k} (1-\lambda \Delta)^{\left\lceil\frac{h}{\Delta}\right\rceil-k}. \nonumber
	\end{align*}
	Employing the following property of the ceiling function
	\begin{equation*}\label{eq:ceil}
		\frac{h}{\Delta} \leq\left\lceil \frac{h}{\Delta}\right\rceil < \frac{h}{\Delta}+1, 
	\end{equation*}
	we have
	\begin{align*}
		&\lim_{\Delta \rightarrow 0} \frac{\Delta^{k}\left(\frac{h}{\Delta}\right)^{k}}{k!}
		\lambda^{k} (1-\lambda \Delta)^{\frac{h}{\Delta}-k}
		&\leq \lim_{\Delta \rightarrow 0} \mathbb{P}\left(M_{\Delta}  = k\right) 
		&< \lim_{\Delta \rightarrow 0} \frac{\Delta^{k}\left(\frac{h}{\Delta}+1\right)^{k}}{k!}
		\lambda^{k} (1-\lambda \Delta)^{\frac{h}{\Delta}+1-k} \nonumber\\
		&\implies \lim_{\Delta \rightarrow 0} \frac{h^{k}}{k!}
		\lambda^{k} (1-\lambda \Delta)^{\frac{h}{\Delta}-k}
		&\leq \lim_{\Delta \rightarrow 0} \mathbb{P}\left(M_{\Delta}  = k\right)  
		&< \lim_{\Delta \rightarrow 0} \frac{\left(h+\Delta\right)^{k}}{k!}
		\lambda^{k} (1-\lambda \Delta)^{\frac{h}{\Delta}+1-k} \nonumber \\ 
		&\implies \lim_{\Delta \rightarrow 0} \frac{\left(\lambda h\right)^{k}}{k!}
		(1-\lambda \Delta)^{\frac{h}{\Delta}-k}
		&\leq \lim_{\Delta \rightarrow 0} \mathbb{P}\left(M_{\Delta}  = k\right) 
		&< \lim_{\Delta \rightarrow 0} \frac{\left(\lambda h\right)^{k}}{k!}
		(1-\lambda \Delta)^{\frac{h}{\Delta}-(k-1)}. \nonumber 
	\end{align*}
	Finally, employing the following results
	\begin{align*} 
		&\lim_{\Delta \rightarrow 0} (1-\lambda \Delta)^{\frac{h}{\Delta}} = \lim_{\frac{h}{\Delta} \rightarrow \infty} \left(1-\frac{\lambda h}{ \frac{h}{\Delta}}\right)^{\frac{h}{\Delta}} = \exp(-\lambda h), \\	
		&\lim_{\Delta \rightarrow 0} (1-\lambda \Delta)^{-k} = \lim_{\Delta \rightarrow 0} (1-\lambda \Delta)^{-(k-1)} = 1,
	\end{align*}
	yields the desired probability
	\begin{equation*} 
		\lim_{\Delta \rightarrow 0} \mathbb{P}\left(M_{\Delta}  = k\right)   = \frac{\left(\lambda h\right)^{k}}{k!} \exp(- \lambda h).
	\end{equation*}
\end{proof}
This paper mainly considers the situation in which $\lambda$ is a scalar, but this proof can be generalised to the situation where the Poisson process is inhomogeneous. As suggested in \cref{prop:discret}, the approximation of the Poisson process becomes increasingly more accurate as the discrete segment length $\Delta$ becomes smaller.
\section{Particle filtering in single-molecule microscopy}
\label{sec:supp:pf}
Given our reformulation of the fundamental model as a discrete state space model, a particle filter, summarised in \cref{alg:smc:PF}, can be applied to track the state of stochastically moving particles.
\begin{algorithm}[htbp]
	\caption{Particle filter}\label{alg:smc:PF}
	\begin{algorithmic}[1]
		\Statex \textit{Where $(i)$ appears, the operation is performed for all $i \in \{1, \ldots, N\}$.}
		\Statex Initialise at $k=1$: 
		\State Sample $X_1^{(i)} \sim \eta_1$ where $\eta_1$ is the initial user-defined proposal density.
		\State Initialise importance weights $w_1\left(X_1^{(i)}\right) = \frac{G^\theta_1\left(X_1^{(i)}\right)\nu_\theta\left(X_1^{(i)}\right)}{\eta_1\left(X_1^{(i)}\right)}$ and normalise to obtain $\omega_1^{(i)}$.
		\Statex Given weighted particle sample $\left(X_{k-1}^{(1:N)}, \omega_{k-1}^{(1:N)}\right)$,
		\For{$k=2, \ldots, n$} 
		\State (\textsc{Resample}) $\left(\iota_{k-1}^{(1:N)}, \omega_{k-1}^{(1:N)}\right) := \textsf{resample}\left(\omega_{k-1}^{(1:N)}\right)$. \label{alg:step:resamp}
		\State (\textsc{Propagate}) Sample $X^{(i)}_{k} \sim \eta_k\left(\cdot | X_{k-1}^{(\iota_{k-1}^{(i)})} \right)$, where $\eta_k$ is the proposal density.
		\State (\textsc{Weight}) Compute the incremental weights
		$$\tilde{w}\left(X^{(\iota_{k-1}^{(i)})}_{k-1}, X^{(i)}_{k}\right) = \frac{G^\theta_k\left(X^{(i)}_{k}\right)f^\theta_\Delta\left(X^{(i)}_{k}|X^{(\iota_{k-1}^{(i)})}_{k-1}\right)}{\eta_k\left(X^{(i)}_{k}|X^{(\iota_{k-1}^{(i)})}_{k-1}\right)},$$ then update and normalise the importance weights to obtain $\omega_k^{(i)}$. \label{alg:step:weight}
		\EndFor
	\end{algorithmic}
\end{algorithm}
There are several approaches to resampling, studied in \cite{douc2005comparison, DelMoral2012a, del2012adaptive}. In this paper, we refer to the resampling step of algorithms as
$$\left(\iota_k^{(1:N)}, \omega_k^{(1:N)}\right) := \textsf{resample}\left(\omega_{k}^{(1:N)}\right).$$
An example of particle filtering for for stochastically moving molecules observed through the 2D Gaussian, Airy and Born and Wolf models is available in \cref{ex:smm:pf}.
\begin{exmp} \label{ex:smm:pf}
	Let the trajectory of a molecule be given by the SDE in \cref{ex:mol:1} and simulated three times (one for each measurement model) using the same parameters and for the same time interval. Observations are generated as per in \cref{ex:smm:obs} for the 2D Gaussian, Airy profiles and Born and Wolf model and the molecules are tracked using the bootstrap filter. The resulting estimated trajectories for each measurement model are given in \cref{fig:fim:pf}.
	\begin{figure}[htbp!]
		\centering
		\begin{subfigure}{.32\textwidth}
			\centering
			\includegraphics[width=\textwidth]{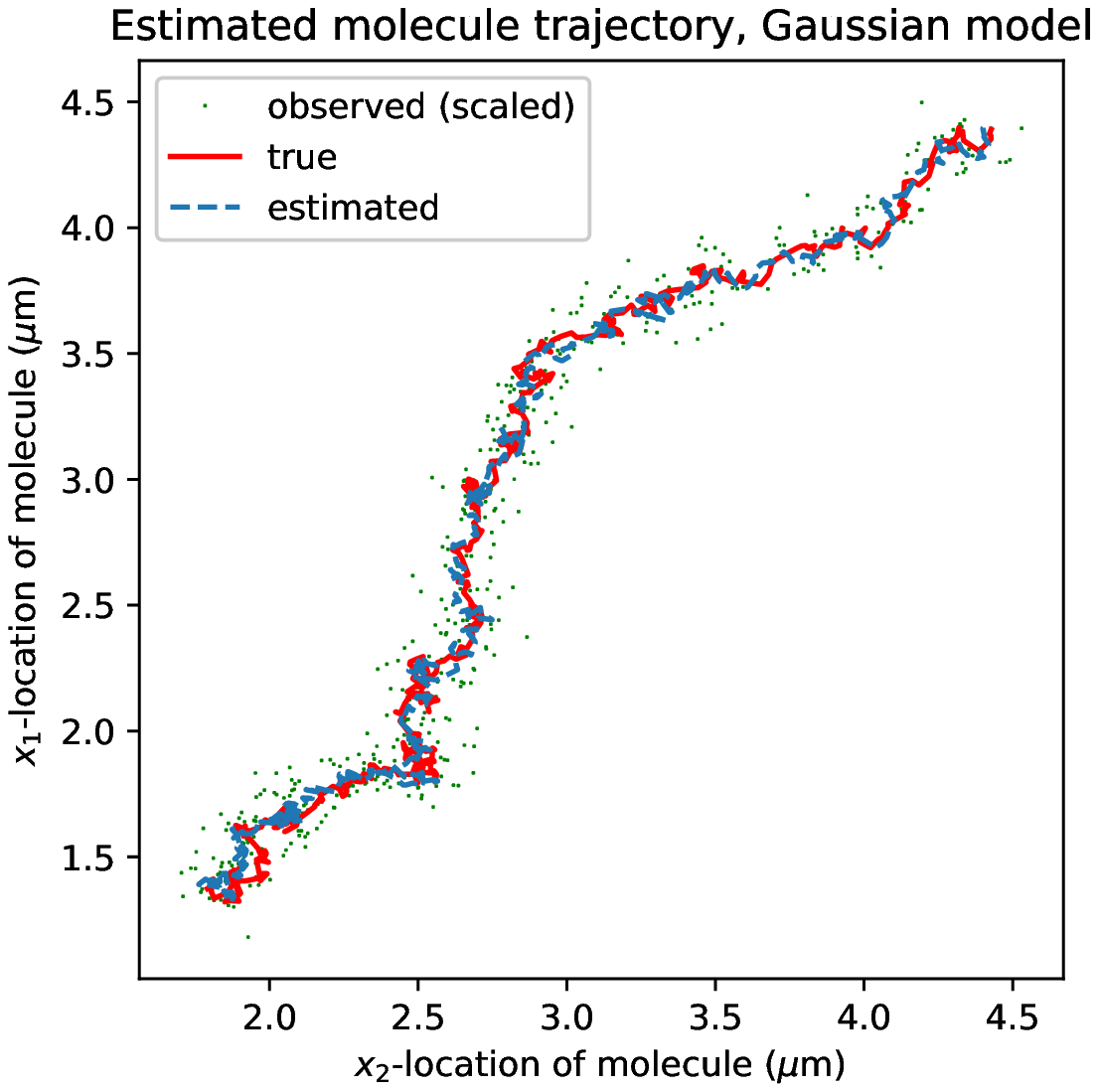} 
			\label{fig:sub-pf-gaussian}
		\end{subfigure}
		\begin{subfigure}{.32\textwidth}
			\centering
			\includegraphics[width=\textwidth]{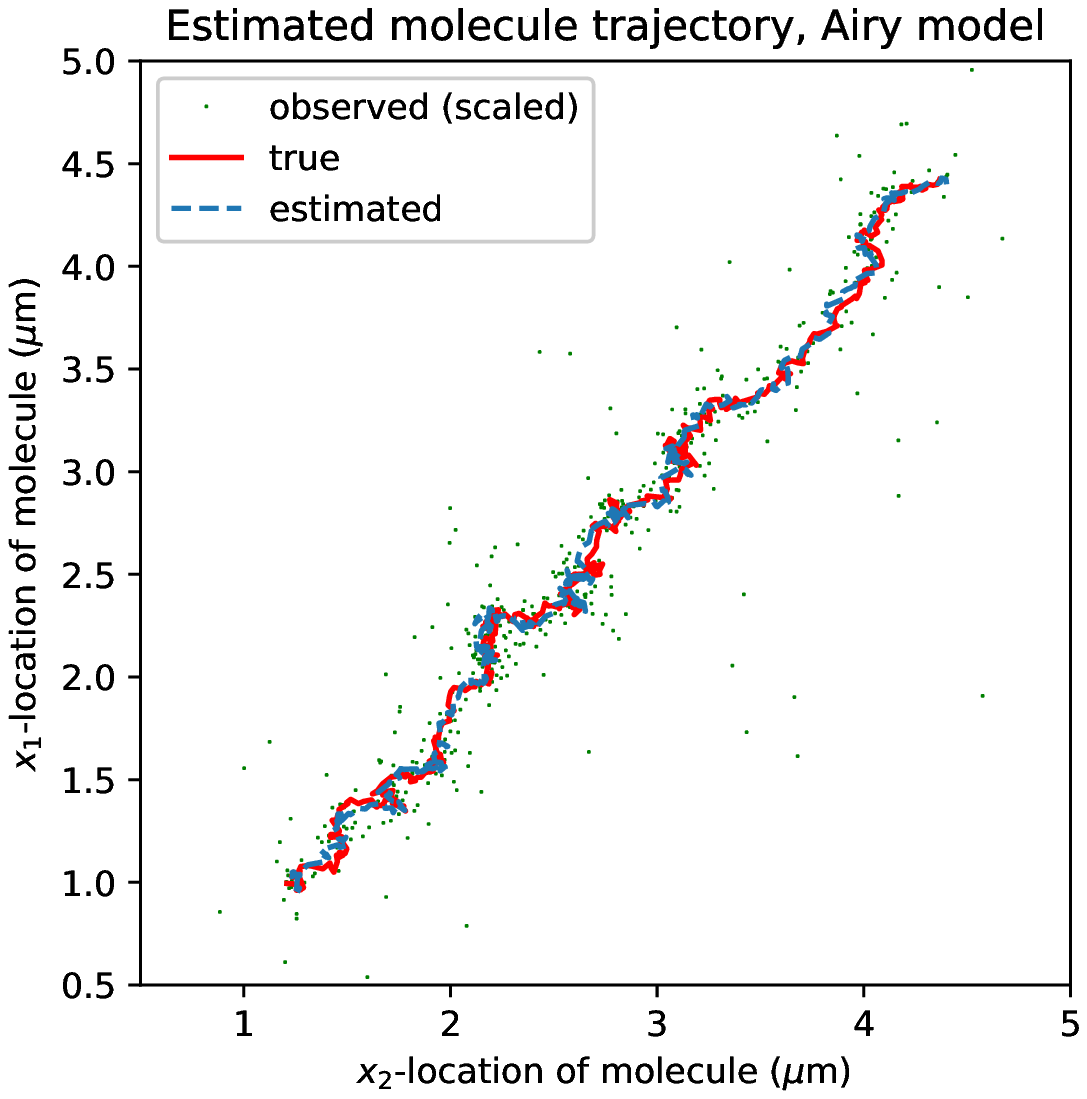} 
			\label{fig:sub-pf-airy}
		\end{subfigure}
		\begin{subfigure}{.34\textwidth}
			\centering
			\includegraphics[width=\textwidth]{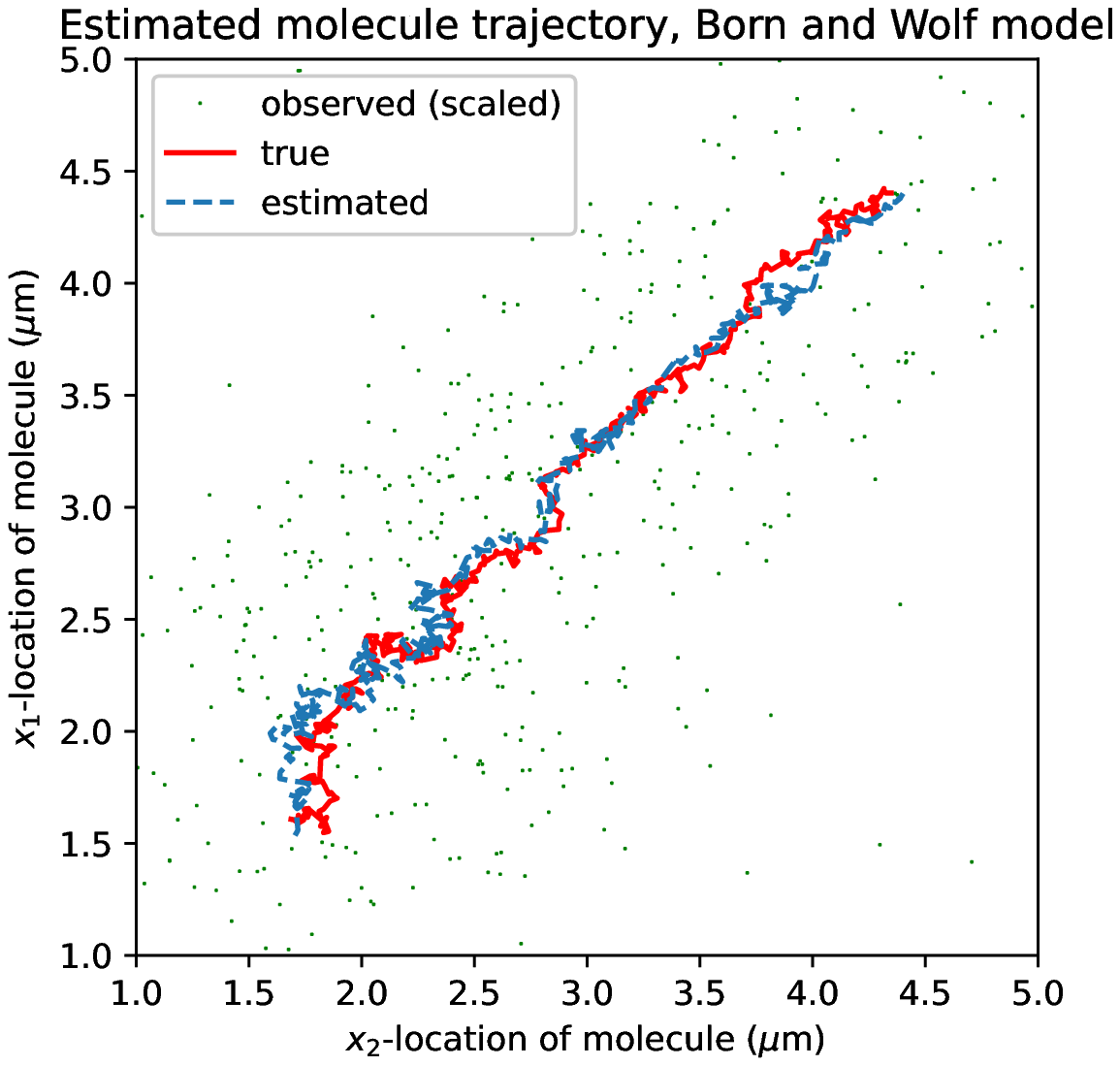} 
			\label{fig:sub-pf-bw}
		\end{subfigure}
		\caption[Estimated molecule trajectories for the 2D Gaussian, Airy profiles and Born and Wolf model.]{Estimated molecule trajectories for the 2D Gaussian (\textit{left}), Airy (\textit{middle}) profiles and Born and Wolf model (\textit{right}). The scaled observations (i.e. divided by $m$) for each measurement model are displayed for information.}
		\label{fig:fim:pf}
	\end{figure}
\end{exmp}
\section{Sufficient statistic for estimating the OIM by forward smoothing} \label{app:smm:suffstat}
In \cref{ex:score}, recall that the molecule trajectory is described by the following SDE in $d$-dimensional space
\begin{equation*}
	\text{d}X_t = b\mathbb{I}_{d\times d}X_t\text{d}t + \sqrt{2}\sigma\text{d}B_t,
\end{equation*}
where in the drift term, $b\neq 0$, in the diffusion term, $\sigma>0$, and $(\text{d}B_t)_{t_0\leq t\leq T}$ is a Wiener process. The log transition density can be written as 
\begin{equation}
	\label{eq:app3:logtrans}
	\log{f_{\Delta}^{\theta}(x_k| x_{k-1})}= -\frac{d}{2}\log\left(2\pi\sigma^2\right)+\frac{d}{2} \log(b)-\frac{d}{2}\log\left(e^{2\Delta b}-1\right)-\frac{b\norm{x_{k}-e^{\Delta b} x_{k-1}}^2}{2\sigma^2(e^{2\Delta b}-1)},
\end{equation}
where $\norm{x}^2 = x^\intercal x = \sum_{i=1}^d x_i^2$ for $d$-dimensional vector $x$. To obtain the sufficient statistics in \eqref{eq:ex:rec:alpha} and \eqref{eq:ex:rec:beta}, if the photon location process is distributed according to the Airy or 2D Gaussian profiles, it suffices to take the gradient and Hessian of the log transition density in \eqref{eq:app3:logtrans} with respect to the diffusion $\sigma^2$ and drift $b$ coefficients, i.e.
\begin{align*}
	\nabla \log{f_{\Delta}^{\theta}(x_k| x_{k-1})} = \begin{pmatrix} \textsf{g}_1 \\
		\textsf{g}_2 \\
	\end{pmatrix},
	\qquad 
	\nabla^2 \log{f_{\Delta}^{\theta}(x_k| x_{k-1})} = \begin{pmatrix} \textsf{H}_{11} & \textsf{H}_{12} \\
		\textsf{H}_{21} & \textsf{H}_{22} \\
	\end{pmatrix},
\end{align*}
where
\begin{itemize}
	\item Gradient w.r.t $\sigma^2$
	\begin{align*}
		\textsf{g}_1 &:= 
		-\frac{d}{2\sigma^2}+\frac{b\norm{x_k-e^{\Delta b}x_{k-1}}^2}{2\sigma^4\left(e^{2\Delta b}-1\right)}.
	\end{align*}
	\item Gradient w.r.t $b$
	\begin{align*}
		\textsf{g}_2 &:= \frac{d}{2b} -\frac{d\Delta e^{2\Delta b}}{\left(e^{2\Delta b}-1\right)} -\frac{\norm{x_{k}-e^{\Delta b}x_{k-1}}^2}{2\sigma^2(e^{2\Delta b}-1)}\\
		&+\frac{\Delta be^{\Delta b}(x_k- e^{\Delta b}x_{k-1})^{\intercal}x_{k-1}}{\sigma^2(e^{2\Delta b}-1)}
		+\frac{\norm{x_k - e^{\Delta b}x_{k-1}}^2\Delta be^{2\Delta b}}
		{\sigma^2(e^{2\Delta b}-1)^2}.
	\end{align*}
	\item Hessian w.r.t $\sigma^2$ then $\sigma^2$
	\begin{align*}
		\textsf{H}_{11} &:=
		\frac{d}{2\sigma^4}-\frac{b\norm{x_k-e^{\Delta b}x_{k-1}}^2}{\sigma^6\left(e^{2\Delta b}-1\right)}.
	\end{align*}
	\item Hessian w.r.t $b$ then $\sigma^2$ and vice versa
	\begin{align*}
		\textsf{H}_{12} = \textsf{H}_{21} &:=  \frac{\norm{x_{k}-e^{\Delta b}x_{k-1}}^2}{2\sigma^4(e^{2\Delta b}-1)}\\
		&-\frac{\Delta be^{\Delta b}(x_k- e^{\Delta b}x_{k-1})^{\intercal}x_{k-1}}{\sigma^4(e^{2\Delta b}-1)}
		-\frac{\norm{x_k - e^{\Delta b}x_{k-1}}^2\Delta be^{2\Delta b}}
		{\sigma^4(e^{2\Delta b}-1)^2}.
	\end{align*}
	\item Hessian w.r.t $b$ then $b$
	\begin{align*}
		\textsf{H}_{22}&:= -\frac{d}{2b^2} -\frac{2d\Delta^2 e^{2\Delta b}}{e^{2\Delta b}-1}
		+\frac{2d\Delta^2 e^{4\Delta b}}{\left(e^{2\Delta b}-1\right)^2}\\ 
		&+\frac{x_{k}^{\intercal}x_k}{\sigma^2}\left[\frac{2\Delta e^{2\Delta b} + 2\Delta^2 be^{2\Delta b}}{(e^{2\Delta b}-1)^2}
		-\frac{4\Delta^2 be^{4\Delta b}}{(e^{2\Delta b}-1)^3}\right]\\
		&+\frac{x_k^{\intercal}x_{k-1}}{\sigma^2}
		\left[\frac{2\Delta e^{\Delta b} + \Delta^2 b e^{\Delta b}}{e^{2\Delta b}-1}
		-\frac{4\Delta e^{3\Delta b}+8\Delta^2 be^{3\Delta b}}{(e^{2\Delta b}-1)^2}
		+\frac{8\Delta^2 be^{5\Delta b}}{(e^{2\Delta b}-1)^3}\right]\\
		&-\frac{x_{k-1}^{\intercal}x_{k-1}}{\sigma^2}
		\left[\frac{2\Delta e^{2\Delta b} + 2\Delta^2 be^{2\Delta b}}{e^{2\Delta b}-1}
		-\frac{2\Delta e^{4\Delta b} + 6\Delta^2 be^{4\Delta b}}{(e^{2\Delta b}-1)^2} 
		+ \frac{4\Delta^2 be^{6\Delta b}}{(e^{2\Delta b}-1)^3}\right].
	\end{align*}
\end{itemize}
\section{Parameter estimation} \label{sec:smm:em}
Being able to estimate the biophysical parameters of the molecular interactions is very important in single-molecule tracking. In this section, we present two \textit{maximum likelihood} (ML) estimation methods that make use of smoothed additive functionals.
\subsection{By gradient ascent}
Given observations $y_{1:n}$ of size $n \in \mathbb{N}$, the marginal log-likelihood of the observations $p_\theta(y_{1:n})$ may be maximised via the \textit{steepest ascent algorithm} \cite{lemarechal2012cauchy, cauchy1847methode}:
\begin{equation}
	\label{eq:mle:sa}
	\theta_{i+1} = \theta_{i}+\gamma_{i+1}\mathcal{G}_n(\theta_i),
\end{equation}
where $\mathcal{G}_n(\theta_i) = \nabla\log p_{\theta}(y_{1:n})|_{\theta=\theta_i}$ is the score vector evaluated at the current estimate $\theta_i$, and the step-size sequence $\left\{\gamma_i\right\}_{i=1}^{\infty}$ consists of small positive numbers and satisfies $\sum_i\gamma_i=\infty$ and $\sum_i\gamma_i^2<\infty$; for example, take $\gamma_i = i^{-a}$ where $0.5<a<1$. 
One can also include the observed information matrix in order to follow the \textit{Newton-Raphson algorithm} described in \cite{nocedal2006numerical}. In this case, \eqref{eq:mle:sa} becomes
\begin{equation*}
	\theta_{i+1} = \theta_{i}-\gamma_{i+1}\mathcal{H}_n(\theta_i)^{-1}\mathcal{G}_n(\theta_i),
\end{equation*}
where $\mathcal{H}_n(\theta_i) = \nabla^2\log p_{\theta}(y_{1:n})|_{\theta=\theta_i}$ is the observed information matrix evaluated at the current estimate $\theta_i$.
\subsection{By expectation-maximization (EM)}
Another approach to obtaining maximum likelihood estimates of the hyperparameters $\theta$ is to use the \textit{expectation-maximization} (EM) algorithm by \cite{dempster1977maximum, wu1983convergence} defined as follows:
\begin{itemize}
	\item \textsl{Expectation} step: given the current parameter estimate $\theta_{i}$ and observations $y_{1:n}$,
	\begin{equation*}
		\mathcal{Q}(\theta,\theta_i) = \mathbb{E}_{\theta_i}\left[\log p_{\theta}(X_{1:n}, y_{1:n})|y_{1:n}\right],
	\end{equation*}
	where the joint density $p_{\theta}(x_{1:n}, y_{1:n})$ is defined in \eqref{eq:joint} and the expectation is with respect to the posterior $p_{\theta_i}(x_{1:n}|y_{1:n})$.
	\item \textsl{Maximisation} step:
	\begin{equation*}
		\theta_{i+1} = \argmax_{\theta \in \Theta} \mathcal{Q}(\theta, \theta_i).
	\end{equation*}
\end{itemize}
Recall that the \textsl{Expectation} step cannot be done exactly when using the Airy or Born and Wolf profile. In this case, the posterior expectation can be estimated using particle approximations of smoothed additive functionals. First of all, let $S_{k}^{\theta}(x_{1:k}):=\log p_{\theta}(x_{1:k}, y_{1:k})$ denote the additive functionals of interest at step $k$. Their corresponding sufficient statistics such that $S_{k}^{\theta}(x_{1:k}) = \sum_{k=1}^ns^{\theta}_{k}(x_{k-1}, x_k)$ are given by
\begin{equation*}
	s^{\theta}_{k}(x_{k-1}, x_k) := \log f_{\Delta}^{\theta}(x_{k}|x_{k-1}) + \log G_k^{\theta}(x_{k}),
\end{equation*}
where for notational simplicity, $f_{\Delta}^{\theta}(x_{1}|x_{0}) := \nu_{\theta}(x_{1})$.
In the \textsl{Maximisation} step,  define the function $\Lambda$ to obtain the maximising argument of $\mathcal{Q}(\theta,\theta_i)$, 
\begin{equation*}
	\theta_{i+1} = \Lambda\left(n^{-1}\mathbb{E}\left[S_n^{\theta_i}(X_{1:n})|y_{1:n}\right]\right),
\end{equation*}
An example of parameter estimation of the drift and diffusion coefficients based on \cref{ex:score} using EM is available in \cref{ex:smm:me}.

\begin{exmp} \label{ex:smm:me}	
	Building on \cref{ex:score}, note that given the model specification in \eqref{eq:ex:score}, it is impossible to compute the maximum of $\mathcal{Q}(\theta, \theta_i)$ for the parameter $b\neq0$ directly. However, as seen previously, the equation can also be written such that we simply have
	\begin{equation*}
		X_k = \varphi_{\theta} X_{k-1} + W_x, 
		\quad W_x \sim \mathcal{N}\left(0, r_{\theta}\mathbb{I}_{2\times2}\right),
	\end{equation*}	
	where the \textit{auxiliary parameters} are given by
	\begin{equation*}
		\varphi_{\theta} := e^{\Delta b} 
		\qquad\text{and}\qquad r_{\theta} := 
		\frac{\sigma^2}{b}\left(e^{2\Delta b}-1\right).
	\end{equation*}
	It is straightforward to maximise $\mathcal{Q}(\theta, \theta_i)$ for the auxiliary parameters $\varphi_{\theta}$ and $r_{\theta}$ as follows: let $\left\{S_{l, k}(x_{1:k})\right\}_{l=1}^3$ denote the additive functionals of interest at time $k$ and $\left\{s_{l, k}(x_{k-1}, x_k)\right\}_{l=1}^3$ their corresponding sufficient statistics. Luckily, the sufficient statistics are easily obtained, since for the Gaussian and Airy profiles, the likelihood $G_k$ does not depend on $\theta$:
	\begin{align*}
		s_{1, k}(x_{k-1}, x_k)= x_k^{\intercal}x_{k-1}, \quad s_{2, k}(x_{k-1}, x_k)= x_{k-1}^{\intercal}x_{k-1}, \quad
		s_{3, k}(x_{k-1}, x_k)= x_k^{\intercal}x_{k}. 
	\end{align*}
	The maximisation function is given by
	\begin{equation*}
		\Lambda(c_1, c_2, c_3) = \left(\frac{c_3}{2} - \frac{c_1^2}{2 c_2}, \frac{c_1}{c_2}\right).
	\end{equation*}
	Finally, to obtain maximum likelihood estimates for $b$ and $\sigma^2$, simply use the following transformation:
	\begin{equation*}
		b = \Delta^{-1}\log{\varphi_{\theta}}\quad \text{and } \quad
		\sigma^2 = \frac{r_{\theta} \log{\varphi_{\theta}}}{\Delta(\varphi_{\theta}^2-1)}.
	\end{equation*}
	Note that when dealing with measurements distributed according to the Born and Wolf model, we must also estimate the optical axis location parameter $z_0$, which is done via gradient ascent.	
	\begin{figure}[htbp!]
		\centering
		\includegraphics[width=\textwidth]{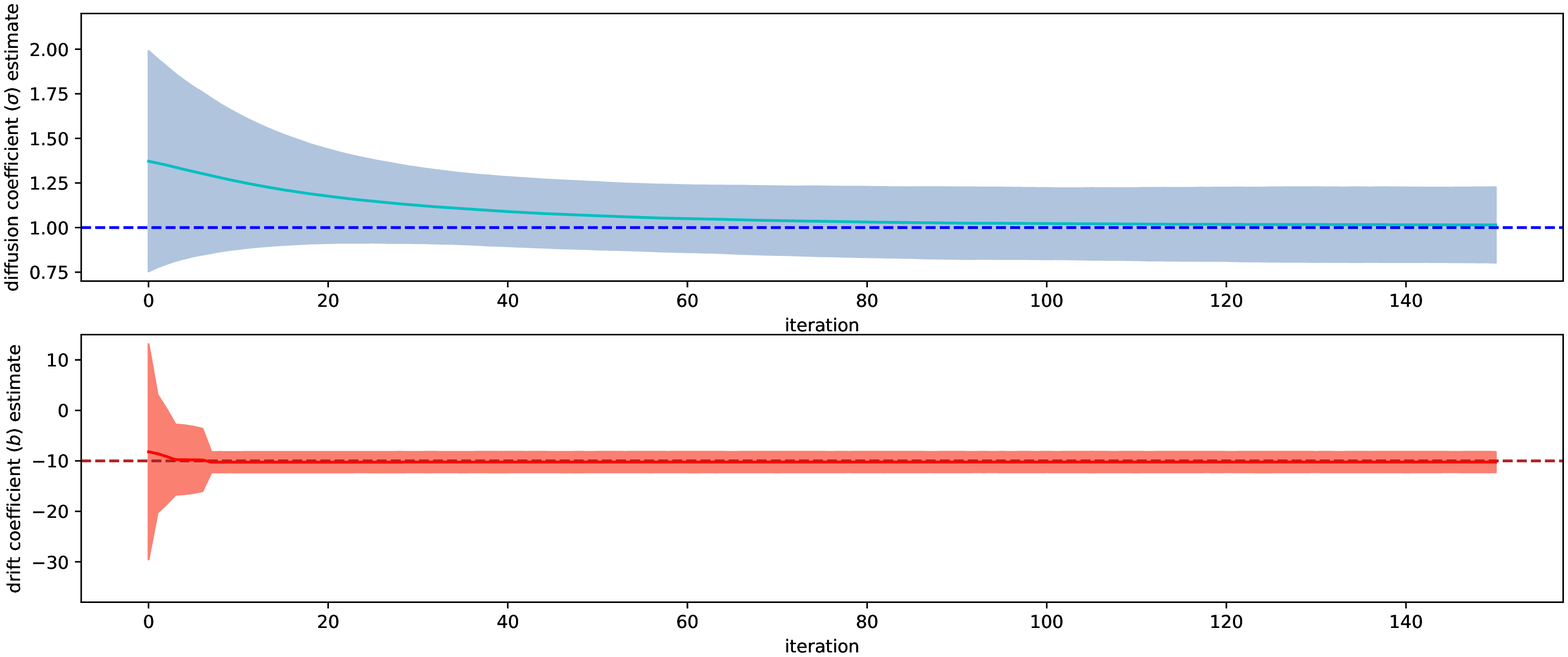} 
		\caption[Estimates of the diffusion ($\sigma^2$) and drift ($b$) coefficient over $150$ EM iterations or passes through the data.]{Estimates of the diffusion ($\sigma^2$) and drift ($b$) coefficient over $150$ EM iterations or passes through the data. The blue and red \textit{dashed lines} represent the true parameter values $\sigma^2=1$ $\mu$m$^2/$s and $b=-10$ s$^{-1}$, respectively. The red and blue \textit{solid lines} and \textit{bands} correspond to the mean estimates and their corresponding $95 \%$ confidence intervals over 50 datasets generated during the time interval $[0, 0.2]$ seconds, with initial location $x_0=(5.5, 5.5)^\intercal$ $\mu$m and a mean photon count of $1000$. The observations were generated according to the Airy profile with parameters as in \cref{ex:smm:obs} and the sufficient statistics were estimated using the PaRIS algorithm \cite{olsson2017efficient}.}
		\label{fig:smm:em}
	\end{figure}
\end{exmp}
\section{Score and OIM for a static molecule observed via the Airy profile} \label{app:smm:airy-static}
In \cref{ex:smm:fimest}, we consider the problem of estimating the FIM for the location parameters $(x_{1}, x_{2})$ of an in-focus static molecule. This is achieved by computing the score and OIM for the observed data. If the photon detection locations are described by the 2D Gaussian profile, the differentiation is straightforward, but in the case of the Airy profile \eqref{eq:airy}, the computations are more involved. \par 
Given observation $y\in\mathbb{R}^2$ and invertible lateral magnification matrix $M\in\mathbb{R}^{2\times2}$, for notational simplicity let $v:= M^{-1}y$, $r := \sqrt{(v_1-x_1)^2 + (v_2-x_2)^2}$ and $\alpha := \frac{2\pi n_{\alpha}}{\lambda_e}$. The log photon distribution profile \eqref{eq:image} is given by  
\begin{equation*}
	\log g(y) = -\log\left(|M|\right) + \log q(y),
\end{equation*}
where the image function is
$$q(y) = \frac{J_1^2(\alpha r)}{\pi r^2}.$$
First of all, use the relation $\frac{\partial}{\partial x} x^{-n}J_n(x) = -x^{-n}J_{n+1}(x)$ for $n \in \mathbb{N}$ in order to obtain the gradient and hessian of $q(y)$. Where the subscript $i$ appears, the result is valid for $i=1, 2$
\begin{align*}
	\frac{\partial q(y)}{\partial x_i}
	&= \frac{2\alpha}{\pi} \left(v_i-x_i\right) \frac{J_1(\alpha r)}{r} \frac{J_2(\alpha r)}{r^2},\\
	\frac{\partial^2 q(y)}{\partial x_i^2} 
	&= \frac{2\alpha^2}{\pi r^4} \left(v_i-x_i\right)^2 
	\left[J_1(\alpha r) J_3(\alpha r) 
	+ J_2^2(\alpha r)\right]
	-\frac{2\alpha}{\pi} \frac{J_1(\alpha r)}{r} \frac{J_2(\alpha r)}{r^2},\\
	\frac{\partial^2 q(y)}{\partial x_1\partial x_2} &= \frac{2\alpha^2}{\pi r^4} \left(v_1-x_1\right)\left(v_2-x_2\right)
	\left[J_1(\alpha r) J_3(\alpha r) 
	+ J_2^2(\alpha r)\right].
\end{align*}
To derive the components of the gradient and hessian of $\log q(y)$, we make use of the following identities:
\begin{align*}
	\nabla \log q(y) = \frac{\nabla q(y)}{q(y)} , \qquad 
	\nabla^2 \log q(y) = \frac{\nabla^2 q(y)}{q(y)} - \left[\nabla \log q(y)\right]^2.
\end{align*}
Therefore, for $i=1, 2$, the components of the log gradient are given by
\begin{align*}
	\frac{\partial \log q(y)}{\partial x_i}
	= \frac{2\alpha}{r} \frac{J_2(\alpha r)}{J_1(\alpha r)}\left(v_i-x_i\right),
\end{align*}
and the diagonal components of the log hessian are 
\begin{align*}
	\frac{\partial [\log q(y)]^2}{\partial x_i^2}
	= \frac{2\alpha^2}{r^2} \left(v_i-x_i\right)^2
	\left[\frac{J_3(\alpha r)}{J_1(\alpha r)} -
	\frac{J_2^2(\alpha r)}{J_1^2(\alpha r)}  \right]
	-\frac{2\alpha}{r} \frac{J_2(\alpha r)}{J_1(\alpha r)}.
\end{align*}
And finally, the cross terms are given by
\begin{align*}
	\frac{\partial [\log q(y)]^2}{\partial x_1 x_2}
	= \frac{2\alpha^2}{r^2} \left(v_1-x_1\right)\left(v_2-x_2\right)
	\left[\frac{J_3(\alpha r)}{J_1(\alpha r)} -
	\frac{J_2^2(\alpha r)}{J_1^2(\alpha r)}  \right].
\end{align*}
To summarise, the log gradient and negative log hessian for the Airy profile are
\begin{align*}
	\nabla \log g(y) = \gamma (M^{-1}y-x), &\qquad 
	\gamma =  \frac{2\alpha}{r} \frac{J_2(\alpha r)}{J_1(\alpha r)}, \\
	-\nabla^2 \log g(y) = \chi (M^{-1}y-x)(M^{-1}y-x)^{\intercal} + \gamma \mathbb{I}_{2\times2} , &\qquad  
	\chi =  -\frac{2\alpha^2}{r^2}
	\left[\frac{J_3(\alpha r)}{J_1(\alpha r)} -
	\frac{J_2^2(\alpha r)}{J_1^2(\alpha r)}\right].
\end{align*}
\end{appendices}

\end{document}